\documentclass[%
reprint,
nofootinbib,
amsmath,amssymb,
aps,
pra,
]{revtex4-2}
\pdfoutput=1
\usepackage[utf8]{inputenc}
\usepackage[english]{babel}
\usepackage[T1]{fontenc}
\usepackage{amsmath}
\usepackage{hyperref}

\usepackage{tikz}
\usepackage{lipsum}

\usepackage{dsfont}
\usepackage[normalem]{ulem}
\usepackage{mathtools}
\usepackage[makeroom]{cancel}
\usepackage{bbm}
\usepackage{enumerate}

\usepackage{xcolor}
\usepackage{colortbl}
\usepackage{makecell}

\usepackage{graphicx}
\usepackage{amsthm}

\definecolor{ms}{rgb}{0,.4,1}

\definecolor{colTab}{RGB}{75,170,129}
\definecolor{err}{rgb}{0.72, 0.07, 0.2}

\newcommand{\mb}[1]{\mathbb{#1}}

\newcommand{\Tr}[1]{\mathrm{Tr}\left[ #1\right]} 
\newcommand{\supp}{\mathrm{supp}}

\newcommand{\N}{\mb{N}}

\newcommand{\norm}[1]{\left\Vert #1 \right\Vert}

\newcommand{\ket}[1]{\left.\left|{#1}\right.\right\rangle}

\newcommand{\bra}[1]{\left.\left\langle{#1}\right.\right|}

\newcommand{\ketbra}[2]{\ket{#1} \!\! \bra{#2}}

\newcommand{\sandwich}[3]
{\left\langle  #1 \right| #2 \left| #3 \right\rangle}

\newcommand{\average}[2]{\left \langle #1 \right \rangle_{#2}}
\newcommand{\rom}[1]{\uppercase\expandafter{\romannumeral #1\relax}}
\newcommand{\norbra}[1]{\left( #1\right)}
\newcommand{\sqrbra}[1]{\left[ #1\right]}

\usepackage{amssymb}             

\newcommand{\de}{{\rm d}}

\newcommand{\idO}{\mathds{I}}

\newtheorem{theorem}{Theorem}

\newtheorem{lemma}{Lemma}

\usepackage{tcolorbox}
\tcbuselibrary{breakable}
\tcbuselibrary{skins}
\definecolor{my-green}{RGB}{0,144,81}
\definecolor{my-red}{RGB}{255,113,91}
\tcbset{enhanced jigsaw, colback=ms!10!white,colframe=ms!30!black,fonttitle=\bfseries}

\makeatletter 

\renewcommand\onecolumngrid{
	\do@columngrid{one}{\@ne}%
	\def\set@footnotewidth{\onecolumngrid}
	\def\footnoterule{\kern-6pt\hrule width 1.5in\kern6pt}%
}

\renewcommand\twocolumngrid{
	\def\footnoterule{
		\dimen@\skip\footins\divide\dimen@\thr@@
		\kern-\dimen@\hrule width.5in\kern\dimen@}
	\do@columngrid{mlt}{\tw@}
}%

\makeatother

\begin{document}
    
	\title{A Berry-Esseen Bound for Quantum Lattice Systems}

\author{\begingroup
\hypersetup{urlcolor=navyblue}
\href{}{Marcus Cramer
\endgroup}
}
 \affiliation{MCBH UG, Kronberg, Germany}

\author{\begingroup
\hypersetup{urlcolor=navyblue}
\href{}{Fernando G.S.L. Brand\~{a}o
\endgroup}
}
 \affiliation{Amazon Center for Quantum Computing, Pasadena, California 91106, USA}
 \affiliation{Institute for Quantum Information and Matter,
California Institute of Technology, Pasadena, California 91125, USA}

\author{\begingroup
\hypersetup{urlcolor=navyblue}
\href{}{M\u{a}d\u{a}lin Gu\c{t}\u{a}
\endgroup}
}
 \affiliation{School of Mathematical Sciences, University of Nottingham, Nottingham, NG7 2RD, United Kingdom}
  \affiliation{Centre for the Mathematics and Theoretical Physics of Quantum Non-equilibrium Systems, University of Nottingham,
Nottingham, NG7 2RD, United Kingdom}

\author{\begingroup
\hypersetup{urlcolor=navyblue}
\href{}{\'Alvaro M. Alhambra
\endgroup}
}
 \affiliation{Instituto de F\'{i}sica T\'{e}orica UAM/CSIC, C. Nicol\'{a}s Cabrera 13-15, Cantoblanco, 28049 Madrid, Spain}

	\author{\begingroup
\hypersetup{urlcolor=navyblue}
\href{}{Matteo Scandi }
\endgroup}
\affiliation{Instituto de F\'{i}sica T\'{e}orica UAM/CSIC, C. Nicol\'{a}s Cabrera 13-15, Cantoblanco, 28049 Madrid, Spain}

    

	
\begin{abstract}
It is expected that the statistical fluctuations of local observables in large quantum systems obey the central limit theorem, and approximate a normal distribution as their size grows. Here, we prove a version of the Berry-Esseen theorem for quantum lattice systems, which strengthens that central limit theorem by providing a rigorous convergence estimate towards the normal distribution for large but finite system size. Given a local quantum Hamiltonian on $N$ particles and a quantum state with a finite correlation
length, the result states that the measurement of local observables such as the energy follows a normal distribution,
up to an error scaling as $\mathcal{O}\left(N^{-\frac{1}{2}} \text{polylog}(N)\right)$, which is optimal up to logarithmic factors.
	\end{abstract}
	
	\maketitle

	\onecolumngrid
	
	\section{Introduction}
	
Quantum many-body systems are typically characterized via Hamiltonians with local interactions among the individual particles. 
This includes numerous spin models, which can describe a wealth of physical and computational scenarios of wide interest. However, the fact that their Hilbert space dimension grows exponentially with their size prevents us from having exact descriptions as soon as they reach just a few dozens of particles.  Still, what physically relevant models have in common is the locality of the interactions between the individual particles, so that quantum Hamiltonians are a very particular class of matrices within that Hilbert space. It is thus of great interest to understand which physical features can be inferred directly from locality alone. Prominent examples include the approximate light-cone structure of the dynamics through Lieb-Robinson bounds \cite{Lieb1972} or many properties at equilibrium such as bounds on their correlation structure \cite{Hastings_2006,Nachtergaele2006,Hastings_2007,Wolf_2008,Eisert_2010,Kliesch_2014}.

Another important series of features has to do with the statistical fluctuations of macroscopic observables, such as the energy or magnetization. Intuitively, locality implies that distant particles are roughly uncorrelated, and as such, approximately satisfy the assumptions behind the \emph{central limit theorem}. One can thus expect that the statistics of measuring local observables approach increasingly peaked normal distributions as the system size grows. This has been shown to be the case in the thermodynamic limit in previous works, under different sets of conditions (see \cite{Hartmann_2004, Erd_s_2014,Keating_2015} for some examples). However, in this context, it is also important to understand how these normal distributions emerge in large but finite systems, and what is the most general set of conditions under which they do. 

In statistics, for finite but large sample sizes of independent identically distributed (i.i.d.) random variables, the Berry-Esseen theorem \cite{berry1941accuracy,esseen1942liapounoff} provides a refinement of the central limit theorem, with an explicit rate of convergence towards the normal distribution as the sample size grows (see \cite{feller1991introduction,vershynin2026friendlyproofberryesseentheorem} for more pedagogical proofs). Specifically, the result states that cumulative distribution functions approach the corresponding error function in Kolmogorov distance as $\mathcal{O}(N^{-\frac{1}{2}})$. This has also been generalized to settings of weakly correlated classical variables \cite{tikhomirov1981convergence, sunklodas1984rate}. 

In this manuscript, we show that this theorem also applies to the statistical fluctuations of finite but large interacting many-body systems, even when the underlying state $\rho$ has fast-decaying spatial correlations. The result thus significantly strengthens previous works  \cite{Hartmann_2004, Erd_s_2014,Keating_2015} in two concrete ways: by considering the finite-size scaling with system size, and by considering states in arbitrary spatial dimension with decaying but non-zero correlations. Other previous related works largely focus on the concentration around the average of observable statistics in similar settings \cite{Anshu_2016,Kuwahara_2016,Kuwahara_2020,DePalma2022,Anshu_2023,Wild_2023}, rather than on the approach to Gaussianity. 

This manuscript is based on the previously unpublished results in~\cite{brandao2015berry}, and we here extend the proof to weaker notions of correlation decay, explicitly including algebraic decay. The proof closely follows classical works on the statistics of weakly correlated random variables~\cite{tikhomirov1981convergence, sunklodas1984rate}, since it also relies on showing that the corresponding characteristic function is close to a Gaussian function. It also involves intermediate lemmas controlling exponentials of non-commuting local operators, which might be of independent interest.

The result deals with universal features of the statistics of macroscopic quantum observables in a wide variety of physically relevant situations. As such, we expect it to have a range of implications in any context in which the observable statistics of interacting quantum systems are relevant, including statistical physics, many-body and complexity theory, classical and quantum algorithms, or quantum statistics and estimation theory \cite{hayashi2006quantumestimationquantumcentral,Gu__2007,Kahn2009}. 

Since the initial reporting of the main result here in \cite{brandao2015berry} (Theorem \ref{thm:BE}), the statement has already been used in technical proofs of several follow-up works on a variety of topics in quantum many-body and statistical physics. A (possibly non-exhaustive) summary of direct existing applications is the following:
\begin{itemize}
    \item It is one of key ingredients of the proof of the equivalence of the canonical and microcanonical ensembles in \cite{brandao2015equivalencestatisticalmechanicalensembles}, which holds under the only assumption of correlation decay in the canonical ensemble. See \cite{Tasaki_2018,Kuwahara_2020,KuwaharaETH} for later related results. These arguments were later extended in \cite{Bertoni_2025} to prove thermalization of low-entangled states under random energy smoothings.
    \item It can be used to upper bound the effective dimension of diagonal ensembles, thus bounding the average fluctuations of a quantum system around equilibrium and thermalization \cite{Farrelly_2017} (see also \cite{Hovhannisyan_2020} for subsequent applications in quantum thermodynamics).
    \item It is also a key technical ingredient of the feature that many-body dynamics with periodic revivals contain anomalous ``scarred" eigenstates in their spectrum \cite{ScarsProof}, since the Gaussianity implies that the initial state of the dynamics has well distributed support within those eigenstates (see \cite{Schecter_2019} for an example of a model for which the bound is tight).
    \item It is used in the proof of the  efficiency bounds of coarse-grained many-body thermometry \cite{Thermometry}, to compute the relevant Fisher information using the approximate Gaussianity of the distribution.
    \item The approximate Gaussianity of the wavefunction of product states is also leveraged in \cite{Rai_2024} to give rigorous bounds on the efficiency of energy filtering algorithms (both classical and quantum) for probing non-equilibrium dynamics \cite{LuQuantum,Schuckert2023,YilunClassical,Irmejs2024}.
    \item The Gaussianity of the density of states (when setting $\rho= \mathbb{I}/d$) is also used in \cite{Huang_2024} to characterize the entanglement of eigenstates of chaotic Hamiltonians.
\end{itemize}

The paper is structured as follows. In Sec. \ref{sec:prelim} we first explain the setting, main assumptions, and technical preliminaries. In Sec. \ref{sec:prod} we then prove the main result for the particular case of the states being product, for illustrative purposes. In Sec. \ref{sec:main}  we state and give the proof of the main result, except for that of the main technical Lemma, which we refer to Sec. \ref{app:lemmaDiff}. We conclude in Sec. \ref{sec:conclusion} with some open questions. In the Appendix we show the full proof of some of the more technical lemmas, including a proof of Esseen's original inequality and the results on series approximations to matrix exponentials.


	\section{Preliminaries}\label{sec:prelim}
	
	We consider $\mathcal{X}$ to be a lattice of $N$ sites, and let $\mathcal{H}:= \bigotimes_{i\in\mathcal{X}}\;\mathcal{H}_i$ be the associated Hilbert space. Assume $\mathcal{X}$ is equipped with a metric $d: \mathcal{X}\times\mathcal{X}\rightarrow \mathbb{N}$. This induces a canonical notion of semidistance between subsets $\mathcal{A},\, \mathcal{B}\subset\mathcal{X}$ given by:
	\begin{align}
		d(\mathcal{A},\mathcal{B}) := \min_{\substack{i\in\mathcal{A}\\j\in\mathcal{B}}}\,d(i,j)\,.
	\end{align}
	We also define the dimension $D$ of $\mathcal{X}$ to be the smallest $D\geq1$ such that there exists a constant $c_D>0$ satisfying:
	\begin{align}
		\left|\{i\in \mathcal{X}| d(i,j) = \ell\}\right|\leq c_D \, \ell^{D-1}\label{eq:dimDef}
	\end{align}
	for all $\ell$. 
	
	Then, let $\rho$ be a state acting on $\mathcal{H}$. We use the notation $\average{X}{\rho}:=\Tr{X\rho}$ to denote the averages with respect to $\rho$. We say that $\rho$ satisfies decay of correlations with correlation function $\alpha: \N^+\rightarrow\mathbb{R}^+$, if for all pairs of operators $A := \hat{A}\otimes\idO_{\mathcal{X}\setminus \mathcal{A}}$ and $B := \hat{B}\otimes\idO_{\mathcal{X}\setminus \mathcal{B}}$ it holds that:
	\begin{align}
		\left|\average{AB}{\rho}-\average{A}{\rho}\average{B}{\rho}\right|\leq\norm{A}\norm{B} \min\{|\mathcal{A}|,|\mathcal{B}|\}\,\alpha\norbra{d(\mathcal{A},\mathcal{B})}\,.\label{eq:decayOfCorrelations}
	\end{align}
    
This condition with $\alpha(l)$ exponentially decaying is known to hold in a variety of settings, including gapped ground states \cite{Hastings_2006,Nachtergaele2006}, matrix product and finitely correlated states \cite{Fannes1992,perezgarcia2007matrixproductstaterepresentations}, thermal states in 1D \cite{araki1969gibbs,Bluhm_2022,bergamaschi2026fastmixingquantumspin} and at high temperatures \cite{Kliesch_2014} and fixed points of rapidly mixing Lindbladians \cite{Kastoryano_2013,Roon_2024} among others. In most of these situations, it even holds without the pre-factor $\min\{|\mathcal{A}|,|\mathcal{B}|\}$. In long-range systems, the condition is also expected to hold if $\alpha(l)$ is a suitably fast-decaying polynomial \cite{kimura2025clustering,Kim_2025,mobusLongRange}. A similar polynomial decay is also expected to hold in critical systems at zero temperature.

In order to assess the effect of decaying terms, we also introduce the function:
	\begin{align}
		C_\alpha(\ell) := c_D\,\sum_{r=\ell}^\infty \, \alpha(r)\, {(r+1)^{D-1}}\,,\label{eq:decayOfCorrConst}
	\end{align}
	which we assume to be well defined for all $\ell\geq1$. We will consider the case of $\rho$ being a product state separately.
	
	We focus on the statistics of a Hamiltonian defined as:
	\begin{align}
		H := \sum_{j\in\mathcal{X}} \; h_j = \sum_{n} \; e_n \ketbra{n}{n}\,,
	\end{align}
	where $h_i$ are $R$-local (in the sense that they act trivially on sites that are at a distance larger than $R$), and bounded as $\|h_i\|\leq E$. Then, we define the cumulative distribution as:
	\begin{align}
		F_{H}(y) := \sum_{e_n\leq y} \; \sandwich{n}{\rho}{n}\,.
	\end{align}
	This corresponds to a probability distribution with average $\mu_H:=\average{H}{\rho}$ and variance $\sigma^2_H := \average{(H-\mu)^2}{\rho}$. A central assumption in the derivation is the following:
	\begin{align}
		\sigma_H^2 \geq c_0 E^2N\,,\label{eq:varianceScaling}
	\end{align}
	where $c_0$ is some arbitrary $\Theta(1)$ constant. 
    This assumption on the variance of the state is necessary for the normal distribution to appear. It guarantees that we are away from exceptional cases such as $\rho$ being an eigenstate of $H$ (with zero variance), or the fluctuations being confined to a small spatial region such that $\sigma_H =o(N)$. It is automatically satisfied in i.i.d. scenarios, it follows from translational invariance, and it can be efficiently checked given some standard classical description of $\rho$ (e.g. as a product state or an MPS) and $H$. Finally, let us also note that since the states we consider are short-range correlated we have that $\sigma^2_H = \mathcal{O}(N)$, so this is the fastest scaling possible too.
    
    It is also useful to introduce the re-normalized Hamiltonian $\widehat{H}$ given by:
	\begin{align}\label{eq:hhat}
		\widehat{H} := \frac{1}{\sigma_H}\,\sum_{j\in\mathcal{X}}\; \widehat{h}_j\,,
	\end{align}
	where $\widehat{h}_i:= h_i-\average{h_i}{\rho}$. This transformation allows us to have a standardized distribution with $\mu_{\widehat{H}}=0$ and $\sigma^2_{\widehat{H}}=1$. The cumulative distribution of $\widehat{H}$ is related to the one of the original Hamiltonian as:
	\begin{align}
		F_{\widehat{H}}(y) = F_H\norbra{\sigma_H\,y+\mu_H}\,.
	\end{align}
	
	Finally, let us introduce the cumulative distribution of a standard Gaussian $G(y)$ as:
	\begin{align}
		G(y) := \frac{1}{\sqrt{2\pi}}\int_{-\infty}^{y}\de x\; e^{-\frac{x^2}{2}}\,.\label{eq:GaussianCumulative}
	\end{align}
	This is what ultimately allows us to define a notion of distance between two cumulative distributions. In particular, let us define our figure of merit $\Delta$ to be:
	\begin{align}
		\Delta := \sup_y \, \left|F_{\widehat{H}}(y)-G(y)\right|\,.
	\end{align}
	
	An important Lemma we will use is given by Esseen's inequality~\cite{esseen1945fourier}:
	\begin{align}
		\Delta = \sup_y \,\left|F_{\widehat{H}}(y)-G(y)\right|\leq \frac{C}{\Omega} +\frac{2}{\pi}\int_{0}^{\Omega}\de \omega \,\,\frac{|\varphi_{\widehat{H}}({\omega})-e^{-\frac{\omega^2}{2}}|}{ {|\omega|}}\,,\label{eq:esseenInequality}
	\end{align}
	where $C$ is some numerical constant. 
	This allows us to estimate the convergence of $F_{\widehat{H}}(y)$ to $G(y)$  as the size of the lattice $N:=|\mathcal{X}|$ grows, and the main results of the paper will be stated as upper bounds on $\Delta$ in terms of $N$. We refer to App.~\ref{app:esseenInequality} for a precise statement and a proof of this Lemma. 
	
\section{Product states}\label{sec:prod}

  Before moving to the main result, we cover a narrower case with a significantly simpler proof for illustrative purposes: that in which the correlations are exactly zero and the state is product.
	
	\begin{theorem}\label{thm:BEprod}
		Let $\rho=\bigotimes_i^{ N }  \rho_i $ be a product state among all sites, and $H$ be a $R-$local Hamiltonian as above. Then
		\begin{equation}
			\Delta \le  \frac{K}{N^{\frac{1}{2}}},
		\end{equation}
		where $K= \mathcal{O}(1)$.
	\end{theorem}
	\begin{proof}
		The proof, previously sketched in \cite{Rai_2024}, is simplified by using the results from \cite{Wild_2023} bounding the cumulant series via the cluster expansion. 
		Let $\varphi_{\widehat{H}}({\omega}) := \langle e^{i\omega \widehat{H}}\rangle_{\rho} $ with $\widehat{H}$ as in Eq. \eqref{eq:hhat}, and let  $\omega^* =\frac{1}{2  e^2 R^{2D}(R^{2D}+1)E}=\mathcal{O}(1)$. It follows from \cite[Theorem 10]{Wild_2023} that, for $\omega \le \frac{\omega^* \sigma_H}{2}$, 
		\begin{align}
			\left \vert \log \varphi_{\widehat{H}}({\omega})  + \frac{\omega^2}{2}\right \vert \le  \frac{2N}{\sigma_H^3} \frac{\omega^3}{(\omega^*)^3}\leq \tilde{B}_1\,\frac{\omega^3}{N^{\frac{1}{2}}}\,,\label{eq:20}
		\end{align}
		where we implicitly used the bound $\sigma_H^2 \ge c_0 E^2 N$ and we defined $\tilde{B_1} := (c_0^{\frac{1}{2}} \omega^*E)^{-3}$.  Eq.~\eqref{eq:20} follows from the fact that the Taylor series of $ \log \varphi_{\widehat{H}}({\omega})$ has a convergence radius given by $\omega^*$, as it can be by directly bounding the coefficients of the series (i.e. the cumulants of the distribution). Applying $\vert e^z -1 \vert \le \vert z \vert e^{\vert z \vert}$ to Eq.~\eqref{eq:20}, and assuming that $\omega \le\min \{ \frac{\omega^* \sigma_H}{2} ,\frac{\sqrt{N}}{4 \tilde{B}_1} \}$, we have
		\begin{align}
			\left \vert \varphi_{\widehat{H}}({\omega}) - e^{-\frac{\omega^2}{2}}  \right \vert &=e^{-\frac{\omega^2}{2}} \left \vert e^{-(\log\varphi_{\widehat{H}}({\omega})+\frac{\omega^2}{2})} - 1 \right \vert\le \tilde{B}_1\,\frac{\omega^3}{N^{\frac{1}{2}}}\, e^{-\frac{\omega^2}{2}(1-(2\tilde{B}_1 N^{-1/2})\,\omega)} \le \tilde{B}_1\,\frac{\omega^3}{N^{\frac{1}{2}}}\, e^{-\frac{\omega^2}{4}}\,.
		\end{align}
		
		We can now combine this with Esseen's inequality Eq. \ref{eq:esseenInequality} as follows (see Appendix \ref{app:esseenInequality} or \cite{Kolassa1997}). Let $\Omega \le \min \{ \frac{\omega^* \sigma_H}{2} ,\frac{\sqrt{N}}{4 \tilde{B}_1} \}$,
		\begin{align}
			\left|F_{\widehat{H}}(y)-G(y)\right|&\leq \frac{C}{\Omega}+\frac{2}{\pi}\int_{0}^{\Omega}\de \omega \,\,\frac{|\varphi_{\widehat{H}}({\omega})-e^{-\frac{\omega^2}{2}}|}{ {|\omega|}} \le  \frac{C}{\Omega}+ \frac{2\tilde{B}_1}{\pi \,N^{\frac{1}{2}}} \int_{0}^{\Omega} \vert \omega \vert ^2 e^{-\frac{\omega^2}{4}}\de \omega \le  \frac{C}{\Omega}+ \frac{4\tilde{B}_1}{ N ^{\frac{1}{2}}  }\,.
		\end{align}
		Then, choosing  $\Omega = C' \sqrt{N} $ with   $ C'= \min \{ \frac{\omega^* \sigma_H}{2\sqrt{N}} ,\frac{1}{4 \tilde{B}_1} \}$, we obtain $\Delta \le \frac{C+4\tilde{B}_1}{C' \sqrt{N}}$.
		
	\end{proof}
    Here, the decay is as the inverse of a square root, without polylogarithmic factors, which matches (up to constant factors) the scaling of the standard case of classical i.i.d. variables.
	This result for product states also extends to Hamiltonians with polynomially decaying interactions \cite{Sanchez_Segovia_2025}, for which the cumulant series converges in the same way as in Eq.~\eqref{eq:20}.  

	\section{Main results}\label{sec:main}
    
    Our main technical result is an upper bound on $\Delta$ for the cumulative distribution function of states under the weaker assumption of correlation decay. It is stated as follows:
	\begin{theorem}\label{thm:BE}
		Let $\rho$ be a state on a $D$-dimensional lattice $\mathcal{X}$ of size $|\mathcal{X}|=N$ and $D\geq 1$, satisfying decay of correlations of the form in Eq.~\eqref{eq:decayOfCorrelations}
		with function $\alpha(\ell)$, such that $C_\alpha(\ell)$ is well-defined for all $\ell\geq1$. Let $H$ be the sum of $R$-local terms $h_i$ satisfying $\|h_i\|\leq E$, and such that $\sigma_H^2 \geq c_0 E^2N$. 
		
		Then, the distance between the cumulative distribution $F_{\widehat{H}}(y)$ and the Gaussian one satisfies:
		\begin{enumerate}
			\item let $\alpha(\ell):= L_0e^{-\ell/\xi}$, and $\xi>0$. For all $N> \max\{e^2, \,4\xi\,\frac{(\log N)^2}{\log 2}\}$ we have:\label{it:expDecay}
			\begin{align}
				\Delta\leq  f_1(N)\,\frac{(\log N)^{2D}}{\sqrt{N}} \,,\label{eq:dEst1}
			\end{align}
			for some explicit function $f_1(N)$ (see Eq.~\eqref{eq:delExp}) which can be upper-bounded by a constant. 
			\item let $\alpha(\ell):=L_0\ell^{-(D+\beta)}$, and $\beta>D+1$. For any $\varepsilon\in (0, \frac{\beta-D}{2(\beta +3D-1)})$ and $N> \max\{\lfloor 2^{\beta+3D-1}\rfloor,\,\frac{4RN^{\frac{1}{\beta+3D-1}+\frac{\varepsilon}{2D}} (\log N)}{\log 2} \}$ we have:\label{it:algDecay}
			\begin{align}
				\Delta\leq  f_2(N)\,\frac{1}{N^{\frac{1}{2}(\beta-D)/(\beta+3D-1)-\varepsilon}}\,,\label{eq:dEst2}
			\end{align}
			for some explicit function $f_2(N)$ (see Eq.~\eqref{eq:delExp2}) which can be upper-bounded by a constant. Moreover, in the limit $\varepsilon\rightarrow0$ we obtain:
			\begin{align}
				\Delta\leq  f_2(N)\,\frac{\log N}{N^{\frac{1}{2}(\beta-D)/(\beta+3D-1)}}\,.\label{eq:dEst3}
			\end{align}
			\item let $\alpha(\ell):=L_0\ell^{-(D+\beta)}$, and $\beta>D$. Assume the decay rate to satisfy the condition stronger than Eq. \eqref{eq:decayOfCorrelations}:
			\begin{align}
				\left|\average{AB}{\rho}-\average{A}{\rho}\average{B}{\rho}\right|\leq\norm{A}\norm{B} \alpha\norbra{d(\mathcal{A},\mathcal{B})}\,.\label{eq:decayOfCorrelations2}
			\end{align}
			Then, for any $\varepsilon\in (0, \frac{\beta-D}{2(\beta +3D)})$ and $N> \max\{\lfloor 2^{\beta+3D}\rfloor,\,\frac{4RN^{\frac{1}{\beta+3D}+\frac{\varepsilon}{2D}} (\log N)}{\log 2} \}$ we have:\label{it:algDecay2}
			\begin{align}
				\Delta\leq  \tilde{f}_2(N)\,\frac{1}{N^{\frac{1}{2}(\beta-D)/(\beta+3D)-\varepsilon}}\,,\label{eq:dEst22}
			\end{align}
			for some explicit function $\tilde{f}_2(N)$ (see Eq.~\eqref{eq:delExp3}) which can be upper-bounded by a constant. Moreover, in the limit $\varepsilon\rightarrow0$ we obtain:
			\begin{align}
				\Delta\leq  \tilde{f}_2(N)\,\frac{\log N}{N^{\frac{1}{2}(\beta-D)/(\beta+3D)}}\,.\label{eq:dEst32}
			\end{align}
		\end{enumerate}
	\end{theorem}
	Notably, the scaling in Eq.~\eqref{eq:dEst22} matches the one obtained in~\cite{tikhomirov1981convergence} for the one-dimensional classical case. The proof does not apply straightforwardly to the case where $H$ is long-range in the same way as Thm \ref{thm:BE} as shown in \cite{Sanchez_Segovia_2025}. However, we expect that generalization to also hold, with the proof likely requiring a more refined technical analysis of the effect of the long-ranged tails.

	The proof of Thm~\ref{thm:BE} is much more involved than that of Thm~\ref{thm:BEprod}. It is based on the following three lemmas that build on each other. The first considers a differential equation that the characteristic function must obey. It closely resembles that of a Gaussian function, with additional error terms that the result bounds. 
	
	\begin{lemma}\label{lemma:diffEq}
		Let $\varphi_{\widehat{H}}({\omega}) := \langle e^{i\omega \widehat{H}}\rangle_{\rho} $. This  satisfies the differential equation:
		\begin{align}
			\begin{cases}
				\varphi_{\widehat{H}}'({\omega}) =(-\omega + \eta(\omega))\varphi_{\widehat{H}}(\omega) + \nu(\omega)\\
				\varphi_{\widehat{H}}(0) =1
			\end{cases}\label{eq:diffEqPhi}\,,
		\end{align}
		for some explicit functions $\eta(\omega)$ and $\nu(\omega)$. 
		
		Let $K,\ell,M\in \mathbb{N}^+$ be arbitrary constants satisfying $K,\ell>1$, $2R \ell K \leq N$, and $\ell-M\geq1$. Let $\omega$ satisfy $\omega\in\sqrbra{0,\min\{\Omega_1,\,\Omega_2\}}$ (see Table~\ref{tab:constants} for the definition of these constants). Then, it holds that $|\eta(\omega)|\leq c_1 \omega + c_2 \omega^2$ and $|\nu(\omega)|\leq c_3 +c_4\omega + c_5 \omega^2$, with:
		\begin{align}
			c_1 = \norbra{\frac{(2R)^{D}}{c_0}}C_\alpha(2R(\ell-1))\,;
			\qquad\qquad\;\;&\qquad\qquad
			c_2 = \frac{B_4 }{c_0^{3/2}}\norbra{\frac{\ell^{2D}}{ \sqrt{N}}}\,;\label{eq:etaConstants}
			\\
			c_3 =\frac{8R}{\sqrt{c_0}}\norbra{\sqrt{N}\,\ell \alpha(2R(\ell - M-1))} \,;
			\qquad\qquad&\qquad\quad\;\;\;\;
			c_4 =\frac{B_2}{c_0}\norbra{\frac{\,\ell^DK^{D-1}}{2^{K-1}}}\,;\nonumber
			\\
			c_5 =\frac{ 1}{c_0^{\frac{3}{2}}}\Bigg(B_5\norbra{\frac{\ell^{2D+\frac{D}{2}}}{N}}+B_6&\norbra{\frac{1}{2^{M}\ell^{\frac{D}{2}(M-3)}\sqrt{N}}}\Bigg)\,,\label{eq:nuConstants}
		\end{align}
		where $B_2,\,B_4,\,B_5$ and $B_6$ are numerical constants, whose explicit expression is given in Table~\ref{tab:constants}. Moreover, assuming the decay rate in Eq.~\eqref{eq:decayOfCorrelations2}, we can replace $c_3$ with $\tilde{c}_3 := c_3/(4R\ell)$.
	\end{lemma}
	Notice that if $\nu(\omega)=\eta(\omega)=0$ then $\varphi_{\widehat{H}}({\omega})=e^{-\frac{\omega^2}{2}}$ is exactly Gaussian.
	The second lemma bounds the effect of the error terms in the differential equation, and how the characteristic function resembles a Gaussian for small $\omega$.
	
	\begin{lemma}\label{lemma:boundPhi} 
		In the same language of Lemma~\ref{lemma:diffEq}, 
		assume $\ell$ big enough so that $c_1<\frac{1}{2}$, and let $\omega\in[0,\min\{\Omega_1,\,\Omega_2,\,\Omega_3\}]$, where $\Omega_3:=\frac{c_0 \sigma_H}{4B_4E \ell^{2D}}$. Then, it holds that:
		\begin{align}
			|\varphi_{\widehat{H}}({\omega}) - e^{-\frac{\omega^2}{2}}|&\leq \,e^{-\frac{\omega^2}{6}}\omega^2\norbra{\frac{c_1}{2} +\frac{c_2}{3}\omega}+4\norbra{c_4 + c_5 \omega}\norbra{1-e^{-\frac{\omega^2}{4}}}+c_3 \min\bigg\{\frac{8}{ \omega},\,\omega\bigg\}\,,\label{eq:diffPhiGau}
		\end{align}
		where $c_1,\,c_2,\,c_3,\,c_4,\, c_5$ are the constants defined in Lemma~\ref{lemma:diffEq}.
	\end{lemma}

   Lemmas \ref{lemma:diffEq} and \ref{lemma:boundPhi} together imply that the characteristic function, which can be identified with a Loschmidt echo \cite{Gorin2006,Wisniacki_2012} is close to a Gaussian for short times, which might be of independent interest \cite{Izrailev_2006,TorresGaussian}. 
	
	Finally, the last lemma shows how the error term in the Gaussian approximation of the characteristic function bounds the distance of the cumulative distribution from the normal one, in the form of an upper bound in $\Delta$.
	
	\begin{lemma}\label{lemma:DeltaEst1} 
		In the same language of Lemma~\ref{lemma:diffEq}, and $\Omega\in\sqrbra{0,\min\{\Omega_1,\,\Omega_2,\,\Omega_3\}}$:
		\begin{align}
			\Delta \leq \frac{C}{\Omega}  + \frac{6 \,c_1}{2\pi}+\frac{\sqrt{6}\, c_2}{\sqrt{\pi}}+\frac{8}{\pi}(c_3+ c_4)\norbra{1+\,\log(\Omega)}+\frac{8\,c_5\,\Omega}{\pi}\,,
		\end{align}
		where $c_1,\,c_2,\,c_3,\,c_4,\, c_5$ are the constants defined in Lemma~\ref{lemma:diffEq}. 
	\end{lemma}
	
	In the rest of the section, we first prove Thm~\ref{thm:BE}, and then we move on to prove Lemma~\ref{lemma:boundPhi} and~\ref{lemma:DeltaEst1}. The proof of Lemma~\ref{lemma:diffEq} is much more technical and, for this reason, we defer it to Sec.~\ref{app:lemmaDiff}.
	
	\setlength{\tabcolsep}{0pt}
	\renewcommand{\arraystretch}{1.6} 
	\begin{table}
		\centering
		\begin{tabular}{|c!{\vrule width 1.2pt}c|}
			\hline
			\rowcolor{colTab!25}
			\;\;\;Constant\;\;\;&\;\; Definition\;\;  \\
			\Xhline{1.2pt}
			$\Gamma$  &  $\qquad\qquad\qquad\max\{(4c_D (2R)^D ),  (2c_D^2(2R)^{D})\}\qquad\qquad\qquad$ \\
			\hline
			$s_D$  & $\sum_{m=0}^{\infty}\; (m+2)^{D}2^{-m}$  \\
			\hline
			$B_1 $  & $(c_D(2R)^{D}+\sqrt{2}\,\Gamma)^2$  \\
			\hline
			$B_2 $  & $2\,(c_D(2R)^D + \Gamma)$  \\
			\hline
			$B_3$  & $2\,(c_D (2R)^{D}+2^{\frac{D-1}{2}}\Gamma)$  \\
			\hline
			$B_4 $ & $(B_1+B_3^2(1+2s_{2(D-1)}))$  \\
			\hline
			$B_5$  & $(12B_2^2R^{\frac{D}{2}}\sqrt{(1+C_\alpha(1))}\,s_{3(D-1)})$  \\
			\hline
			$B_6$ & \qquad\qquad$(\Gamma^{2} s_{D-1}) $ (in general) \textbf{or} $0$ (for commuting models)  \qquad\qquad\qquad  \\
			\Xhline{1.2pt}
			$\Omega_1 $ & $\sigma_H\big(2E\,\Gamma \ell^{\frac{D}{2}} K^{\frac{D-1}{2}}\big)^{\!-1}$ \\
			\hline
			$\Omega_2$ &$\sigma_H(2B_2\,E\ell^D K^{D-1})^{-1}$   \\
			\hline
			$\Omega_3$ &$c_0 \sigma_H\norbra{B_4E \ell^{2D}}^{-1}$ \\
			\hline
		\end{tabular}
		\caption{Numerical constants used during the derivation.}
		\label{tab:constants}
	\end{table}
	
	\begin{proof}[Proof of Thm.~\ref{thm:BE}.\ref{it:expDecay}] 
		If $\alpha(\ell) = L_0e^{-\frac{\ell}{\xi}}$, $C_\alpha(\ell)$ satisfies $C_\alpha (\ell) \leq c_DL_0\,\xi  (\ell+1)^{D-1}e^{-\frac{\ell}{\xi}}$.
		
		Let us choose:
		\begin{align}
			\begin{cases}
				\ell=\frac{\xi}{R}\log N+1\\
				M = \frac{\ell-1}{2} =\frac{\xi}{2R}\log N\\
				K = \log_2 N
			\end{cases}\,,
		\end{align}
		where we assume $N> \lfloor e^{2R/\xi}\rfloor$, so that $M>1$, and that $N\geq 4\xi\,\frac{(\log N)^2}{\log 2}$, so that $2R \ell K \leq N$ is satisfied, and we can apply Lemma~\ref{lemma:diffEq}--\ref{lemma:DeltaEst1}. With this choice, we have: 
		\begin{align}
			\begin{cases}
				\Omega_1 \geq \norbra{\frac{\sqrt{c_0}\,(\log 2)^{\frac{D-1}{2}}R^{\frac{D}{2}}}{2 \Gamma(2\xi)^{\frac{D}{2}}}}\frac{\sqrt{N}}{(\log N)^{D-\frac{1}{2}}}\, 
				\\
				\Omega_2 \geq \norbra{\frac{\sqrt{c_0}\,(\log 2)^{D-1}R^D}{2B_2 (2\xi)^D}}\frac{\sqrt{N}}{(\log N)^{2D-1}} \,
				\\
				\Omega_3\geq\norbra{\frac{c_0\sqrt{c_0}R^{2D}}{4B_4(2\xi)^{2D}}}\frac{\sqrt{N}}{(\log N)^{2D}} 
			\end{cases}\,.
		\end{align}
		Let $B_7 := \min\left\{\norbra{\frac{\sqrt{c_0}\,(\log 2)^{\frac{D-1}{2}}R^{\frac{D}{2}}}{2 \Gamma(2\xi)^{\frac{D}{2}}}},\,\norbra{\frac{\sqrt{c_0}\,(\log 2)^{D-1}R^D}{2B_2 (2\xi)^D}},\,\norbra{\frac{c_0\sqrt{c_0}R^{2D}}{4B_4(2\xi)^{2D}}}\right\}$, and let:
		\begin{align}
			\Omega := B_7\,\frac{\sqrt{N}}{(\log N)^{2D}}\,,
		\end{align}
		so that ${\Omega\in\sqrbra{0,\min\{\Omega_1,\,\Omega_2,\,\Omega_3\}}}$, and we can apply Lemma~\ref{lemma:DeltaEst1}. Then, we obtain:
		\begin{align}
				\Delta &\leq \frac{(\log N)^{2D}}{\sqrt{N}}\Bigg(\frac{C}{B_7} +  \frac{6c_DL_0\xi^{D}}{(2\pi c_0R^{D-1})(\log N)^{D+1}N^{\frac{3}{2}}} + \frac{\sqrt{6}B_4 (2\xi)^{2D}}{\sqrt{\pi}c_0^{3/2}R^{2D}}+\nonumber
			\\
			& \qquad\qquad+ \frac{8}{\pi}\norbra{\frac{16\xi L_0}{\sqrt{c_0}\,(\log N)^{2D-2}}+\norbra{\frac{2B_2(2\xi)^D}{c_0R^D(\log 2)^{D-1}\,\sqrt{N}}}}\norbra{\frac{(1+\log(B_7))}{(\log N)}+\norbra{\frac{1}{2}- 2D\,\frac{\log(\log N)}{\log N}}} +\nonumber 
			\\
			& \;\;\;\;\qquad\qquad\qquad\qquad\qquad\qquad\qquad\qquad\qquad+\frac{8 B_7B_5 (2\xi)^{\frac{5D}{2}}}{\pi c_0^{\frac{3}{2}}R^{\frac{5D}{2}}(\log N)^{\frac{3D}{2}}}+\frac{8 B_7 B_6  }{\pi c_0^{\frac{3}{2}}N^{(\frac{\xi(\log 2)}{2R}-\frac{1}{2})}\log N^{\frac{D}{2}(\frac{\xi}{2R}\log N+5)}}\Bigg)\,.\label{eq:delExp}
		\end{align}
		It should be noticed that for the last term in the parenthesis to decay, we need $\frac{\xi(\log 2)}{2R}-\frac{1}{2}\geq 0$, that is $\xi\geq \frac{R}{(\log 2)}$. This restriction arises from the factor $2^M$ in the denominator of $c_5$ in Eq.~\eqref{eq:nuConstants}, so it is likely a proof artifact. In fact, for any constant $x\in \mathbb{R}^+$ satisfying $x\geq 2$, we can define $c_5$ so that it has $x^M$ in the denominator instead of $2^M$ (of course, this cannot be done at will, as this procedure reduces the range of $\Omega$ by a factor of $2/x$). For this reason, for any finite $x$, we can prove Eq.~\eqref{eq:dEst1} in the range $\xi\geq \frac{R}{(\log x)}$. As such, the claim holds for $\xi>0$. 
	\end{proof}
	\begin{proof}[Proof of Thm.~\ref{thm:BE}.\ref{it:algDecay}] If $\alpha(\ell):=L_0\ell^{-(D+\beta)}$, then $C_\alpha(\ell)\leq (2\,c_DL_0 \ell^{-\beta})/\beta$.
		In this case, let us choose: 
		\begin{align}
			\begin{cases}
				\ell=N^{\delta}+1\\
				M = \frac{\ell-1}{2} =\frac{1}{2}\,N^\delta\\
				K = \log_2 N
			\end{cases}\,,
		\end{align}
		for some $\delta$ to be specified. Moreover, we assume $N> \lfloor 2^{1/\delta}\rfloor$, so that $M>1$, and that $N\geq \frac{4RN^\delta (\log N)}{\log 2}$, so that $2R \ell K \leq N$. 
		With this choice, we have: 
		\begin{align}
			\begin{cases}
				\Omega_1 \geq \norbra{\frac{\sqrt{c_0}(\log 2)^{\frac{D-1}{2}}}{2\Gamma }}\frac{\sqrt{N}}{N^{\frac{\delta D}{2}} (\log N)^{\frac{D-1}{2}}}\, 
				\\
				\Omega_2 \geq \norbra{\frac{\sqrt{c_0}\,(\log 2)^{D-1}}{2B_2 }}\frac{\sqrt{N}}{N^{\delta D}(\log N)^{D-1}} \,
				\\
				\Omega_3\geq\norbra{\frac{c_0\sqrt{c_0}}{4B_4}}\frac{\sqrt{N}}{N^{2\delta D}} 
			\end{cases}\,.
		\end{align}
		Define $\tilde{B}_7 := \min\left\{\norbra{\frac{\sqrt{c_0}(\log 2)^{\frac{D-1}{2}}}{2\Gamma }},\,\norbra{\frac{\sqrt{c_0}\,(\log 2)^{D-1}}{2B_2 }},\,\norbra{\frac{c_0\sqrt{c_0}}{4B_4}}\right\}$. Then, choosing:
		\begin{align}
			\Omega := \tilde{B}_7 \,N^{\frac{1}{2}-2\delta D}
		\end{align}
		we can apply Lemma~\ref{lemma:DeltaEst1}. This gives:
		\begin{align}
			\Delta &\leq\frac{1}{N^{\frac{1}{2}-2\delta D}}\Bigg(\frac{C}{\tilde{B}_7}  +\frac{1}{N^{(\beta +2D)\delta-\frac{1}{2}}} \norbra{\frac{6c_DL_0(2R)^{D-1}}{2\beta\pi c_0\,(2R)^{\beta} }}+\norbra{\frac{\sqrt{6}B_42^{2D}}{ \sqrt{\pi}\,c_0^{3/2}}}+\nonumber
			\\
			&\qquad\qquad\qquad\quad+\frac{1}{N^{(\beta +3D-1)\delta-1}}\norbra{\frac{8(1+\log(\tilde{B}_7))}{\pi}+\frac{8}{\pi}\,\norbra{\frac{1}{2}-2\delta D}\log N}\norbra{\frac{16RL_0\,}{\sqrt{c_0}	\,R^{\beta+D}}}+\nonumber
			\\
			&\qquad\qquad\qquad\qquad\qquad\qquad\qquad+ \frac{(\log_2 N)^{D-1}}{N^{\frac{1}{2}+\delta D}}\norbra{\frac{8(1+\log(\tilde{B}_7))}{\pi}+\frac{8}{\pi}\,\norbra{\frac{1}{2}-2\delta D}\log N}\norbra{\frac{2B_2\,2^D }{c_0 }}\nonumber
			\\
			&\qquad\qquad\quad\qquad\qquad\quad\qquad\qquad\quad\quad\quad\quad\quad\;\;+\frac{1}{N^{\frac{3D\delta}{2}}}\norbra{\frac{8\,B_5\tilde{B}_72^{\frac{5D}{2}}}{\pi c_0^{\frac{3}{2}}}}+\norbra{\frac{8B_6 \tilde{B}_7}{c_0^{\frac{3}{2}}2^{N^\delta/2}N^{\frac{1}{2}+\delta D(4 +\frac{1}{2}(N^\delta/2-3)}}}\Bigg)\,.\label{eq:delExp2}
		\end{align}
		At this point, we need to choose $\delta$ so that $\Delta$ decreases with $N$. Then, the first condition we need to impose is:
		\begin{align}
			\frac{1}{2}-2\delta D >0 \qquad\implies\qquad \delta<\frac{1}{4D}\,.\label{eq:delCond1}
		\end{align}
		This takes care of the factor in front of the parenthesis in Eq \eqref{eq:delExp2}. In order for the function inside of the parenthesis not to diverge with $N$, we also need to impose the following two inequalities:
		\begin{align}
			(\beta+2D)\delta -\frac{1}{2}&\geq 0\implies \delta>\frac{1}{2(\beta+2D)}\,;\qquad\qquad
			(\beta+3D-1)\delta -1> 0\implies \delta >\frac{1}{(\beta+3D-1)}\,.\label{eq:31}
		\end{align}
		The latter equation implies the first, so we can just consider this case.  Together with Eq.~\eqref{eq:delCond1}, we see that $\delta\in(\frac{1}{\beta +3D-1},\,\frac{1}{4D})$. In order for this range not to be empty, we need $\beta>D+1$. Finally, choosing $\delta=\frac{1}{\beta +3D-1}+\frac{\varepsilon}{2D}$, we obtain the claim. Note that we need to require $\varepsilon\in (0, \frac{\beta-D-1}{2(\beta +3D-1)})$ for $\delta$ to be in the correct range. For $\delta=\frac{1}{\beta +3D}$ the term in the second line of Eq.~\eqref{eq:delExp2} diverges logarithmically, which gives the second part of the claim.
	\end{proof}
	\begin{proof}[Proof of Thm.~\ref{thm:BE}.\ref{it:algDecay2}]
		In the same framework as in the proof of Thm.~\ref{thm:BE}.\ref{it:algDecay}, we obtain:
		\begin{align}
			\Delta &\leq\frac{1}{N^{\frac{1}{2}-2\delta D}}\Bigg(\frac{C}{\tilde{B}_7}  +\frac{1}{N^{(\beta +2D)\delta-\frac{1}{2}}} \norbra{\frac{6c_DL_0(2R)^{D-1}}{2\beta\pi c_0\,(2R)^{\beta} }}+\norbra{\frac{\sqrt{6}B_42^{2D}}{ \sqrt{\pi}\,c_0^{3/2}}}+\nonumber
			\\
			&\qquad\qquad\qquad\quad+\frac{1}{N^{(\beta +3D)\delta-1}}\norbra{\frac{8(1+\log(\tilde{B}_7))}{\pi}+\frac{8}{\pi}\,\norbra{\frac{1}{2}-2\delta D}\log N}\norbra{\frac{4L_0\,}{\sqrt{c_0}	\,R^{\beta+D}}}+\nonumber
			\\
			&\qquad\qquad\qquad\qquad\qquad\qquad\qquad+ \frac{(\log_2 N)^{D-1}}{N^{\frac{1}{2}+\delta D}}\norbra{\frac{8(1+\log(\tilde{B}_7))}{\pi}+\frac{8}{\pi}\,\norbra{\frac{1}{2}-2\delta D}\log N}\norbra{\frac{2B_2\,2^D }{c_0 }}\nonumber
			\\
			&\qquad\qquad\quad\qquad\qquad\quad\qquad\qquad\quad\quad\quad\quad\quad\;\;+\frac{1}{N^{\frac{3D\delta}{2}}}\norbra{\frac{8\,B_5\tilde{B}_72^{\frac{5D}{2}}}{\pi c_0^{\frac{3}{2}}}}+\norbra{\frac{8B_6 \tilde{B}_7}{c_0^{\frac{3}{2}}2^{N^\delta/2}N^{\frac{1}{2}+\delta D(4 +\frac{1}{2}(N^\delta/2-3)}}}\Bigg)\,.\label{eq:delExp3}
		\end{align}
		where the only difference with respect to Eq.~\eqref{eq:delExp2} is in the second line, since with the decay function defined in Eq.~\eqref{eq:decayOfCorrelations2} we can use the constant in $\tilde{c}_3$ instead of $c_3$. Then, this means that Eq.~\eqref{eq:31} becomes in this context:
		\begin{align}
			(\beta+2D)\delta -\frac{1}{2}&\geq 0\implies \delta>\frac{1}{2(\beta+2D)}\,;\qquad\qquad
			(\beta+3D-1)\delta -1> 0\implies \delta >\frac{1}{(\beta+3D-1)}\,,
		\end{align}
		which corresponds to restricting $\delta$ to the interval $\delta\in(\frac{1}{\beta +3D},\,\frac{1}{4D})$. In order for this range not to be empty, we need $\beta>D$. Then, choosing $\delta=\frac{1}{\beta +3D}+\frac{\varepsilon}{2D}$, gives the result.
	\end{proof}

	\begin{proof}[Proof of Lemma~\ref{lemma:boundPhi}] 
		We begin by noticing that  Eq.~\eqref{eq:diffEqPhi} can be formally integrated to give:
		\begin{align}
			\varphi_{\widehat{H}}({\omega}) = e^{-\frac{\omega^2}{2} + \int_0^{\omega}\de \tilde{\omega}\;\eta(\tilde{\omega})}\norbra{1+ \int_{0}^{\omega}\de \omega_1\;e^{\frac{\omega_1^2}{2} - \int_0^{\omega_1}\de \tilde{\omega}\;\eta(\tilde{\omega})}\nu(\omega_1)}\,.
		\end{align}
		This allows for the estimate: 
		\begin{align}
			|\varphi_{\widehat{H}}({\omega}) - e^{-\frac{\omega^2}{2}}|&\leq e^{-\frac{\omega^2}{2}}\left|e^{ \int_0^{\omega}\de \tilde{\omega}\;\eta(\tilde{\omega})}-1\right| + e^{-\frac{\omega^2}{2}}\left|\int_{0}^{\omega}\de \omega_1\,e^{\frac{\omega_1^2}{2} + \int_{\omega_1}^\omega\de \tilde{\omega}\;\eta(\tilde{\omega})}\nu(\omega_1)\right|\leq
			\\
			&\leq e^{-\frac{\omega^2}{2}}\left|\int_{0}^\omega\de \omega_1\;e^{ \int_0^{\omega_1}\de \tilde{\omega}\;\eta(\tilde{\omega})}\,\eta(\omega_1)\right| +  e^{-\frac{\omega^2}{2}}\int_{0}^{\omega}\de \omega_1\,e^{\frac{\omega_1^2}{2} + \int_{\omega_1}^\omega\de \tilde{\omega}\;|\eta(\tilde{\omega})|}\left|\nu(\omega_1)\right|\,,\label{eq:integroDiff}
		\end{align}
		where in the second line we derived and integrated again in the first term, and we brought the absolute values inside the integral in the second term. At this point, we can use Lemma~\ref{lemma:diffEq} to bound Eq.~\eqref{eq:integroDiff}, together with the assumptions $c_1<\frac{1}{2} $ and $\omega\in[0,\Omega_3]$. In particular, it should be noticed that the $\Omega_3$ is chosen so that $\omega\leq \Omega_3\leq \frac{1-c_1}{2c_2}$. 
		
		First, it should be noticed that:
		\begin{align}
			\left|\int_0^{\omega_1}\de \tilde{\omega}\;\eta(\tilde{\omega})\right| &\leq \int_0^{\omega_1}\de \tilde{\omega} \norbra{c_1\tilde{\omega}+c_2\tilde{\omega}^2} =\omega_1^2\norbra{\frac{c_1}{2} +\frac{c_2}{3}\omega_1}\leq \omega_1^2\norbra{\frac{1}{6}+\frac{c_1}{3}} \leq \frac{\omega_1^2}{3}\label{eq:26}
		\end{align}
		where the two last inequalities follow from applying $\omega_1\leq \frac{1-c_1}{2c_2}$ first, and then $c_1<\frac{1}{2}$. Then, applying this inequality to the first term in Eq.~\eqref{eq:integroDiff}, we obtain:
		\begin{align}
			e^{-\frac{\omega^2}{2}}\left|\int_{0}^\omega\de \hat{\omega}\;e^{ \int_0^{\hat\omega}\de \tilde{\omega}\;\eta(\tilde{\omega})}\,\eta(\tilde{\omega})\right| &\leq e^{-\frac{\omega^2}{2}}e^{\frac{\omega^2}{3}}\int_0^{\omega}\de \tilde{\omega}\;\left|\eta(\tilde{\omega})\right| \leq e^{-\frac{\omega^2}{6}}\omega^2\norbra{\frac{c_1}{2} +\frac{c_2}{3}\omega}
		\end{align}
		where in the last step we used the fact that $\omega\in[0,\frac{1-c_1}{2c_2}]$.
		
		The second term in Eq.~\eqref{eq:integroDiff} needs additional steps, due to the exponential inside of the integral. First, let use Eq.~\eqref{eq:26} to obtain:
		\begin{align}
			e^{-\frac{\omega^2}{2}}\int_{0}^{\omega}\de \omega_1\,e^{\frac{\omega_1^2}{2} + \int_{\omega_1}^\omega\de \tilde{\omega}\;|\eta(\tilde{\omega})|}\left|\nu(\omega_1)\right| 
			&\leq e^{-(1-c_1)\frac{\omega^2}{4} -((1-c_1)\frac{\omega^2}{4}- \omega^3\frac{c_2}{3})}\int_{0}^{\omega}\de \omega_1\,e^{(1-c_1)\frac{\omega_1^2}{4} +((1-c_1)\frac{\omega_1^2}{4}- \omega_1^3\frac{c_2}{3})}\left|\nu(\omega_1)\right| \,,
		\end{align}
		where we have carried out the integration and then isolated the term $((1-c_1)\frac{\omega_1^2}{4}- \omega_1^3\frac{c_2}{3})$. By taking a derivative, it can be verified that this function is monotonically increasing in $\omega_1$ for $\omega_1\leq \frac{1-c_1}{2c_2}$. For this reason, we can upper-bound the second term in Eq.~\eqref{eq:integroDiff} as
		\begin{align}
			e^{-\frac{\omega^2}{2}}\int_{0}^{\omega}\de \omega_1\,e^{\frac{\omega_1^2}{2} + \int_{\omega_1}^\omega\de \tilde{\omega}\;|\eta(\tilde{\omega})|}\left|\nu(\omega_1)\right| 
			&\leq e^{-(1-c_1)\frac{\omega^2}{4}}\int_{0}^{\omega}\de \omega_1\,e^{(1-c_1)\frac{\omega_1^2}{4} } (c_3 +c_4\omega + c_5 \omega^2) \,,
		\end{align}
		where we used Lemma~\ref{lemma:diffEq} to bound $\nu(\omega)$. At this point, let us first focus on the linear and of the quadratic terms. These can be estimated as:
		\begin{align}
			e^{-(1-c_1)\frac{\omega^2}{4}}\int_{0}^{\omega}\de \omega_1\,e^{(1-c_1)\frac{\omega_1^2}{4}}\omega_1&\norbra{c_4 + c_5 \omega_1}=e^{-(1-c_1)\frac{\omega^2}{4}}\int_{0}^{\omega}\de \omega_1\,\frac{2\norbra{c_4 + c_5 \omega_1}}{1-c_1}\norbra{\frac{\de}{\de \tilde{\omega}}\norbra{e^{(1-c_1)\frac{\omega^2}{4}}}\bigg|_{\tilde{\omega}=\omega_1}}\leq\label{eq:inter30}
			\\
			&\leq e^{-(1-c_1)\frac{\omega^2}{4}}\frac{2\norbra{c_4 + c_5 \omega}}{1-c_1}\norbra{e^{(1-c_1)\frac{{\omega}^2}{4}}-1}\leq 4\norbra{c_4 + c_5 \omega}\norbra{1-e^{-\frac{\omega^2}{4}}} \label{eq:le36}.
		\end{align}
		where we implicitly applied Hölder's inequality to Eq.~\eqref{eq:inter30}, carried out the integral, and finally used the condition $c_1<\frac{1}{2}$. It remains to bound the linear term. First, one can give the trivial bound:
		\begin{align}
			c_3\,e^{-(1-c_1)\frac{\omega^2}{4}}\int_{0}^{\omega}\de \omega_1\,e^{(1-c_1)\frac{\omega_1^2}{4}} \leq c_3\,\int_{0}^{\omega}\de \omega_1\, = c_3\, \omega\,,\label{eq:31b}
		\end{align}
		where we used Hölder's inequality. Additionally, for $\omega\geq  \sqrt{4/(1-c_1)}$, it also hold that:
		\begin{align}
			c_3\,e^{-(1-c_1)\frac{\omega^2}{4}}\int_{0}^{\omega}\de \omega_1\,e^{(1-c_1)\frac{\omega_1^2}{4}} \leq\frac{c_3\, 4}{(1-c_1)\omega}\leq\frac{c_3\, 8}{\omega}\,,\label{eq:32b}
		\end{align}
		This can be shown as follows. First notice that for $\omega^*= \sqrt{4/(1-c_1)}$, it holds that $\omega^* = (4/(1-c_1))/\omega^*$. Then, thanks to Eq.~\eqref{eq:31b}, it holds that:
		\begin{align}
			e^{-(1-c_1)\frac{(\omega^*)^2}{4}}\norbra{\frac{4 e^{(1-c_1)\frac{(\omega^*)^2}{4}}}{(1-c_1)\omega^*} -\int_{0}^{\omega^*}\de \omega_1\,e^{(1-c_1)\frac{\omega_1^2}{4}}} = e^{-(1-c_1)\frac{(\omega^*)^2}{4}}\norbra{\omega^*e^{(1-c_1)\frac{(\omega^*)^2}{4}} -\int_{0}^{\omega^*}\de e^{(1-c_1)\frac{\omega_1^2}{4}}}   \geq 0\,.
		\end{align}
		Moreover, if we take the derivative of the first parenthesis, we obtain:
		\begin{align}
			\frac{\de}{\de \omega}\norbra{\frac{4 e^{(1-c_1)\frac{(\omega)^2}{4}}}{(1-c_1)\omega} -\int_{0}^{\omega}\de \omega_1\,e^{(1-c_1)\frac{\omega_1^2}{4}}} = \norbra{2-\frac{4 }{(1-c_1)\omega^2}-1} e^{(1-c_1)\frac{(\omega)^2}{4}}\,,
		\end{align}
		which is strictly positive for $\omega > \sqrt{4/(1-c_1)}$. Then, taking the minimum between Eq.~\eqref{eq:31b} and Eq.~\eqref{eq:32b} gives the estimate in the claim.
	\end{proof}
	
	\begin{proof}[Proof of Lemma~\ref{lemma:DeltaEst1}]
		We begin by proving that $\varphi_{\widehat{H}}({\omega}) := \langle e^{i\omega \widehat{H}}\rangle_{\rho} $ coincides with the Fourier transform of $F_{\widehat{H}}(y)$. Indeed, an explicit computation gives:
		\begin{align}
			\int_{-\infty}^{\infty}\de F_{\widehat{H}}(y)\; e^{i\omega y} = \int_{-\infty}^{\infty}\de y\norbra{\sum_{\norbra{\frac{e_n-\mu_H}{\sigma_H}}=y} \; \sandwich{n}{\rho}{n}\,e^{i\omega y}} = \sum_{n} \; \sandwich{n}{\rho}{n}\,e^{i\omega \norbra{\frac{e_n-\mu_H}{\sigma_H}}} = \average{e^{i\omega \widehat{H}}}{\rho} = \varphi_{\widehat{H}}({\omega})\,.
		\end{align}
		
		Now, assume $\Omega\in\sqrbra{0,\min\{\Omega_1,\,\Omega_2,\,\Omega_3\}}$. Then, we can apply Lemma~\ref{lemma:boundPhi} to Eq.~\eqref{eq:esseenInequality} to obtain:
		\begin{align}
			&\Delta\leq\frac{C}{\Omega} +\frac{2}{\pi}\int_{0}^{\Omega}\de \omega \,\norbra{e^{-\frac{\omega^2}{6}}\omega\norbra{\frac{c_1}{2} +\frac{c_2}{3}\omega}+\frac{4\norbra{c_4 + c_5 \omega}\norbra{1-e^{-\frac{\omega^2}{4}}}+c_3\min\left\{\frac{8}{ \omega},\,\omega\right\}}{\omega}}\leq
			\\
			&\leq \frac{C}{\Omega} +\frac{2}{\pi}\int_{0}^{\infty}\de \omega \,e^{-\frac{\omega^2}{6}}\norbra{\frac{c_1\,\omega}{2}+\frac{c_2\,\omega^2}{3}}+\frac{2}{\pi}\int_{0}^{2}\de \omega\norbra{c_3 +4(1-e^{-\frac{\omega^2}{4}})\frac{ c_4}{\omega}}+\frac{8}{\pi}\int_{2}^{\Omega}\de \omega\norbra{ \frac{2\,c_3}{\omega^2}+\frac{c_4}{\omega}}+\frac{8 \,c_5\,\Omega}{\pi}\leq
			\\
			&\leq \frac{C}{\Omega}  + \frac{6 \,c_1}{2\pi}+\frac{\sqrt{6}\, c_2}{\sqrt{\pi}}+(c_3+ c_4)\norbra{\frac{8}{\pi}+\frac{8}{\pi}\,\log(\Omega)}+\frac{8 \,c_5\,\Omega}{\pi}\,,
		\end{align}
		where in the second line we extended the limit of integration to infinity, we used the fact that $\min\left\{\frac{8}{ \omega},\,\omega\right\}$ is smaller than either functions, and finally that $(1-e^{-\frac{\omega^2}{4}})/\omega\leq 1/2$. This proves the claim.
	\end{proof}

	\section{Proof of Lemma~\ref{lemma:diffEq}}\label{app:lemmaDiff}
	
	Our goal is to rewrite $\varphi_{\widehat{H}}'({\omega}) $ in the differential form presented in Eq.~\eqref{eq:diffEqPhi} where $\eta(\omega)$ and $\nu(\omega)$ can be bounded using decay of correlation in the state $\rho$, and the assumption that $\omega$ is in the range $\omega \in [0,\min\{\Omega_1,\,\Omega_2\}]$.
	
	In order to do so, let us rewrite $\varphi_{\widehat{H}}'({\omega}) $ explicitly as:
	\begin{align}
		\varphi_{\widehat{H}}'({\omega})  = i\average{\widehat{H}\,e^{i\omega \widehat{H}}}{\rho} =  \frac{i}{\sigma_H}\,\sum_{j\in\mathcal{X}}\; \average{\widehat{h}_j\,e^{i\omega \widehat{H}}}{\rho}\label{eq:appB1}
	\end{align}
	In order to get an intuition about the proof method, consider the simpler case of $\widehat{H}$ being a commuting model, so that $\forall i,\,j$ $[\widehat{h}_i,\widehat{h}_j]=0$), with $\widehat{h}_i$ having support on a single site, and $\rho$ having no correlations (that is $\average{A_iA_j}{\rho}=\average{A_i}{\rho}\average{A_j}{\rho}$ for any two $A_i$ and $A_j$ with different support). In this case, Eq.~\eqref{eq:appB1} can be rewritten as: 
	\begin{align}
		\varphi_{\widehat{H}}'({\omega})  &=\frac{i}{\sigma_H}\,\sum_{j\in\mathcal{X}}\norbra{\average{\widehat{h}_j\,e^{i\omega \widehat{h}_j/\sigma_H}}{\rho}\,\prod_{k\neq j} \average{e^{i\omega \widehat{h}_k/\sigma_H}}{\rho}}=\label{eq:48b}
		\\
		& =\frac{i}{\sigma_H}\,\sum_{j\in\mathcal{X}}\average{\widehat{h}_j\,e^{i\omega \widehat{h}_j/\sigma_H}}{\rho}\varphi_{\widehat{H}}({\omega})  + \norbra{\frac{i}{\sigma_H}\,\sum_{j\in\mathcal{X}}\norbra{\average{\widehat{h}_j\,e^{i\omega \widehat{h}_j/\sigma_H}}{\rho}\, \average{\norbra{e^{-i\omega \widehat{h}_j/\sigma_H}-\idO}e^{i\omega \widehat{H}}}{\rho}}} = \label{eq:49b}
		\\
		&=-\omega\,\varphi_{\widehat{H}}({\omega})  +\norbra{\frac{i}{\sigma_H}\,\sum_{j\in\mathcal{X}}\average{\widehat{h}_j\,\norbra{e^{i\omega \widehat{h}_j/\sigma_H}-\frac{i\omega\widehat{h}_j}{\sigma_H} }}{\rho}}\varphi_{\widehat{H}}({\omega}) +\nu(\omega) =\label{eq:41}
		\\
		&= (-\omega +\eta(\omega))\varphi_{\widehat{H}}(\omega) + \nu(\omega)\,,\label{eq:51b}
	\end{align}
	where we implicitly defined $\nu(\omega)$ to be the term inside of the parenthesis in Eq.~\eqref{eq:49b}, and $\eta(\omega)$ to be the term in Eq.~\eqref{eq:41}. In this way, we obtain an expression of the form of  Eq.~\eqref{eq:diffEqPhi}. The procedure just presented is the simplified version of what will be carried out in this section, and it works in three steps: first, we exploit $\widehat{h}_j$ as an anchor point to leverage the locality of $\rho$ and rewrite the series as a sum of local contributions (in the previous simpler case, this was automatic due to the commutativity and the fact that $\widehat{h}_j$ were supported on single sites); then, we add and subtract $\varphi_{\widehat{H}}({\omega}) $, by exploiting the fact that $\prod_{k\neq j} e^{i\omega \widehat{h}_k/\sigma_H} = e^{-i\omega \widehat{h}_j/\sigma_H}e^{i\omega \widehat{H}}$; and finally we use $\sigma_{\widehat{H}}^2  = \sum_j\langle\widehat{h}_j^2\rangle_{\rho}/\sigma_H^2=1$ to sum and subtract $\omega\varphi_{\widehat{H}}({\omega})$ and get the desired form of the differential equation. 
	
	Before moving on to the general case, let us introduce the notation (see Fig. \ref{fig:plot1} for an illustration): 
	\begin{align}
		z_j^\ell(m) := \frac{1}{\sigma_H}\,\sum_{\substack{k\in\mathcal{X}
				\\
				d(k,j)> 2 R \ell m}}\; \widehat{h}_k\,;\qquad\qquad\qquad\qquad
		\widehat{H}_j^\ell(m) = z_j^\ell(m-1) -z_j^\ell(m)\label{eq:defZ}
	\end{align}
	where we remind the reader $R$ denotes the locality of $\widehat{H}$ and $\ell$ is some arbitrary integer $\ell>1$. In this way $\widehat{H}_j^\ell(m)$ contains only Hamiltonian terms $\widehat{h}_k$ such that $d(k,j)\in[2R\ell (m-1),2R\ell m]$. That is, we have a decomposition of the full Hamiltonian as $\widehat{H} = \sum_{m=1}^\infty\,\widehat{H}_j^\ell(m) =  \sum_{m=1}^{K}\,\widehat{H}_j^\ell(m)+z_j^\ell(K)$, where each term corresponds to a layer of width $2R\ell$, and with minimum distance $2R\ell (m-1)$ from the anchor point $j$. 
	
	\begin{figure}
		\centering
		\includegraphics[width=0.3\linewidth]{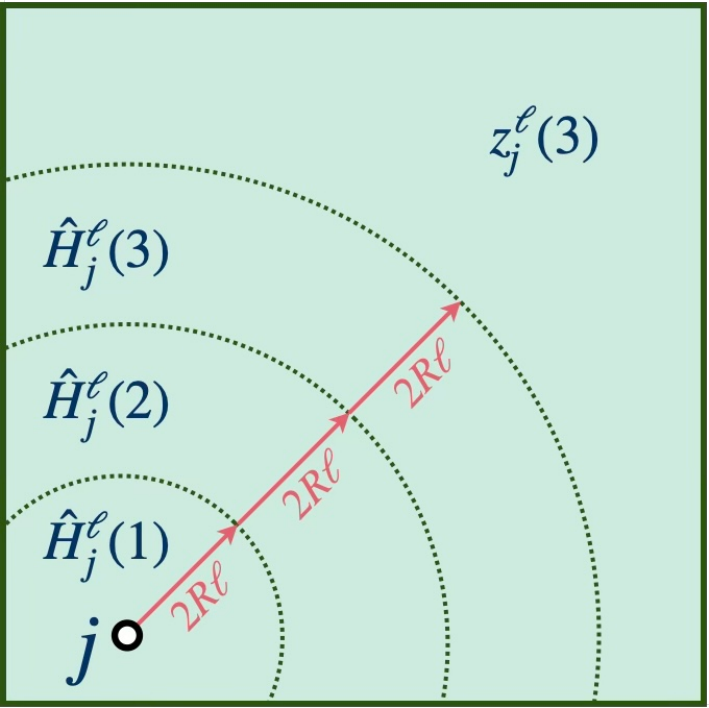}
		\caption{Representation of the operators defined in Eq.~\eqref{eq:defZ}. Each $\widehat{H}_j^\ell(m)$ corresponds to the terms in the Hamiltonian that are in the $m$-th layer (of width $2R\ell$) starting from $j$. On the other hand, $z_j^\ell(m)$ corresponds to the remaining Hamiltonian when $m$ layers have been discarded.}
		\label{fig:plot1}
	\end{figure}
	
	Then, we can rewrite Eq.~\eqref{eq:appB1} as:
	\begin{align}
		\varphi_{\widehat{H}}'({\omega}) = \frac{i}{\sigma_H}\,\sum_{j\in\mathcal{X}}\;\bigg( \average{\widehat{h}_j\,e^{i\omega z_j^\ell(1) }}{\rho}+&\average{\widehat{h}_j\norbra{e^{i\omega (z_j^\ell(0)-z_j^\ell(1))}{\tilde R}_{1,j}(\omega)-\idO}\,e^{i\omega z_j^\ell(1) }}{\rho}+\nonumber
		\\
		&\qquad\qquad+\average{\widehat{h}_j\norbra{e^{i\omega z_j^\ell(0)}-e^{i\omega (z_j^\ell(0)-z_j^\ell(1))}{\tilde R}_{1,j}(\omega)e^{i\omega z_j^\ell(1) }}}{\rho}\bigg) \,,\label{eq:appB3}
	\end{align}
	where we used the fact that $z_j^\ell(0) = \widehat{H}$, and we introduced an additional operator ${\tilde R}_{1,j}(\omega)$, which for the moment we assume to be generic, with the only extra assumption that ${\tilde R}_{1,j}(\omega)= \idO$ for commuting models. In this way, the exponent of the first term contains only terms that act at a distance larger that $2R\ell$ from the origin (which allows us to exploit decay of correlation), while the other two are correction terms arising from this peeling procedure, and from the non-commutativity of the Hamiltonian respectively (in particular, the last term is zero for a commuting Hamiltonian). This procedure is what allows us to exploit the decay of correlations of $\rho$ to decompose the sum in its local contributions. 
	
	Still, the error that appears at this order is too big to give a useful approximation. For this reason, we iterate this procedure further. Before doing so, though, let us introduce some notation:
	\begin{align}
		\begin{cases}
			\xi_j^0(\omega) := \idO\,\\
			\xi_j^m(\omega) := \norbra{e^{i\omega (z_j^\ell(m-1)-z_j^\ell(m))}\tilde{R}_{m,j}(\omega)-\idO} = \norbra{e^{i\omega \widehat{H}_j^\ell(m)}\tilde{R}_{m,j}(\omega)-\idO}\,\\
			\Xi_j^m(\omega) := \norbra{e^{i\omega z_j^\ell(m-1)}-e^{i\omega (z_j^\ell(m-1)-z_j^\ell(m))}\tilde{R}_{m,j}(\omega)e^{i\omega z_j^\ell(m) }}\,
		\end{cases}\,,\label{eq:appc4}
	\end{align}
	where, once again, $\tilde{R}_{m,j}(\omega)$ are for the moment generic functions such that $\tilde{R}_{m,j}(\omega)=\idO$ for commuting models. This allows us to rewrite $e^{i\omega z_j^\ell (m-1)}$ as:
	\begin{align}
		e^{i\omega z_j^\ell (m-1)} = \norbra{\idO+\xi_j^{m}(\omega)}e^{i\omega z_j^\ell (m)}+\Xi_j^m(\omega)\,.
	\end{align}
	Then, Eq.~\eqref{eq:appB3} can be rewritten in this notation and further expanded as:
	\begin{align}
		\varphi_{\widehat{H}}'({\omega}) &= \frac{i}{\sigma_H}\,\sum_{j\in\mathcal{X}}\;\norbra{\average{\widehat{h}_j\, \xi_j^0(\omega)e^{i\omega z_j^\ell(1)}}{\rho}+\average{\widehat{h}_j\, \xi_j^0(\omega)\xi_j^1(\omega)e^{i\omega z_j^\ell(1)}}{\rho}+\average{\widehat{h}_j\, \xi_j^0(\omega)\Xi_j^1(\omega)}{\rho}} = 
		\\
		&= \frac{i}{\sigma_H}\,\sum_{j\in\mathcal{X}}\;\bigg(\norbra{\average{\widehat{h}_j\, \xi_j^0(\omega)e^{i\omega z_j^\ell(1)}}{\rho}+\average{\widehat{h}_j\, \xi_j^0(\omega)\xi_j^1(\omega)e^{i\omega z_j^\ell(2)}}{\rho}}+\average{\widehat{h}_j\, \xi_j^0(\omega)\xi_j^1(\omega)\xi_j^2(\omega)e^{i\omega z_j^\ell(2)}}{\rho}+\nonumber
		\\
		&\qquad\qquad\qquad\qquad\qquad\qquad\qquad\qquad\qquad\qquad\quad+\norbra{\average{\widehat{h}_j\, \xi_j^0(\omega)\Xi_j^1(\omega)}{\rho}+\average{\widehat{h}_j\, \xi_j^0(\omega)\xi_j^1(\omega)\Xi_j^2(\omega)}{\rho}}\bigg)=
		\\
		&=\frac{i}{\sigma_H}\,\sum_{j\in\mathcal{X}}\bigg(\norbra{\sum_{m=0}^{K-1}\;\average{\widehat{h}_j\, \xi_j^0(\omega)\dots\xi_j^m(\omega)e^{i\omega z_j^\ell(m+1)}}{\rho}}+\average{\widehat{h}_j\, \xi_j^0(\omega)\dots\xi_j^K(\omega)e^{i\omega z_j^\ell(K)}}{\rho}+\nonumber
		\\
		&\qquad\qquad\qquad\qquad\qquad\qquad\qquad\qquad\qquad\qquad\qquad\qquad+\norbra{\sum_{m=1}^{K-1}\;\average{\widehat{h}_j\, \xi_j^0(\omega)\dots\xi_j^{m-1}(\omega)\Xi_j^{m}(\omega)}{\rho}}\bigg)\,, \label{eq:appB8}
	\end{align}
	where $K$ is some arbitrary integer satisfying $2R \ell K \leq N$.
	Let us introduce the notation $\nu_{1,j}(\omega,K)$ and $\nu_{2,j}^{nc}(\omega,K)$ for the second and third term in Eq.~\eqref{eq:appB8}. Their contribution can be bounded thanks to the small norm of $\xi_j^k(\omega)$, which justifies the iteration of this procedure. 
	
	Moreover, we can finally leverage the locality of $\rho$ to rewrite the first term in Eq.~\eqref{eq:appB8} as:
	\begin{align}
		\sum_{m=0}^{K-1}\;\average{\widehat{h}_j\, \xi_j^0(\omega)\dots\xi_j^m(\omega)e^{i\omega z_j^\ell(m+1)}}{\rho} &= \sum_{m=0}^{K-1}\;\average{\widehat{h}_j\, \xi_j^0(\omega)\dots\xi_j^m(\omega)}{\rho}\average{e^{i\omega z_j^\ell(m+1)}}{\rho} + \nonumber
		\\
		&\qquad+\norbra{\sum_{m=0}^{K-1}\;\average{\widehat{h}_j\, \xi_j^0(\omega)\dots\xi_j^m(\omega)\norbra{e^{i\omega z_j^\ell(m+1)}-\average{e^{i\omega z_j^\ell(m+1)}}{\rho}}}{\rho}}\,.\label{eq:c9}
	\end{align}
	This procedure is akin to the first step in Eq.~\eqref{eq:48b}. We denote the term in the second line by $\nu_{3,j}(\omega,K)$, which will be small thanks to decay of correlations. At this point, we need to introduce the last bit of notation. We define:
	\begin{align}
		\begin{cases}
			\gamma_j^{m}(\omega) := \norbra{e^{-i \omega(\widehat{H}-z^{\ell}_j(m))} \tilde{S}_{m,j}(\omega)- \idO}\,;\\
			\Gamma_j^{m}(\omega) := \norbra{e^{i \omega z^{\ell}_j(m)}-e^{-i \omega(\widehat{H}-z^{\ell}_j(m))} \tilde{S}_{m,j}(\omega)e^{i \omega\widehat{H}}}\,.
		\end{cases}\label{eq:appc10}
	\end{align}
	Once again we choose $\tilde{S}_{m,j}(\omega)$ (to be defined below) so that for commuting models we have $\tilde{S}_{m,j}(\omega) = \idO$. Then, we can expand the expectation value over $e^{i\omega z_j^\ell(m+1)}$ as:
	\begin{align}
		&\average{e^{i\omega z_j^\ell(m+1)}}{\rho} =  \nonumber
		\\
		&\;\;\;\;\;=\average{\norbra{e^{i \omega z^{\ell}_j(m+1)}-e^{-i \omega(\widehat{H}-z^{\ell}_j(m+1))} S_{m+1,j}(\omega)e^{i \omega\widehat{H}}}}{\rho} + \average{\norbra{e^{-i \omega(\widehat{H}-z^{\ell}_j(m+1))}S_{m+1,j}(\omega)-\idO} e^{i \omega\widehat{H}}}{\rho} +\varphi_{\widehat{H}}({\omega}) = \nonumber
		\\
		&\;\;\;\;\;=\average{\Gamma_j^{m+1}(\omega) }{\rho}+\average{\norbra{\gamma_j^{m+1}(\omega) - \average{\gamma_j^{m+1}(\omega) }{\rho}} e^{i \omega\widehat{H}}}{\rho} + \average{\idO+\gamma_j^{m+1}(\omega)}{\rho} \, \varphi_{\widehat{H}}({\omega})\,.
	\end{align}
	This is similar in spirit to the second step in Eq.~\eqref{eq:49b}, where we add and subtract $\varphi_{\widehat{H}}({\omega})$ from the equation.  The corrections arising from this procedure are denoted by $\nu^{nc}_{4,j}(\omega,K) $ and $\nu_{5,j}(\omega,K)$, and take the form: 
	\begin{align}
		\nu^{nc}_{4,j}(\omega,K) :&= \sum_{m=1}^{K-1}\;\average{\widehat{h}_j\, \xi_j^0(\omega)\dots\xi_j^m(\omega)}{\rho}\average{\Gamma_j^{m+1}(\omega)}{\rho}\,;\vspace{0.2cm}
		\\
		\nu_{5,j}(\omega,K) :&=\sum_{m=1}^{K-1}\;\average{\widehat{h}_j\, \xi_j^0(\omega)\dots\xi_j^m(\omega)}{\rho}\average{\norbra{\gamma_j^{m+1}(\omega) - \average{\gamma_j^{m+1}(\omega) }{\rho}} e^{i \omega\widehat{H}}}{\rho}\,.
	\end{align}
	Let $\nu(\omega, K) =  \frac{i}{\sigma_H}\,\sum_{j\in\mathcal{X}}\sum_{k=1}^5\nu_{k,j}(\omega,K)$. Then, Eq.~\eqref{eq:appB8} can be rewritten as:
	\begin{align}
		\varphi_{\widehat{H}}'({\omega}) &= \frac{i}{\sigma_H}\norbra{\sum_{j\in\mathcal{X}}  \sum_{m=0}^{K-1}\;\average{\widehat{h}_j\, \xi_j^0(\omega)\dots\xi_j^m(\omega)}{\rho}\average{\idO+\gamma_j^{m+1}(\omega)}{\rho}} \varphi_{\widehat{H}}({\omega}) + \nu(\omega, K)
		\\
		&=\norbra{-\omega +  \frac{i}{\sigma_H}\sum_{j\in\mathcal{X}} \,\norbra{ \sum_{m=1}^{K-1}\;\average{\widehat{h}_j\, \xi_j^0(\omega)\dots\xi_j^m(\omega)}{\rho}\average{\idO+\gamma_j^{m+1}(\omega)}{\rho} -i\omega\average{\widehat{h}_j\widehat{H}}{\rho}}}\, \varphi_{\widehat{H}}({\omega}) + \nu(\omega, K)\,,
	\end{align}
	where we implicitly used the fact that $\langle \widehat{h}_j\rangle_{\rho}=0$ to take care of the term with $m=0$, and that $\sigma^2_{\widehat{H}}= \sum_{j\in\mathcal{X}} \langle \widehat{h}_j\widehat{H}\rangle_{\rho}/\sigma_H^2=1$ to add and subtract $\omega$ (once again, see the analogy with Eq.~\eqref{eq:41}). This also allows us to formally define $\eta(\omega, k)$ as the second term inside of the parenthesis, which we decompose as $\eta(\omega, K) =  \frac{i}{\sigma_H}\,\sum_{j\in\mathcal{X}}\sum_{k=1}^3\eta_{k,j}(\omega,K)$, where:
	\begin{align}
		\begin{cases}
			\eta_{1,j}(\omega,K) = \average{\widehat{h}_j\norbra{e^{i\omega \widehat{H}_j^\ell(1)}\tilde{R}_{1,j}(\omega)-\idO}}{\rho} -i\omega\average{\widehat{h}_j\widehat{H}_j^\ell(1)}{\rho}\,\vspace{0.2cm}
			\\
			\eta_{2,j}(\omega,K) =-i\omega\average{\widehat{h}_jz_j^\ell(1)}{\rho}\,\vspace{0.2cm}
			\\
			\eta_{3,j}(\omega,K) = \average{\widehat{h}_j\xi_j^1(\omega)}{\rho}\average{\gamma_j^{2}(\omega)}{\rho}+\sum_{m=2}^{K-1}\;\average{\widehat{h}_j\, \xi_j^1(\omega)\dots\xi_j^m(\omega)}{\rho}\average{\idO+\gamma_j^{m+1}(\omega)}{\rho}\,
		\end{cases}\,;
	\end{align}
	notice that we implicitly used the fact that $\widehat{H}=\widehat{H}_j^\ell(1)+z_j^\ell(1)$ to decompose the term $\langle\widehat{h}_j\widehat{H}\rangle_{\rho}$. 
	
	In App. \ref{sec:bounds} we bound the various correcting terms appearing throughout one by one (see Table \ref{tab:operators}), and we wrap everything up in App.~\ref{app:boundEta} and App.~\ref{app:boundNu}. In order to do so, we still need to specify ${\tilde R}_{m,j}(\omega)$ and ${\tilde S}_{m,j}(\omega)$ in Eq.~\eqref{eq:appc4} and Eq.~\eqref{eq:appc10}, done below. 
	
	\setlength{\tabcolsep}{0pt}
	\renewcommand{\arraystretch}{1.6} 
	\begin{table}
		\centering
		\begin{tabular}{|c!{\vrule width 1.2pt}c!{\vrule width 1.2pt}c|}
			\hline
			\rowcolor{colTab!25}
			\;\;\;Operator\;\;\;\qquad&\;\; Definition\;\;  &\;\;\;\;Section\;\;\;\;\\
			\Xhline{1.2pt}
			$\xi_j^m(\omega)$  &  $\norbra{e^{i\omega \widehat{H}_j^\ell(m)}\tilde{R}_{m,j}(\omega)-\idO}$ (for $m\geq1$) \textbf{or} $\idO$ (for $m=0$)& \\
			\hline
			$\Xi_j^m(\omega)$  & $\norbra{e^{i\omega z_j^\ell(m-1)}-e^{i\omega (z_j^\ell(m-1)-z_j^\ell(m))}\tilde{R}_{m,j}(\omega)e^{i\omega z_j^\ell(m) }}$  &\\
			\hline
			$\gamma_j^{m}(\omega)$  & $\norbra{e^{-i \omega(\widehat{H}-z^{\ell}_j(m))} \tilde{S}_{m,j}(\omega)- \idO}$ & \\
			\hline
			$\Gamma_j^{m}(\omega) $  & $\norbra{e^{i \omega z^{\ell}_j(m)}-e^{-i \omega(\widehat{H}-z^{\ell}_j(m))} \tilde{S}_{m,j}(\omega)e^{i \omega\widehat{H}}}$  &\\
			\Xhline{1.2pt}
			$\eta_{1,j}(\omega,K) $  & $\average{\widehat{h}_j\norbra{e^{i\omega \widehat{H}_j^\ell(1)}\tilde{R}_{1,j}(\omega)-\idO}}{\rho} -i\omega\average{\widehat{h}_j\widehat{H}_j^\ell(1)}{\rho}$ & App.~\ref{app:eta1}\\
			\hline
			$\eta_{2,j}(\omega,K)$  & $-i\omega\average{\widehat{h}_jz_j^\ell(1)}{\rho}$  & App.~\ref{app:eta2}\\
			\hline
			$\eta_{3,j}(\omega,K) $  &\qquad\qquad $\average{\widehat{h}_j\xi_j^1(\omega)}{\rho}\average{\gamma_j^{2}(\omega)}{\rho}+\sum_{m=2}^{K-1}\;\average{\widehat{h}_j\, \xi_j^1(\omega)\dots\xi_j^m(\omega)}{\rho}\average{\idO+\gamma_j^{m+1}(\omega)}{\rho}$\qquad\qquad\qquad\qquad  & App.~\ref{app:eta3}\\
			\Xhline{1.2pt}
			$\nu_{1,j}(\omega,K) $  & $\average{\widehat{h}_j\, \xi_j^0(\omega)\dots\xi_j^K(\omega)e^{i\omega z_j^\ell(K)}}{\rho}$  &App.~\ref{app:nu1}\\
			\hline
			$\nu_{2,j}^{nc}(\omega,K)  $ & $\sum_{m=1}^{K-1}\;\average{\widehat{h}_j\, \xi_j^0(\omega)\dots\xi_j^{m-1}(\omega)\Xi_j^{m}(\omega)}{\rho}$ & App.~\ref{app:nu2}\\
			\hline
			$\nu_{3,j}(\omega,K) $  & $\sum_{m=0}^{K-1}\;\average{\widehat{h}_j\, \xi_j^0(\omega)\dots\xi_j^m(\omega)\norbra{e^{i\omega z_j^\ell(m+1)}-\average{e^{i\omega z_j^\ell(m+1)}}{\rho}}}{\rho}$ & App.~\ref{app:nu3}\\
			\hline
			$\nu_{4,j}^{nc}(\omega,K)  $ & \qquad\qquad$\sum_{m=1}^{K-1}\;\average{\widehat{h}_j\, \xi_j^0(\omega)\dots\xi_j^m(\omega)}{\rho}\average{\Gamma_j^{m+1}(\omega)}{\rho}$ \qquad\qquad\qquad & App.~\ref{app:nu2}\\
			\hline
			$\nu_{5,j}(\omega,K) $  & $\sum_{m=1}^{K-1}\;\average{\widehat{h}_j\, \xi_j^0(\omega)\dots\xi_j^m(\omega)}{\rho}\average{\norbra{\gamma_j^{m+1}(\omega) - \average{\gamma_j^{m+1}(\omega) }{\rho}} e^{i \omega\widehat{H}}}{\rho}$ &App.~\ref{app:nu5}\\
			\hline
		\end{tabular}
		\caption{Glossary of operators and terms used along the derivation, together with the section in which each term is bounded.}
		\label{tab:operators}
	\end{table}
	
	\subsection{The choice of the ${\tilde R}_{m,j}(\omega)$ and ${\tilde S}_{m,j}(\omega)$}
	There are two possible limiting choices. First, one can choose ${\tilde{R}}_{m,j}(\omega) \equiv \idO$. In this case we have:
	\begin{align}
		\begin{cases}
			\xi_j^0(\omega) := \idO\,;\\
			\xi_j^m(\omega) := \norbra{e^{i\omega (z_j^\ell(m-1)-z_j^\ell(m))}-\idO} = \norbra{e^{i\omega \widehat{H}_j^\ell(m)}-\idO}\,;\\
			\Xi_j^m(\omega) := \norbra{e^{i\omega z_j^\ell(m-1)}-e^{i\omega (z_j^\ell(m-1)-z_j^\ell(m))}e^{i\omega z_j^\ell(m) }}\,.
		\end{cases}
		\qquad\text{For ${\tilde{R}}_{m,j}(\omega) = \idO$}\,,\label{eq:M1}
	\end{align}
	meaning that $\xi_j^m(\omega)$ becomes completely localized in the region $d(k,j)\in[2R\ell (m-1),2R\ell m]$, while $\Xi_j^m(\omega)$ exactly accounts for the non-commutativity between $\widehat{H}_j^\ell(m)$ and $z_j^\ell(m)$. In the case of commuting Hamiltonian, this choice gives the desired result, since $\Xi_j^m(\omega) = 0$ (a similar discussion also applies to $\tilde{S}_{m,j}(\omega)$ and $\Gamma_j^{m}(\omega)$). Still, for non-commuting models this correction is finite, and it does not decay sufficiently fast to be a useful choice in the derivation.
	
	In the opposite limit, one can choose ${\tilde R}_{m,j}(\omega) \equiv R^{\infty}_{m,j}(\omega)$ and ${\tilde S}_{m,j}(\omega) \equiv S^{\infty}_{m,j}(\omega)$, where we introduced the two operators:
	\begin{align}
		R^{\infty}_{m,j}(\omega):=e^{-i\omega (z_j^\ell(m-1)-z_j^\ell(m))}e^{i\omega z_j^\ell(m-1) }e^{-i\omega z_j^\ell(m) }\,;\qquad S^{\infty}_{m,j}(\omega):=e^{i \omega(\widehat{H}-z^{\ell}_j(m))} e^{i \omega z^{\ell}_j(m)}e^{-i \omega\widehat{H}}\,,\label{eq:RSinfty}
	\end{align}
	Then, if we consider the effect of this choice on $\xi_j^m(\omega)$ and $\Xi_j^m(\omega)$, it follows that:
	\begin{align}
		\begin{cases}
			\xi_j^0(\omega) := \idO\,;\\
			\xi_j^m(\omega) := \norbra{e^{i\omega z_j^\ell(m-1) }e^{-i\omega z_j^\ell(m) }-\idO}\,;\\
			\Xi_j^m(\omega) := 0\,.
		\end{cases}
		\qquad\text{For ${\tilde{R}}_{m,j}(\omega) = R^{\infty}_{m,j}(\omega)$}\,.\label{eq:Minfinity}
	\end{align}
	This means that for commuting and non-commuting models the correction arising from $\Xi_j^m(\omega)$ disappears. On the other hand, though, $\xi_j^m(\omega)$ becomes completely delocalized, which prevents the error arising from the decomposition in Eq.~\eqref{eq:c9} to be small. 
	
	For this reason, we look for a choice of ${\tilde R}_{m,j}(\omega)$ that makes $\Xi_j^m(\omega)$ sufficiently small, while keeping the locality of $\xi_j^m(\omega)$. This is achieved with the definition:
	\begin{align}
		\begin{cases}
			{\tilde R}_{m,j}(\omega)=R_{m,j}^M(\omega) = \idO + \sum_{n=1}^{M}\;\frac{\omega^n}{n!}\,R^{(n)}_{m,j}(0)
			\\
			{\tilde S}_{m,j}(\omega)=S_{m,j}^M(\omega) = \idO + \sum_{n=1}^{M}\;\frac{\omega^n}{n!}\,S^{(n)}_{m,j}(0)
		\end{cases}\,,\label{eq:c17}
	\end{align}
	where $R^{(m)}_{m,j}(0) :=\partial^{m}_\omega R^{\infty}_{m,j}(\omega)\big|_{\omega=0}$  and $S^{(m)}_{m,j}(0) :=\partial^{m}_\omega S^{\infty}_{m,j}(\omega)\big|_{\omega=0}$. That is, we choose ${\tilde R}_{m,j}(\omega)$ and ${\tilde S}_{m,j}(\omega)$ to be the Taylor expansion of $R^{\infty}_{m,j}(\omega)$ and $S^{\infty}_{m,j}(\omega)$ truncated at order $M$. Choosing $M=0$ or $M=\infty$ gives the two limiting behaviors discussed above in Eq.~\eqref{eq:M1} and Eq.~\eqref{eq:Minfinity}. 
	
	As it will be clear in the following, fixing $M$ in an intermediate regime between the two will allow us to achieve the correct scaling. Indeed, we can prove the following:
	
	\begin{lemma}\label{lemma:clusterExp}
		Let $R^{\infty}_{m,j}(\omega)$ and $S^{\infty}_{m,j}(\omega)$ be the two operators defined in Eq.~\eqref{eq:RSinfty}, and let $R_{m,j}^M(\omega)$ and $S_{m,j}^M(\omega)$ be their truncated Taylor expansion at order $M$.  the following properties hold:
		\begin{enumerate}
			\item $R^{\infty}_{m,j}(0)=S^{\infty}_{m,j}(0)=\idO$ and $R^{(1)}_{m,j}(0)=S^{(1)}_{m,j}(0)=0$;\label{it:zeros1}
			\item Their support satisfies:\label{it:supp1}
			\begin{align}
				&{\rm supp}(R_{m,j}^M(\omega) ) \subseteq\{i| d(j,i)\in [2R (\ell (m-1)-M),2R (\ell\, m+M)]\}\,;
				\\
				&{\rm supp}(S_{m,j}^M(\omega) ) \subseteq\{i| d(j,i)\in [2R (\ell (m-1)-M),2R (\ell\, m+M)]\}\,.
			\end{align}
			\item Their $n$-th derivative is upper bounded by:\label{it:norm1}
			\begin{align}
				\|R^{(m)}_{m,j}(0)\|,\,\|S^{(m)}_{m,j}(0)\|\leq \norbra{\frac{E\,\Gamma \ell^{\frac{D}{2}} m^{\frac{D-1}{2}}}{\sigma_H}}^{n} (n)!\,.\label{eq:l4e31M}
			\end{align}
			where we introduced the constant $\Gamma= \max\{(4c_D (2R)^D ),  (2c_D^2(2R)^{D})\}$. This implies that:
			\begin{align}
				\|R_{m,j}^M(\omega)\|,\,\|S_{m,j}^M(\omega)\| \leq \frac{1-\norbra{\frac{\Gamma \ell^{\frac{D}{2}} m^{\frac{D-1}{2}}}{\sigma_H} (E|\omega|)}^{M+1}}{1-\norbra{\frac{\Gamma \ell^{\frac{D}{2}} m^{\frac{D-1}{2}}}{\sigma_H} (E|\omega|)}}\,.
			\end{align}
		\end{enumerate}
		Now, let $m\leq K$, and let $\omega\in[0,\Omega_1)$, where $\Omega_1 := \norbra{\frac{\sigma_H}{2E\,\Gamma \ell^{\frac{D}{2}} K^{\frac{D-1}{2}}}}$. Then, it also holds that:
		\begin{align}
			\|R_{m,j}^M(\omega)\|\leq 2\norbra{1 - \frac{1}{2^{M+1}}}\,;\;\;&\;\;\|S_{m,j}^M(\omega)\|\leq 2\norbra{1 - \frac{1}{2^{M+1}}}\,;\nonumber
			\\
			\|R_{m,j}^\infty(\omega)-R_{m,j}^M(\omega)\|\leq 2\norbra{\frac{\Gamma \ell^{\frac{D}{2}} m^{\frac{D-1}{2}}}{\sigma_H} (E|\omega|)}^{M+1}\,;\;\;&\;\;\|S_{m,j}^\infty(\omega)-R_{m,j}^M(\omega)\|\leq 2\norbra{\frac{\Gamma \ell^{\frac{D}{2}} m^{\frac{D-1}{2}}}{\sigma_H} (E|\omega|)}^{M+1}\,;\label{it:convNorm1}
			\\
			\|\partial^{(k)}_\omega R_{m,j}^M(\omega)\|\leq 2\,(k!) \,\norbra{\frac{E\,\Gamma \ell^{\frac{D}{2}} m^{\frac{D-1}{2}}}{\sigma_H} }^k\,;\;\;&\;\;\|\partial^{(k)}_\omega S_{m,j}^M(\omega)\|\leq 2\,(k!) \,\norbra{\frac{E\,\Gamma \ell^{\frac{D}{2}} m^{\frac{D-1}{2}}}{\sigma_H} }^k\,.\nonumber
		\end{align}
	\end{lemma}
	The proof is based on the cluster expansion, and it is deferred to App.~\ref{app:clusterExp}. 



\section{Conclusion}\label{sec:conclusion}

The main open technical question of this work is whether a proof similar to that of Theorem \ref{thm:BEprod} can be found also for Theorem \ref{thm:BE}, which will simplify the arguments dramatically. This requires proving the convergence of the cluster expansion of the characteristic function of a state with correlation decay, for which none of the existing techniques directly apply \cite{Kuwahara_2020,Haah_2024,Wild_2023,Mann_2024}. The proof of the concentration bound for states with decaying correlations in \cite{Anshu_2016} might yield some insights along these lines. It would also be interesting to understand whether an alternative set of conditions similar to the classical of Lindeberg's \cite{Lindeberg1922} can also yield a similar result.  It is also of relevance to extend this type of result to the infinite dimensional bosonic setting \cite{Cramer_2010,Arous2013} and to long-range Hamiltonians beyond \cite{Sanchez_Segovia_2025}.

It is understood already from simple classical i.i.d. variables that the bound $\tilde {\mathcal{O}}(N^{-\frac{1}{2}})$ derived here is tight up to polylogarithmic factors, and there are also explicit examples of that tightness in quantum many-body systems \cite{ScarsProof}. However, chaotic quantum systems, with the well-established level repulsion, have much smoother spectra than discrete classical random variables, so it is possible that there is a large class of interacting quantum systems presenting a much faster rate of convergence, up to exponentially fast (see \cite{Rai_2024} for numerical evidence for small system sizes). This in itself could have relevant consequences for the thermalization of chaotic quantum systems \cite{Wilming_2018}.

    
	\acknowledgements
	
	AMA thanks Dominik Wild for useful discussions regarding Thm. \ref{thm:BEprod} and the proof in \cite{Rai_2024}. AMA and MS acknowledge support from the Spanish Agencia Estatal de Investigacion through the grants ``IFT Centro de Excelencia Severo Ochoa CEX2020-001007-S", ``PID2023-150847NA-I00'', ``EUR2025-164823'' and  ``PCI2024-153448'', funded by MICIU/AEI/10.13039/501100011033 and co-funded by the EU. MG acknowledges the support from the Engineering and Physical Sciences Research Council, grant no. EP/T022140/1. This
project was funded within the QuantERA II Programme that has received funding from the EU’s H2020 research and innovation programme under the GA No 101017733.

\bibliography{bib.bib}

\begin{thebibliography}{64}%
\makeatletter
\providecommand \@ifxundefined [1]{%
 \@ifx{#1\undefined}
}%
\providecommand \@ifnum [1]{%
 \ifnum #1\expandafter \@firstoftwo
 \else \expandafter \@secondoftwo
 \fi
}%
\providecommand \@ifx [1]{%
 \ifx #1\expandafter \@firstoftwo
 \else \expandafter \@secondoftwo
 \fi
}%
\providecommand \natexlab [1]{#1}%
\providecommand \enquote  [1]{``#1''}%
\providecommand \bibnamefont  [1]{#1}%
\providecommand \bibfnamefont [1]{#1}%
\providecommand \citenamefont [1]{#1}%
\providecommand \href@noop [0]{\@secondoftwo}%
\providecommand \href [0]{\begingroup \@sanitize@url \@href}%
\providecommand \@href[1]{\@@startlink{#1}\@@href}%
\providecommand \@@href[1]{\endgroup#1\@@endlink}%
\providecommand \@sanitize@url [0]{\catcode `\\12\catcode `\$12\catcode
  `\&12\catcode `\#12\catcode `\^12\catcode `\_12\catcode `\%12\relax}%
\providecommand \@@startlink[1]{}%
\providecommand \@@endlink[0]{}%
\providecommand \url  [0]{\begingroup\@sanitize@url \@url }%
\providecommand \@url [1]{\endgroup\@href {#1}{\urlprefix }}%
\providecommand \urlprefix  [0]{URL }%
\providecommand \Eprint [0]{\href }%
\providecommand \doibase [0]{https://doi.org/}%
\providecommand \selectlanguage [0]{\@gobble}%
\providecommand \bibinfo  [0]{\@secondoftwo}%
\providecommand \bibfield  [0]{\@secondoftwo}%
\providecommand \translation [1]{[#1]}%
\providecommand \BibitemOpen [0]{}%
\providecommand \bibitemStop [0]{}%
\providecommand \bibitemNoStop [0]{.\EOS\space}%
\providecommand \EOS [0]{\spacefactor3000\relax}%
\providecommand \BibitemShut  [1]{\csname bibitem#1\endcsname}%
\let\auto@bib@innerbib\@empty
\bibitem [{\citenamefont {Lieb}\ and\ \citenamefont
  {Robinson}(1972)}]{Lieb1972}%
  \BibitemOpen
  \bibfield  {author} {\bibinfo {author} {\bibfnamefont {E.~H.}\ \bibnamefont
  {Lieb}}\ and\ \bibinfo {author} {\bibfnamefont {D.~W.}\ \bibnamefont
  {Robinson}},\ }\bibfield  {title} {\bibinfo {title} {The finite group
  velocity of quantum spin systems},\ }\href
  {https://doi.org/10.1007/BF01645779} {\bibfield  {journal} {\bibinfo
  {journal} {Communications in Mathematical Physics}\ }\textbf {\bibinfo
  {volume} {28}},\ \bibinfo {pages} {251} (\bibinfo {year} {1972})}\BibitemShut
  {NoStop}%
\bibitem [{\citenamefont {Hastings}\ and\ \citenamefont
  {Koma}(2006)}]{Hastings_2006}%
  \BibitemOpen
  \bibfield  {author} {\bibinfo {author} {\bibfnamefont {M.~B.}\ \bibnamefont
  {Hastings}}\ and\ \bibinfo {author} {\bibfnamefont {T.}~\bibnamefont
  {Koma}},\ }\bibfield  {title} {\bibinfo {title} {Spectral gap and exponential
  decay of correlations},\ }\href {https://doi.org/10.1007/s00220-006-0030-4}
  {\bibfield  {journal} {\bibinfo  {journal} {Communications in Mathematical
  Physics}\ }\textbf {\bibinfo {volume} {265}},\ \bibinfo {pages} {781–804}
  (\bibinfo {year} {2006})}\BibitemShut {NoStop}%
\bibitem [{\citenamefont {Nachtergaele}\ and\ \citenamefont
  {Sims}(2006)}]{Nachtergaele2006}%
  \BibitemOpen
  \bibfield  {author} {\bibinfo {author} {\bibfnamefont {B.}~\bibnamefont
  {Nachtergaele}}\ and\ \bibinfo {author} {\bibfnamefont {R.}~\bibnamefont
  {Sims}},\ }\bibfield  {title} {\bibinfo {title} {Lieb-{R}obinson bounds and
  the exponential clustering theorem},\ }\href
  {https://doi.org/10.1007/s00220-006-1556-1} {\bibfield  {journal} {\bibinfo
  {journal} {Communications in Mathematical Physics}\ }\textbf {\bibinfo
  {volume} {265}},\ \bibinfo {pages} {119} (\bibinfo {year}
  {2006})}\BibitemShut {NoStop}%
\bibitem [{\citenamefont {Hastings}(2007)}]{Hastings_2007}%
  \BibitemOpen
  \bibfield  {author} {\bibinfo {author} {\bibfnamefont {M.~B.}\ \bibnamefont
  {Hastings}},\ }\bibfield  {title} {\bibinfo {title} {An area law for
  one-dimensional quantum systems},\ }\href
  {https://doi.org/10.1088/1742-5468/2007/08/p08024} {\bibfield  {journal}
  {\bibinfo  {journal} {Journal of Statistical Mechanics: Theory and
  Experiment}\ }\textbf {\bibinfo {volume} {2007}},\ \bibinfo {pages}
  {P08024–P08024} (\bibinfo {year} {2007})}\BibitemShut {NoStop}%
\bibitem [{\citenamefont {Wolf}\ \emph {et~al.}(2008)\citenamefont {Wolf},
  \citenamefont {Verstraete}, \citenamefont {Hastings},\ and\ \citenamefont
  {Cirac}}]{Wolf_2008}%
  \BibitemOpen
  \bibfield  {author} {\bibinfo {author} {\bibfnamefont {M.~M.}\ \bibnamefont
  {Wolf}}, \bibinfo {author} {\bibfnamefont {F.}~\bibnamefont {Verstraete}},
  \bibinfo {author} {\bibfnamefont {M.~B.}\ \bibnamefont {Hastings}},\ and\
  \bibinfo {author} {\bibfnamefont {J.~I.}\ \bibnamefont {Cirac}},\ }\bibfield
  {title} {\bibinfo {title} {Area laws in quantum systems: Mutual information
  and correlations},\ }\bibfield  {journal} {\bibinfo  {journal} {Physical
  Review Letters}\ }\textbf {\bibinfo {volume} {100}},\ \href
  {https://doi.org/10.1103/physrevlett.100.070502}
  {10.1103/physrevlett.100.070502} (\bibinfo {year} {2008})\BibitemShut
  {NoStop}%
\bibitem [{\citenamefont {Eisert}\ \emph {et~al.}(2010)\citenamefont {Eisert},
  \citenamefont {Cramer},\ and\ \citenamefont {Plenio}}]{Eisert_2010}%
  \BibitemOpen
  \bibfield  {author} {\bibinfo {author} {\bibfnamefont {J.}~\bibnamefont
  {Eisert}}, \bibinfo {author} {\bibfnamefont {M.}~\bibnamefont {Cramer}},\
  and\ \bibinfo {author} {\bibfnamefont {M.~B.}\ \bibnamefont {Plenio}},\
  }\bibfield  {title} {\bibinfo {title} {Colloquium: Area laws for the
  entanglement entropy},\ }\href {https://doi.org/10.1103/revmodphys.82.277}
  {\bibfield  {journal} {\bibinfo  {journal} {Reviews of Modern Physics}\
  }\textbf {\bibinfo {volume} {82}},\ \bibinfo {pages} {277–306} (\bibinfo
  {year} {2010})}\BibitemShut {NoStop}%
\bibitem [{\citenamefont {Kliesch}\ \emph {et~al.}(2014)\citenamefont
  {Kliesch}, \citenamefont {Gogolin}, \citenamefont {Kastoryano}, \citenamefont
  {Riera},\ and\ \citenamefont {Eisert}}]{Kliesch_2014}%
  \BibitemOpen
  \bibfield  {author} {\bibinfo {author} {\bibfnamefont {M.}~\bibnamefont
  {Kliesch}}, \bibinfo {author} {\bibfnamefont {C.}~\bibnamefont {Gogolin}},
  \bibinfo {author} {\bibfnamefont {M.~J.}\ \bibnamefont {Kastoryano}},
  \bibinfo {author} {\bibfnamefont {A.}~\bibnamefont {Riera}},\ and\ \bibinfo
  {author} {\bibfnamefont {J.}~\bibnamefont {Eisert}},\ }\bibfield  {title}
  {\bibinfo {title} {Locality of temperature},\ }\bibfield  {journal} {\bibinfo
   {journal} {Physical Review X}\ }\textbf {\bibinfo {volume} {4}},\ \href
  {https://doi.org/10.1103/physrevx.4.031019} {10.1103/physrevx.4.031019}
  (\bibinfo {year} {2014})\BibitemShut {NoStop}%
\bibitem [{\citenamefont {Hartmann}\ \emph {et~al.}(2004)\citenamefont
  {Hartmann}, \citenamefont {Mahler},\ and\ \citenamefont
  {Hess}}]{Hartmann_2004}%
  \BibitemOpen
  \bibfield  {author} {\bibinfo {author} {\bibfnamefont {M.}~\bibnamefont
  {Hartmann}}, \bibinfo {author} {\bibfnamefont {G.}~\bibnamefont {Mahler}},\
  and\ \bibinfo {author} {\bibfnamefont {O.}~\bibnamefont {Hess}},\ }\bibfield
  {title} {\bibinfo {title} {Gaussian quantum fluctuations in interacting many
  particle systems},\ }\href
  {https://doi.org/10.1023/b:math.0000043321.00896.86} {\bibfield  {journal}
  {\bibinfo  {journal} {Letters in Mathematical Physics}\ }\textbf {\bibinfo
  {volume} {68}},\ \bibinfo {pages} {103–112} (\bibinfo {year}
  {2004})}\BibitemShut {NoStop}%
\bibitem [{\citenamefont {Erdös}\ and\ \citenamefont
  {Schröder}(2014)}]{Erd_s_2014}%
  \BibitemOpen
  \bibfield  {author} {\bibinfo {author} {\bibfnamefont {L.}~\bibnamefont
  {Erdös}}\ and\ \bibinfo {author} {\bibfnamefont {D.}~\bibnamefont
  {Schröder}},\ }\bibfield  {title} {\bibinfo {title} {Phase transition in the
  density of states of quantum spin glasses},\ }\href
  {https://doi.org/10.1007/s11040-014-9164-3} {\bibfield  {journal} {\bibinfo
  {journal} {Mathematical Physics, Analysis and Geometry}\ }\textbf {\bibinfo
  {volume} {17}},\ \bibinfo {pages} {441–464} (\bibinfo {year}
  {2014})}\BibitemShut {NoStop}%
\bibitem [{\citenamefont {Keating}\ \emph {et~al.}(2015)\citenamefont
  {Keating}, \citenamefont {Linden},\ and\ \citenamefont
  {Wells}}]{Keating_2015}%
  \BibitemOpen
  \bibfield  {author} {\bibinfo {author} {\bibfnamefont {J.~P.}\ \bibnamefont
  {Keating}}, \bibinfo {author} {\bibfnamefont {N.}~\bibnamefont {Linden}},\
  and\ \bibinfo {author} {\bibfnamefont {H.~J.}\ \bibnamefont {Wells}},\
  }\bibfield  {title} {\bibinfo {title} {Spectra and eigenstates of spin chain
  {H}amiltonians},\ }\href {https://doi.org/10.1007/s00220-015-2366-0}
  {\bibfield  {journal} {\bibinfo  {journal} {Communications in Mathematical
  Physics}\ }\textbf {\bibinfo {volume} {338}},\ \bibinfo {pages} {81–102}
  (\bibinfo {year} {2015})}\BibitemShut {NoStop}%
\bibitem [{\citenamefont {Berry}(1941)}]{berry1941accuracy}%
  \BibitemOpen
  \bibfield  {author} {\bibinfo {author} {\bibfnamefont {A.~C.}\ \bibnamefont
  {Berry}},\ }\bibfield  {title} {\bibinfo {title} {The accuracy of the
  {G}aussian approximation to the sum of independent variates},\ }\href
  {https://doi.org/https://doi.org/10.2307/1990053} {\bibfield  {journal}
  {\bibinfo  {journal} {Transactions of the american mathematical society}\
  }\textbf {\bibinfo {volume} {49}},\ \bibinfo {pages} {122} (\bibinfo {year}
  {1941})}\BibitemShut {NoStop}%
\bibitem [{\citenamefont {Esseen}(1942)}]{esseen1942liapounoff}%
  \BibitemOpen
  \bibfield  {author} {\bibinfo {author} {\bibfnamefont {C.}~\bibnamefont
  {Esseen}},\ }\href {https://books.google.es/books?id=VjXgPgAACAAJ} {\emph
  {\bibinfo {title} {On the Liapounoff Limit of Error in the Theory of
  Probability}}},\ Arkiv f{\"o}r matematik, astronomi och fysik\ (\bibinfo
  {publisher} {Almqvist \& Wiksell},\ \bibinfo {year} {1942})\BibitemShut
  {NoStop}%
\bibitem [{\citenamefont {Feller}(1991)}]{feller1991introduction}%
  \BibitemOpen
  \bibfield  {author} {\bibinfo {author} {\bibfnamefont {W.}~\bibnamefont
  {Feller}},\ }\href@noop {} {\emph {\bibinfo {title} {An Introduction to
  Probability Theory and Its Applications, Volume I}}},\ \bibinfo {edition}
  {3rd}\ ed.\ (\bibinfo  {publisher} {Wiley},\ \bibinfo {year}
  {1991})\BibitemShut {NoStop}%
\bibitem [{\citenamefont
  {Vershynin}(2026)}]{vershynin2026friendlyproofberryesseentheorem}%
  \BibitemOpen
  \bibfield  {author} {\bibinfo {author} {\bibfnamefont {R.}~\bibnamefont
  {Vershynin}},\ }\href {https://arxiv.org/abs/2602.06234} {\bibinfo {title} {A
  friendly proof of the {B}erry-{E}sseen theorem}} (\bibinfo {year} {2026}),\
  \Eprint {https://arxiv.org/abs/2602.06234} {arXiv:2602.06234} \BibitemShut
  {NoStop}%
\bibitem [{\citenamefont {Tikhomirov}(1981)}]{tikhomirov1981convergence}%
  \BibitemOpen
  \bibfield  {author} {\bibinfo {author} {\bibfnamefont {A.~N.}\ \bibnamefont
  {Tikhomirov}},\ }\bibfield  {title} {\bibinfo {title} {On the convergence
  rate in the central limit theorem for weakly dependent random variables},\
  }\href {https://doi.org/doi.org/10.1137/1125092} {\bibfield  {journal}
  {\bibinfo  {journal} {Theory of Probability \& Its Applications}\ }\textbf
  {\bibinfo {volume} {25}},\ \bibinfo {pages} {790} (\bibinfo {year}
  {1981})}\BibitemShut {NoStop}%
\bibitem [{\citenamefont {Sunklodas}(1984)}]{sunklodas1984rate}%
  \BibitemOpen
  \bibfield  {author} {\bibinfo {author} {\bibfnamefont {J.}~\bibnamefont
  {Sunklodas}},\ }\bibfield  {title} {\bibinfo {title} {Rate of convergence in
  the central limit theorem for random variables with strong mixing},\ }\href
  {https://doi.org/https://doi.org/10.1007/BF00970405} {\bibfield  {journal}
  {\bibinfo  {journal} {Lithuanian Mathematical Journal}\ }\textbf {\bibinfo
  {volume} {24}},\ \bibinfo {pages} {182} (\bibinfo {year} {1984})}\BibitemShut
  {NoStop}%
\bibitem [{\citenamefont {Anshu}(2016)}]{Anshu_2016}%
  \BibitemOpen
  \bibfield  {author} {\bibinfo {author} {\bibfnamefont {A.}~\bibnamefont
  {Anshu}},\ }\bibfield  {title} {\bibinfo {title} {Concentration bounds for
  quantum states with finite correlation length on quantum spin lattice
  systems},\ }\href {https://doi.org/10.1088/1367-2630/18/8/083011} {\bibfield
  {journal} {\bibinfo  {journal} {New Journal of Physics}\ }\textbf {\bibinfo
  {volume} {18}},\ \bibinfo {pages} {083011} (\bibinfo {year}
  {2016})}\BibitemShut {NoStop}%
\bibitem [{\citenamefont {Kuwahara}(2016)}]{Kuwahara_2016}%
  \BibitemOpen
  \bibfield  {author} {\bibinfo {author} {\bibfnamefont {T.}~\bibnamefont
  {Kuwahara}},\ }\bibfield  {title} {\bibinfo {title} {Connecting the
  probability distributions of different operators and generalization of the
  {C}hernoff–{H}oeffding inequality},\ }\href
  {https://doi.org/10.1088/1742-5468/2016/11/113103} {\bibfield  {journal}
  {\bibinfo  {journal} {Journal of Statistical Mechanics: Theory and
  Experiment}\ }\textbf {\bibinfo {volume} {2016}},\ \bibinfo {pages} {113103}
  (\bibinfo {year} {2016})}\BibitemShut {NoStop}%
\bibitem [{\citenamefont {Kuwahara}\ and\ \citenamefont
  {Saito}(2020{\natexlab{a}})}]{Kuwahara_2020}%
  \BibitemOpen
  \bibfield  {author} {\bibinfo {author} {\bibfnamefont {T.}~\bibnamefont
  {Kuwahara}}\ and\ \bibinfo {author} {\bibfnamefont {K.}~\bibnamefont
  {Saito}},\ }\bibfield  {title} {\bibinfo {title} {Gaussian concentration
  bound and ensemble equivalence in generic quantum many-body systems including
  long-range interactions},\ }\href {https://doi.org/10.1016/j.aop.2020.168278}
  {\bibfield  {journal} {\bibinfo  {journal} {Annals of Physics}\ }\textbf
  {\bibinfo {volume} {421}},\ \bibinfo {pages} {168278} (\bibinfo {year}
  {2020}{\natexlab{a}})}\BibitemShut {NoStop}%
\bibitem [{\citenamefont {De~Palma}\ and\ \citenamefont
  {Rouz{\'e}}(2022)}]{DePalma2022}%
  \BibitemOpen
  \bibfield  {author} {\bibinfo {author} {\bibfnamefont {G.}~\bibnamefont
  {De~Palma}}\ and\ \bibinfo {author} {\bibfnamefont {C.}~\bibnamefont
  {Rouz{\'e}}},\ }\bibfield  {title} {\bibinfo {title} {Quantum concentration
  inequalities},\ }\href {https://doi.org/10.1007/s00023-022-01181-1}
  {\bibfield  {journal} {\bibinfo  {journal} {Annales Henri Poincar{\'e}}\
  }\textbf {\bibinfo {volume} {23}},\ \bibinfo {pages} {3391} (\bibinfo {year}
  {2022})}\BibitemShut {NoStop}%
\bibitem [{\citenamefont {Anshu}\ and\ \citenamefont
  {Metger}(2023)}]{Anshu_2023}%
  \BibitemOpen
  \bibfield  {author} {\bibinfo {author} {\bibfnamefont {A.}~\bibnamefont
  {Anshu}}\ and\ \bibinfo {author} {\bibfnamefont {T.}~\bibnamefont {Metger}},\
  }\bibfield  {title} {\bibinfo {title} {Concentration bounds for quantum
  states and limitations on the {QAOA} from polynomial approximations},\ }\href
  {https://doi.org/10.22331/q-2023-05-11-999} {\bibfield  {journal} {\bibinfo
  {journal} {Quantum}\ }\textbf {\bibinfo {volume} {7}},\ \bibinfo {pages}
  {999} (\bibinfo {year} {2023})}\BibitemShut {NoStop}%
\bibitem [{\citenamefont {Wild}\ and\ \citenamefont
  {Alhambra}(2023)}]{Wild_2023}%
  \BibitemOpen
  \bibfield  {author} {\bibinfo {author} {\bibfnamefont {D.~S.}\ \bibnamefont
  {Wild}}\ and\ \bibinfo {author} {\bibfnamefont {A.~M.}\ \bibnamefont
  {Alhambra}},\ }\bibfield  {title} {\bibinfo {title} {Classical simulation of
  short-time quantum dynamics},\ }\bibfield  {journal} {\bibinfo  {journal}
  {PRX Quantum}\ }\textbf {\bibinfo {volume} {4}},\ \href
  {https://doi.org/10.1103/prxquantum.4.020340} {10.1103/prxquantum.4.020340}
  (\bibinfo {year} {2023})\BibitemShut {NoStop}%
\bibitem [{\citenamefont {Brandao}\ \emph {et~al.}(2015)\citenamefont
  {Brandao}, \citenamefont {Cramer},\ and\ \citenamefont
  {Guta}}]{brandao2015berry}%
  \BibitemOpen
  \bibfield  {author} {\bibinfo {author} {\bibfnamefont {F.}~\bibnamefont
  {Brandao}}, \bibinfo {author} {\bibfnamefont {M.}~\bibnamefont {Cramer}},\
  and\ \bibinfo {author} {\bibfnamefont {M.}~\bibnamefont {Guta}},\ }\bibfield
  {title} {\bibinfo {title} {Berry-{E}sseen theorem for quantum lattice systems
  and the equivalence of statistical mechanical ensembles},\ }\href
  {https://qip.iaqi.org/2015/talks/125-Brandao.pdf} {\bibfield  {journal}
  {\bibinfo  {journal} {QIP2015 Talk}\ } (\bibinfo {year} {2015})}\BibitemShut
  {NoStop}%
\bibitem [{\citenamefont
  {Hayashi}(2006)}]{hayashi2006quantumestimationquantumcentral}%
  \BibitemOpen
  \bibfield  {author} {\bibinfo {author} {\bibfnamefont {M.}~\bibnamefont
  {Hayashi}},\ }\href {https://arxiv.org/abs/quant-ph/0608198} {\bibinfo
  {title} {Quantum estimation and the quantum central limit theorem}} (\bibinfo
  {year} {2006}),\ \Eprint {https://arxiv.org/abs/quant-ph/0608198}
  {arXiv:quant-ph/0608198} \BibitemShut {NoStop}%
\bibitem [{\citenamefont {Guţă}\ and\ \citenamefont
  {Jenčová}(2007)}]{Gu__2007}%
  \BibitemOpen
  \bibfield  {author} {\bibinfo {author} {\bibfnamefont {M.}~\bibnamefont
  {Guţă}}\ and\ \bibinfo {author} {\bibfnamefont {A.}~\bibnamefont
  {Jenčová}},\ }\bibfield  {title} {\bibinfo {title} {Local asymptotic
  normality in quantum statistics},\ }\href
  {https://doi.org/10.1007/s00220-007-0340-1} {\bibfield  {journal} {\bibinfo
  {journal} {Communications in Mathematical Physics}\ }\textbf {\bibinfo
  {volume} {276}},\ \bibinfo {pages} {341–379} (\bibinfo {year}
  {2007})}\BibitemShut {NoStop}%
\bibitem [{\citenamefont {Kahn}\ and\ \citenamefont {Guţă}(2009)}]{Kahn2009}%
  \BibitemOpen
  \bibfield  {author} {\bibinfo {author} {\bibfnamefont {J.}~\bibnamefont
  {Kahn}}\ and\ \bibinfo {author} {\bibfnamefont {M.}~\bibnamefont {Guţă}},\
  }\bibfield  {title} {\bibinfo {title} {Local asymptotic normality for finite
  dimensional quantum systems},\ }\href
  {https://doi.org/10.1007/s00220-009-0787-3} {\bibfield  {journal} {\bibinfo
  {journal} {Communications in Mathematical Physics}\ }\textbf {\bibinfo
  {volume} {289}},\ \bibinfo {pages} {597} (\bibinfo {year}
  {2009})}\BibitemShut {NoStop}%
\bibitem [{\citenamefont {Brandao}\ and\ \citenamefont
  {Cramer}(2015)}]{brandao2015equivalencestatisticalmechanicalensembles}%
  \BibitemOpen
  \bibfield  {author} {\bibinfo {author} {\bibfnamefont {F.~G. S.~L.}\
  \bibnamefont {Brandao}}\ and\ \bibinfo {author} {\bibfnamefont
  {M.}~\bibnamefont {Cramer}},\ }\href {https://arxiv.org/abs/1502.03263}
  {\bibinfo {title} {Equivalence of statistical mechanical ensembles for
  non-critical quantum systems}} (\bibinfo {year} {2015}),\ \Eprint
  {https://arxiv.org/abs/1502.03263} {arXiv:1502.03263} \BibitemShut {NoStop}%
\bibitem [{\citenamefont {Tasaki}(2018)}]{Tasaki_2018}%
  \BibitemOpen
  \bibfield  {author} {\bibinfo {author} {\bibfnamefont {H.}~\bibnamefont
  {Tasaki}},\ }\bibfield  {title} {\bibinfo {title} {On the local equivalence
  between the canonical and the microcanonical ensembles for quantum spin
  systems},\ }\href {https://doi.org/10.1007/s10955-018-2077-y} {\bibfield
  {journal} {\bibinfo  {journal} {Journal of Statistical Physics}\ }\textbf
  {\bibinfo {volume} {172}},\ \bibinfo {pages} {905–926} (\bibinfo {year}
  {2018})}\BibitemShut {NoStop}%
\bibitem [{\citenamefont {Kuwahara}\ and\ \citenamefont
  {Saito}(2020{\natexlab{b}})}]{KuwaharaETH}%
  \BibitemOpen
  \bibfield  {author} {\bibinfo {author} {\bibfnamefont {T.}~\bibnamefont
  {Kuwahara}}\ and\ \bibinfo {author} {\bibfnamefont {K.}~\bibnamefont
  {Saito}},\ }\bibfield  {title} {\bibinfo {title} {Eigenstate thermalization
  from the clustering property of correlation},\ }\href
  {https://doi.org/10.1103/PhysRevLett.124.200604} {\bibfield  {journal}
  {\bibinfo  {journal} {Phys. Rev. Lett.}\ }\textbf {\bibinfo {volume} {124}},\
  \bibinfo {pages} {200604} (\bibinfo {year} {2020}{\natexlab{b}})}\BibitemShut
  {NoStop}%
\bibitem [{\citenamefont {Bertoni}\ \emph {et~al.}(2025)\citenamefont
  {Bertoni}, \citenamefont {Wassner}, \citenamefont {Guarnieri},\ and\
  \citenamefont {Eisert}}]{Bertoni_2025}%
  \BibitemOpen
  \bibfield  {author} {\bibinfo {author} {\bibfnamefont {C.}~\bibnamefont
  {Bertoni}}, \bibinfo {author} {\bibfnamefont {C.}~\bibnamefont {Wassner}},
  \bibinfo {author} {\bibfnamefont {G.}~\bibnamefont {Guarnieri}},\ and\
  \bibinfo {author} {\bibfnamefont {J.}~\bibnamefont {Eisert}},\ }\bibfield
  {title} {\bibinfo {title} {Typical thermalization of low-entanglement
  states},\ }\bibfield  {journal} {\bibinfo  {journal} {Communications
  Physics}\ }\textbf {\bibinfo {volume} {8}},\ \href
  {https://doi.org/10.1038/s42005-025-02161-7} {10.1038/s42005-025-02161-7}
  (\bibinfo {year} {2025})\BibitemShut {NoStop}%
\bibitem [{\citenamefont {Farrelly}\ \emph {et~al.}(2017)\citenamefont
  {Farrelly}, \citenamefont {Brandao},\ and\ \citenamefont
  {Cramer}}]{Farrelly_2017}%
  \BibitemOpen
  \bibfield  {author} {\bibinfo {author} {\bibfnamefont {T.}~\bibnamefont
  {Farrelly}}, \bibinfo {author} {\bibfnamefont {F.~G.}\ \bibnamefont
  {Brandao}},\ and\ \bibinfo {author} {\bibfnamefont {M.}~\bibnamefont
  {Cramer}},\ }\bibfield  {title} {\bibinfo {title} {Thermalization and return
  to equilibrium on finite quantum lattice systems},\ }\bibfield  {journal}
  {\bibinfo  {journal} {Physical Review Letters}\ }\textbf {\bibinfo {volume}
  {118}},\ \href {https://doi.org/10.1103/physrevlett.118.140601}
  {10.1103/physrevlett.118.140601} (\bibinfo {year} {2017})\BibitemShut
  {NoStop}%
\bibitem [{\citenamefont {Hovhannisyan}\ \emph {et~al.}(2020)\citenamefont
  {Hovhannisyan}, \citenamefont {Barra},\ and\ \citenamefont
  {Imparato}}]{Hovhannisyan_2020}%
  \BibitemOpen
  \bibfield  {author} {\bibinfo {author} {\bibfnamefont {K.~V.}\ \bibnamefont
  {Hovhannisyan}}, \bibinfo {author} {\bibfnamefont {F.}~\bibnamefont
  {Barra}},\ and\ \bibinfo {author} {\bibfnamefont {A.}~\bibnamefont
  {Imparato}},\ }\bibfield  {title} {\bibinfo {title} {Charging assisted by
  thermalization},\ }\bibfield  {journal} {\bibinfo  {journal} {Physical Review
  Research}\ }\textbf {\bibinfo {volume} {2}},\ \href
  {https://doi.org/10.1103/physrevresearch.2.033413}
  {10.1103/physrevresearch.2.033413} (\bibinfo {year} {2020})\BibitemShut
  {NoStop}%
\bibitem [{\citenamefont {Alhambra}\ \emph {et~al.}(2020)\citenamefont
  {Alhambra}, \citenamefont {Anshu},\ and\ \citenamefont
  {Wilming}}]{ScarsProof}%
  \BibitemOpen
  \bibfield  {author} {\bibinfo {author} {\bibfnamefont {A.~M.}\ \bibnamefont
  {Alhambra}}, \bibinfo {author} {\bibfnamefont {A.}~\bibnamefont {Anshu}},\
  and\ \bibinfo {author} {\bibfnamefont {H.}~\bibnamefont {Wilming}},\
  }\bibfield  {title} {\bibinfo {title} {Revivals imply quantum many-body
  scars},\ }\href {https://doi.org/10.1103/PhysRevB.101.205107} {\bibfield
  {journal} {\bibinfo  {journal} {Phys. Rev. B}\ }\textbf {\bibinfo {volume}
  {101}},\ \bibinfo {pages} {205107} (\bibinfo {year} {2020})}\BibitemShut
  {NoStop}%
\bibitem [{\citenamefont {Schecter}\ and\ \citenamefont
  {Iadecola}(2019)}]{Schecter_2019}%
  \BibitemOpen
  \bibfield  {author} {\bibinfo {author} {\bibfnamefont {M.}~\bibnamefont
  {Schecter}}\ and\ \bibinfo {author} {\bibfnamefont {T.}~\bibnamefont
  {Iadecola}},\ }\bibfield  {title} {\bibinfo {title} {Weak ergodicity breaking
  and quantum many-body scars in spin-1 $xy$ magnets},\ }\bibfield  {journal}
  {\bibinfo  {journal} {Physical Review Letters}\ }\textbf {\bibinfo {volume}
  {123}},\ \href {https://doi.org/10.1103/physrevlett.123.147201}
  {10.1103/physrevlett.123.147201} (\bibinfo {year} {2019})\BibitemShut
  {NoStop}%
\bibitem [{\citenamefont {Hovhannisyan}\ \emph {et~al.}(2021)\citenamefont
  {Hovhannisyan}, \citenamefont {J\o{}rgensen}, \citenamefont {Landi},
  \citenamefont {Alhambra}, \citenamefont {Brask},\ and\ \citenamefont
  {Perarnau-Llobet}}]{Thermometry}%
  \BibitemOpen
  \bibfield  {author} {\bibinfo {author} {\bibfnamefont {K.~V.}\ \bibnamefont
  {Hovhannisyan}}, \bibinfo {author} {\bibfnamefont {M.~R.}\ \bibnamefont
  {J\o{}rgensen}}, \bibinfo {author} {\bibfnamefont {G.~T.}\ \bibnamefont
  {Landi}}, \bibinfo {author} {\bibfnamefont {A.~M.}\ \bibnamefont {Alhambra}},
  \bibinfo {author} {\bibfnamefont {J.~B.}\ \bibnamefont {Brask}},\ and\
  \bibinfo {author} {\bibfnamefont {M.}~\bibnamefont {Perarnau-Llobet}},\
  }\bibfield  {title} {\bibinfo {title} {Optimal quantum thermometry with
  coarse-grained measurements},\ }\href
  {https://doi.org/10.1103/PRXQuantum.2.020322} {\bibfield  {journal} {\bibinfo
   {journal} {PRX Quantum}\ }\textbf {\bibinfo {volume} {2}},\ \bibinfo {pages}
  {020322} (\bibinfo {year} {2021})}\BibitemShut {NoStop}%
\bibitem [{\citenamefont {Rai}\ \emph {et~al.}(2024)\citenamefont {Rai},
  \citenamefont {Cirac},\ and\ \citenamefont {Alhambra}}]{Rai_2024}%
  \BibitemOpen
  \bibfield  {author} {\bibinfo {author} {\bibfnamefont {K.~S.}\ \bibnamefont
  {Rai}}, \bibinfo {author} {\bibfnamefont {J.~I.}\ \bibnamefont {Cirac}},\
  and\ \bibinfo {author} {\bibfnamefont {A.~M.}\ \bibnamefont {Alhambra}},\
  }\bibfield  {title} {\bibinfo {title} {Matrix product state approximations to
  quantum states of low energy variance},\ }\href
  {https://doi.org/10.22331/q-2024-07-10-1401} {\bibfield  {journal} {\bibinfo
  {journal} {Quantum}\ }\textbf {\bibinfo {volume} {8}},\ \bibinfo {pages}
  {1401} (\bibinfo {year} {2024})}\BibitemShut {NoStop}%
\bibitem [{\citenamefont {Lu}\ \emph {et~al.}(2021)\citenamefont {Lu},
  \citenamefont {Ba\~nuls},\ and\ \citenamefont {Cirac}}]{LuQuantum}%
  \BibitemOpen
  \bibfield  {author} {\bibinfo {author} {\bibfnamefont {S.}~\bibnamefont
  {Lu}}, \bibinfo {author} {\bibfnamefont {M.~C.}\ \bibnamefont {Ba\~nuls}},\
  and\ \bibinfo {author} {\bibfnamefont {J.~I.}\ \bibnamefont {Cirac}},\
  }\bibfield  {title} {\bibinfo {title} {Algorithms for quantum simulation at
  finite energies},\ }\href {https://doi.org/10.1103/PRXQuantum.2.020321}
  {\bibfield  {journal} {\bibinfo  {journal} {PRX Quantum}\ }\textbf {\bibinfo
  {volume} {2}},\ \bibinfo {pages} {020321} (\bibinfo {year}
  {2021})}\BibitemShut {NoStop}%
\bibitem [{\citenamefont {Schuckert}\ \emph {et~al.}(2023)\citenamefont
  {Schuckert}, \citenamefont {Bohrdt}, \citenamefont {Crane},\ and\
  \citenamefont {Knap}}]{Schuckert2023}%
  \BibitemOpen
  \bibfield  {author} {\bibinfo {author} {\bibfnamefont {A.}~\bibnamefont
  {Schuckert}}, \bibinfo {author} {\bibfnamefont {A.}~\bibnamefont {Bohrdt}},
  \bibinfo {author} {\bibfnamefont {E.}~\bibnamefont {Crane}},\ and\ \bibinfo
  {author} {\bibfnamefont {M.}~\bibnamefont {Knap}},\ }\bibfield  {title}
  {\bibinfo {title} {Probing finite-temperature observables in quantum
  simulators of spin systems with short-time dynamics},\ }\href
  {https://doi.org/10.1103/PhysRevB.107.L140410} {\bibfield  {journal}
  {\bibinfo  {journal} {Phys. Rev. B}\ }\textbf {\bibinfo {volume} {107}},\
  \bibinfo {pages} {L140410} (\bibinfo {year} {2023})}\BibitemShut {NoStop}%
\bibitem [{\citenamefont {Yang}\ \emph {et~al.}(2022)\citenamefont {Yang},
  \citenamefont {Cirac},\ and\ \citenamefont {Ba\~nuls}}]{YilunClassical}%
  \BibitemOpen
  \bibfield  {author} {\bibinfo {author} {\bibfnamefont {Y.}~\bibnamefont
  {Yang}}, \bibinfo {author} {\bibfnamefont {J.~I.}\ \bibnamefont {Cirac}},\
  and\ \bibinfo {author} {\bibfnamefont {M.~C.}\ \bibnamefont {Ba\~nuls}},\
  }\bibfield  {title} {\bibinfo {title} {Classical algorithms for many-body
  quantum systems at finite energies},\ }\href
  {https://doi.org/10.1103/PhysRevB.106.024307} {\bibfield  {journal} {\bibinfo
   {journal} {Phys. Rev. B}\ }\textbf {\bibinfo {volume} {106}},\ \bibinfo
  {pages} {024307} (\bibinfo {year} {2022})}\BibitemShut {NoStop}%
\bibitem [{\citenamefont {Irmejs}\ \emph {et~al.}(2024)\citenamefont {Irmejs},
  \citenamefont {Ba{\~{n}}uls},\ and\ \citenamefont {Cirac}}]{Irmejs2024}%
  \BibitemOpen
  \bibfield  {author} {\bibinfo {author} {\bibfnamefont {R.}~\bibnamefont
  {Irmejs}}, \bibinfo {author} {\bibfnamefont {M.~C.}\ \bibnamefont
  {Ba{\~{n}}uls}},\ and\ \bibinfo {author} {\bibfnamefont {J.~I.}\ \bibnamefont
  {Cirac}},\ }\bibfield  {title} {\bibinfo {title} {Efficient {Q}uantum
  {A}lgorithm for {F}iltering {P}roduct {S}tates},\ }\href
  {https://doi.org/10.22331/q-2024-06-27-1389} {\bibfield  {journal} {\bibinfo
  {journal} {{Quantum}}\ }\textbf {\bibinfo {volume} {8}},\ \bibinfo {pages}
  {1389} (\bibinfo {year} {2024})}\BibitemShut {NoStop}%
\bibitem [{\citenamefont {Huang}(2024)}]{Huang_2024}%
  \BibitemOpen
  \bibfield  {author} {\bibinfo {author} {\bibfnamefont {Y.}~\bibnamefont
  {Huang}},\ }\bibfield  {title} {\bibinfo {title} {Deviation from maximal
  entanglement for mid-spectrum eigenstates of local hamiltonians},\ }\href
  {https://doi.org/10.1109/jsait.2024.3487856} {\bibfield  {journal} {\bibinfo
  {journal} {IEEE Journal on Selected Areas in Information Theory}\ }\textbf
  {\bibinfo {volume} {5}},\ \bibinfo {pages} {694–701} (\bibinfo {year}
  {2024})}\BibitemShut {NoStop}%
\bibitem [{\citenamefont {Fannes}\ \emph {et~al.}(1992)\citenamefont {Fannes},
  \citenamefont {Nachtergaele},\ and\ \citenamefont {Werner}}]{Fannes1992}%
  \BibitemOpen
  \bibfield  {author} {\bibinfo {author} {\bibfnamefont {M.}~\bibnamefont
  {Fannes}}, \bibinfo {author} {\bibfnamefont {B.}~\bibnamefont
  {Nachtergaele}},\ and\ \bibinfo {author} {\bibfnamefont {R.~F.}\ \bibnamefont
  {Werner}},\ }\bibfield  {title} {\bibinfo {title} {Finitely correlated states
  on quantum spin chains},\ }\href {https://doi.org/10.1007/BF02099178}
  {\bibfield  {journal} {\bibinfo  {journal} {Communications in Mathematical
  Physics}\ }\textbf {\bibinfo {volume} {144}},\ \bibinfo {pages} {443}
  (\bibinfo {year} {1992})}\BibitemShut {NoStop}%
\bibitem [{\citenamefont {Perez-Garcia}\ \emph {et~al.}(2007)\citenamefont
  {Perez-Garcia}, \citenamefont {Verstraete}, \citenamefont {Wolf},\ and\
  \citenamefont {Cirac}}]{perezgarcia2007matrixproductstaterepresentations}%
  \BibitemOpen
  \bibfield  {author} {\bibinfo {author} {\bibfnamefont {D.}~\bibnamefont
  {Perez-Garcia}}, \bibinfo {author} {\bibfnamefont {F.}~\bibnamefont
  {Verstraete}}, \bibinfo {author} {\bibfnamefont {M.~M.}\ \bibnamefont
  {Wolf}},\ and\ \bibinfo {author} {\bibfnamefont {J.~I.}\ \bibnamefont
  {Cirac}},\ }\href {https://arxiv.org/abs/quant-ph/0608197} {\bibinfo {title}
  {Matrix product state representations}} (\bibinfo {year} {2007}),\ \Eprint
  {https://arxiv.org/abs/quant-ph/0608197} {arXiv:quant-ph/0608197 [quant-ph]}
  \BibitemShut {NoStop}%
\bibitem [{\citenamefont {Araki}(1969)}]{araki1969gibbs}%
  \BibitemOpen
  \bibfield  {author} {\bibinfo {author} {\bibfnamefont {H.}~\bibnamefont
  {Araki}},\ }\bibfield  {title} {\bibinfo {title} {Gibbs states of a one
  dimensional quantum lattice},\ }\href
  {https://doi.org/https://doi.org/10.1007/BF01645134} {\bibfield  {journal}
  {\bibinfo  {journal} {Communications in Mathematical Physics}\ }\textbf
  {\bibinfo {volume} {14}},\ \bibinfo {pages} {120} (\bibinfo {year}
  {1969})}\BibitemShut {NoStop}%
\bibitem [{\citenamefont {Bluhm}\ \emph {et~al.}(2022)\citenamefont {Bluhm},
  \citenamefont {Capel},\ and\ \citenamefont
  {P\'erez-Hern\'andez}}]{Bluhm_2022}%
  \BibitemOpen
  \bibfield  {author} {\bibinfo {author} {\bibfnamefont {A.}~\bibnamefont
  {Bluhm}}, \bibinfo {author} {\bibfnamefont {A.}~\bibnamefont {Capel}},\ and\
  \bibinfo {author} {\bibfnamefont {A.}~\bibnamefont {P\'erez-Hern\'andez}},\
  }\bibfield  {title} {\bibinfo {title} {Exponential decay of mutual
  information for {G}ibbs states of local hamiltonians},\ }\href
  {https://doi.org/10.22331/q-2022-02-10-650} {\bibfield  {journal} {\bibinfo
  {journal} {Quantum}\ }\textbf {\bibinfo {volume} {6}},\ \bibinfo {pages}
  {650} (\bibinfo {year} {2022})}\BibitemShut {NoStop}%
\bibitem [{\citenamefont {Bergamaschi}\ and\ \citenamefont
  {Chen}(2026)}]{bergamaschi2026fastmixingquantumspin}%
  \BibitemOpen
  \bibfield  {author} {\bibinfo {author} {\bibfnamefont {T.}~\bibnamefont
  {Bergamaschi}}\ and\ \bibinfo {author} {\bibfnamefont {C.-F.}\ \bibnamefont
  {Chen}},\ }\href {https://arxiv.org/abs/2510.08533} {\bibinfo {title} {Fast
  mixing of quantum spin chains at all temperatures}} (\bibinfo {year}
  {2026}),\ \Eprint {https://arxiv.org/abs/2510.08533} {arXiv:2510.08533
  [quant-ph]} \BibitemShut {NoStop}%
\bibitem [{\citenamefont {Kastoryano}\ and\ \citenamefont
  {Eisert}(2013)}]{Kastoryano_2013}%
  \BibitemOpen
  \bibfield  {author} {\bibinfo {author} {\bibfnamefont {M.~J.}\ \bibnamefont
  {Kastoryano}}\ and\ \bibinfo {author} {\bibfnamefont {J.}~\bibnamefont
  {Eisert}},\ }\bibfield  {title} {\bibinfo {title} {Rapid mixing implies
  exponential decay of correlations},\ }\bibfield  {journal} {\bibinfo
  {journal} {Journal of Mathematical Physics}\ }\textbf {\bibinfo {volume}
  {54}},\ \href {https://doi.org/10.1063/1.4822481} {10.1063/1.4822481}
  (\bibinfo {year} {2013})\BibitemShut {NoStop}%
\bibitem [{\citenamefont {Roon}\ and\ \citenamefont {Sims}(2024)}]{Roon_2024}%
  \BibitemOpen
  \bibfield  {author} {\bibinfo {author} {\bibfnamefont {E.~B.}\ \bibnamefont
  {Roon}}\ and\ \bibinfo {author} {\bibfnamefont {R.}~\bibnamefont {Sims}},\
  }\bibfield  {title} {\bibinfo {title} {On quasi-locality and decay of
  correlations for long-range models of open quantum spin systems},\ }\href
  {https://doi.org/10.1088/1751-8121/ad8609} {\bibfield  {journal} {\bibinfo
  {journal} {Journal of Physics A: Mathematical and Theoretical}\ }\textbf
  {\bibinfo {volume} {57}},\ \bibinfo {pages} {445206} (\bibinfo {year}
  {2024})}\BibitemShut {NoStop}%
\bibitem [{\citenamefont {Kimura}\ and\ \citenamefont
  {Kuwahara}(2025)}]{kimura2025clustering}%
  \BibitemOpen
  \bibfield  {author} {\bibinfo {author} {\bibfnamefont {Y.}~\bibnamefont
  {Kimura}}\ and\ \bibinfo {author} {\bibfnamefont {T.}~\bibnamefont
  {Kuwahara}},\ }\bibfield  {title} {\bibinfo {title} {Clustering theorem in 1d
  long-range interacting systems at arbitrary temperatures},\ }\href
  {https://doi.org/https://doi.org/10.1007/s00220-025-05242-4} {\bibfield
  {journal} {\bibinfo  {journal} {Communications in Mathematical Physics}\
  }\textbf {\bibinfo {volume} {406}},\ \bibinfo {pages} {1} (\bibinfo {year}
  {2025})}\BibitemShut {NoStop}%
\bibitem [{\citenamefont {Kim}\ \emph {et~al.}(2025)\citenamefont {Kim},
  \citenamefont {Kuwahara},\ and\ \citenamefont {Saito}}]{Kim_2025}%
  \BibitemOpen
  \bibfield  {author} {\bibinfo {author} {\bibfnamefont {D.}~\bibnamefont
  {Kim}}, \bibinfo {author} {\bibfnamefont {T.}~\bibnamefont {Kuwahara}},\ and\
  \bibinfo {author} {\bibfnamefont {K.}~\bibnamefont {Saito}},\ }\bibfield
  {title} {\bibinfo {title} {Thermal area law in long-range interacting
  systems},\ }\bibfield  {journal} {\bibinfo  {journal} {Physical Review
  Letters}\ }\textbf {\bibinfo {volume} {134}},\ \href
  {https://doi.org/10.1103/physrevlett.134.020402}
  {10.1103/physrevlett.134.020402} (\bibinfo {year} {2025})\BibitemShut
  {NoStop}%
\bibitem [{\citenamefont {M\"obus}\ \emph {et~al.}(2026)\citenamefont
  {M\"obus}, \citenamefont {S\'anchez-Segovia}, \citenamefont {Alhambra},\ and\
  \citenamefont {Capel}}]{mobusLongRange}%
  \BibitemOpen
  \bibfield  {author} {\bibinfo {author} {\bibfnamefont {T.}~\bibnamefont
  {M\"obus}}, \bibinfo {author} {\bibfnamefont {J.}~\bibnamefont
  {S\'anchez-Segovia}}, \bibinfo {author} {\bibfnamefont {A.~M.}\ \bibnamefont
  {Alhambra}},\ and\ \bibinfo {author} {\bibfnamefont {A.}~\bibnamefont
  {Capel}},\ }\bibfield  {title} {\bibinfo {title} {Stability of thermal
  equilibrium in long-range quantum systems},\ }\href
  {https://doi.org/10.1103/3wbz-lq3p} {\bibfield  {journal} {\bibinfo
  {journal} {Phys. Rev. Res.}\ }\textbf {\bibinfo {volume} {8}},\ \bibinfo
  {pages} {013198} (\bibinfo {year} {2026})}\BibitemShut {NoStop}%
\bibitem [{\citenamefont {Esseen}(1945)}]{esseen1945fourier}%
  \BibitemOpen
  \bibfield  {author} {\bibinfo {author} {\bibfnamefont {C.-G.}\ \bibnamefont
  {Esseen}},\ }\bibfield  {title} {\bibinfo {title} {{F}ourier analysis of
  distribution functions. {A mathematical study of the Laplace-Gaussian law}},\
  }\href {https://doi.org/10.1007/BF02392223} {\bibfield  {journal} {\bibinfo
  {journal} {Acta Mathematica}\ }\textbf {\bibinfo {volume} {77}},\ \bibinfo
  {pages} {1} (\bibinfo {year} {1945})}\BibitemShut {NoStop}%
\bibitem [{\citenamefont {Kolassa}(1997)}]{Kolassa1997}%
  \BibitemOpen
  \bibfield  {author} {\bibinfo {author} {\bibfnamefont {J.~E.}\ \bibnamefont
  {Kolassa}},\ }\bibinfo {title} {Characteristic functions and the
  {B}erry-{E}sseen theorem},\ in\ \href
  {https://doi.org/10.1007/978-1-4757-4277-0_2} {\emph {\bibinfo {booktitle}
  {Series Approximation Methods in Statistics}}}\ (\bibinfo  {publisher}
  {Springer New York},\ \bibinfo {address} {New York, NY},\ \bibinfo {year}
  {1997})\ pp.\ \bibinfo {pages} {5--24}\BibitemShut {NoStop}%
\bibitem [{\citenamefont {Sanchez-Segovia}\ \emph {et~al.}(2025)\citenamefont
  {Sanchez-Segovia}, \citenamefont {T.~Schneider},\ and\ \citenamefont
  {M.~Alhambra}}]{Sanchez_Segovia_2025}%
  \BibitemOpen
  \bibfield  {author} {\bibinfo {author} {\bibfnamefont {J.}~\bibnamefont
  {Sanchez-Segovia}}, \bibinfo {author} {\bibfnamefont {J.}~\bibnamefont
  {T.~Schneider}},\ and\ \bibinfo {author} {\bibfnamefont {A.}~\bibnamefont
  {M.~Alhambra}},\ }\bibfield  {title} {\bibinfo {title} {High-temperature
  partition functions and classical simulatability of long-range quantum
  systems},\ }\bibfield  {journal} {\bibinfo  {journal} {PRX Quantum}\ }\textbf
  {\bibinfo {volume} {6}},\ \href {https://doi.org/10.1103/cww3-j1vd}
  {10.1103/cww3-j1vd} (\bibinfo {year} {2025})\BibitemShut {NoStop}%
\bibitem [{\citenamefont {Gorin}\ \emph {et~al.}(2006)\citenamefont {Gorin},
  \citenamefont {Prosen}, \citenamefont {Seligman},\ and\ \citenamefont
  {Žnidarič}}]{Gorin2006}%
  \BibitemOpen
  \bibfield  {author} {\bibinfo {author} {\bibfnamefont {T.}~\bibnamefont
  {Gorin}}, \bibinfo {author} {\bibfnamefont {T.}~\bibnamefont {Prosen}},
  \bibinfo {author} {\bibfnamefont {T.~H.}\ \bibnamefont {Seligman}},\ and\
  \bibinfo {author} {\bibfnamefont {M.}~\bibnamefont {Žnidarič}},\ }\bibfield
   {title} {\bibinfo {title} {Dynamics of {L}oschmidt echoes and fidelity
  decay},\ }\href
  {https://doi.org/https://doi.org/10.1016/j.physrep.2006.09.003} {\bibfield
  {journal} {\bibinfo  {journal} {Physics Reports}\ }\textbf {\bibinfo {volume}
  {435}},\ \bibinfo {pages} {33} (\bibinfo {year} {2006})}\BibitemShut
  {NoStop}%
\bibitem [{\citenamefont {Goussev}\ \emph {et~al.}(2012)\citenamefont
  {Goussev}, \citenamefont {Jalabert}, \citenamefont {Pastawski},\ and\
  \citenamefont {Wisniacki}}]{Wisniacki_2012}%
  \BibitemOpen
  \bibfield  {author} {\bibinfo {author} {\bibfnamefont {A.}~\bibnamefont
  {Goussev}}, \bibinfo {author} {\bibfnamefont {R.~A.}\ \bibnamefont
  {Jalabert}}, \bibinfo {author} {\bibfnamefont {H.~M.}\ \bibnamefont
  {Pastawski}},\ and\ \bibinfo {author} {\bibfnamefont {D.~A.}\ \bibnamefont
  {Wisniacki}},\ }\bibfield  {title} {\bibinfo {title} {{L}oschmidt echo},\
  }\href {https://doi.org/10.4249/scholarpedia.11687} {\bibfield  {journal}
  {\bibinfo  {journal} {Scholarpedia}\ }\textbf {\bibinfo {volume} {7}},\
  \bibinfo {pages} {11687} (\bibinfo {year} {2012})}\BibitemShut {NoStop}%
\bibitem [{\citenamefont {Izrailev}\ and\ \citenamefont
  {Castañeda-Mendoza}(2006)}]{Izrailev_2006}%
  \BibitemOpen
  \bibfield  {author} {\bibinfo {author} {\bibfnamefont {F.}~\bibnamefont
  {Izrailev}}\ and\ \bibinfo {author} {\bibfnamefont {A.}~\bibnamefont
  {Castañeda-Mendoza}},\ }\bibfield  {title} {\bibinfo {title} {Return
  probability: Exponential versus {G}aussian decay},\ }\href
  {https://doi.org/10.1016/j.physleta.2005.10.077} {\bibfield  {journal}
  {\bibinfo  {journal} {Physics Letters A}\ }\textbf {\bibinfo {volume}
  {350}},\ \bibinfo {pages} {355–362} (\bibinfo {year} {2006})}\BibitemShut
  {NoStop}%
\bibitem [{\citenamefont {Torres-Herrera}\ and\ \citenamefont
  {Santos}(2014)}]{TorresGaussian}%
  \BibitemOpen
  \bibfield  {author} {\bibinfo {author} {\bibfnamefont {E.~J.}\ \bibnamefont
  {Torres-Herrera}}\ and\ \bibinfo {author} {\bibfnamefont {L.~F.}\
  \bibnamefont {Santos}},\ }\bibfield  {title} {\bibinfo {title} {Quench
  dynamics of isolated many-body quantum systems},\ }\href
  {https://doi.org/10.1103/PhysRevA.89.043620} {\bibfield  {journal} {\bibinfo
  {journal} {Phys. Rev. A}\ }\textbf {\bibinfo {volume} {89}},\ \bibinfo
  {pages} {043620} (\bibinfo {year} {2014})}\BibitemShut {NoStop}%
\bibitem [{\citenamefont {Haah}\ \emph {et~al.}(2024)\citenamefont {Haah},
  \citenamefont {Kothari},\ and\ \citenamefont {Tang}}]{Haah_2024}%
  \BibitemOpen
  \bibfield  {author} {\bibinfo {author} {\bibfnamefont {J.}~\bibnamefont
  {Haah}}, \bibinfo {author} {\bibfnamefont {R.}~\bibnamefont {Kothari}},\ and\
  \bibinfo {author} {\bibfnamefont {E.}~\bibnamefont {Tang}},\ }\bibfield
  {title} {\bibinfo {title} {Learning quantum {H}amiltonians from
  high-temperature gibbs states and real-time evolutions},\ }\href
  {https://doi.org/10.1038/s41567-023-02376-x} {\bibfield  {journal} {\bibinfo
  {journal} {Nature Physics}\ }\textbf {\bibinfo {volume} {20}},\ \bibinfo
  {pages} {1027–1031} (\bibinfo {year} {2024})}\BibitemShut {NoStop}%
\bibitem [{\citenamefont {Mann}\ and\ \citenamefont {Minko}(2024)}]{Mann_2024}%
  \BibitemOpen
  \bibfield  {author} {\bibinfo {author} {\bibfnamefont {R.~L.}\ \bibnamefont
  {Mann}}\ and\ \bibinfo {author} {\bibfnamefont {R.~M.}\ \bibnamefont
  {Minko}},\ }\bibfield  {title} {\bibinfo {title} {Algorithmic cluster
  expansions for quantum problems},\ }\bibfield  {journal} {\bibinfo  {journal}
  {PRX Quantum}\ }\textbf {\bibinfo {volume} {5}},\ \href
  {https://doi.org/10.1103/prxquantum.5.010305} {10.1103/prxquantum.5.010305}
  (\bibinfo {year} {2024})\BibitemShut {NoStop}%
\bibitem [{\citenamefont {Lindeberg}(1922)}]{Lindeberg1922}%
  \BibitemOpen
  \bibfield  {author} {\bibinfo {author} {\bibfnamefont {J.~W.}\ \bibnamefont
  {Lindeberg}},\ }\bibfield  {title} {\bibinfo {title} {Eine neue {H}erleitung
  des {E}xponentialgesetzes in der {W}ahrscheinlichkeitsrechnung},\ }\href
  {https://doi.org/10.1007/BF01494395} {\bibfield  {journal} {\bibinfo
  {journal} {Mathematische Zeitschrift}\ }\textbf {\bibinfo {volume} {15}},\
  \bibinfo {pages} {211} (\bibinfo {year} {1922})}\BibitemShut {NoStop}%
\bibitem [{\citenamefont {Cramer}\ and\ \citenamefont
  {Eisert}(2010)}]{Cramer_2010}%
  \BibitemOpen
  \bibfield  {author} {\bibinfo {author} {\bibfnamefont {M.}~\bibnamefont
  {Cramer}}\ and\ \bibinfo {author} {\bibfnamefont {J.}~\bibnamefont
  {Eisert}},\ }\bibfield  {title} {\bibinfo {title} {A quantum central limit
  theorem for non-equilibrium systems: exact local relaxation of correlated
  states},\ }\href {https://doi.org/10.1088/1367-2630/12/5/055020} {\bibfield
  {journal} {\bibinfo  {journal} {New Journal of Physics}\ }\textbf {\bibinfo
  {volume} {12}},\ \bibinfo {pages} {055020} (\bibinfo {year}
  {2010})}\BibitemShut {NoStop}%
\bibitem [{\citenamefont {Arous}\ \emph {et~al.}(2013)\citenamefont {Arous},
  \citenamefont {Kirkpatrick},\ and\ \citenamefont {Schlein}}]{Arous2013}%
  \BibitemOpen
  \bibfield  {author} {\bibinfo {author} {\bibfnamefont {G.~B.}\ \bibnamefont
  {Arous}}, \bibinfo {author} {\bibfnamefont {K.}~\bibnamefont {Kirkpatrick}},\
  and\ \bibinfo {author} {\bibfnamefont {B.}~\bibnamefont {Schlein}},\
  }\bibfield  {title} {\bibinfo {title} {A central limit theorem in many-body
  quantum dynamics},\ }\href {https://doi.org/10.1007/s00220-013-1722-1}
  {\bibfield  {journal} {\bibinfo  {journal} {Communications in Mathematical
  Physics}\ }\textbf {\bibinfo {volume} {321}},\ \bibinfo {pages} {371}
  (\bibinfo {year} {2013})}\BibitemShut {NoStop}%
\bibitem [{\citenamefont {Wilming}\ \emph {et~al.}(2018)\citenamefont
  {Wilming}, \citenamefont {de~Oliveira}, \citenamefont {Short},\ and\
  \citenamefont {Eisert}}]{Wilming_2018}%
  \BibitemOpen
  \bibfield  {author} {\bibinfo {author} {\bibfnamefont {H.}~\bibnamefont
  {Wilming}}, \bibinfo {author} {\bibfnamefont {T.~R.}\ \bibnamefont
  {de~Oliveira}}, \bibinfo {author} {\bibfnamefont {A.~J.}\ \bibnamefont
  {Short}},\ and\ \bibinfo {author} {\bibfnamefont {J.}~\bibnamefont
  {Eisert}},\ }\bibinfo {title} {Equilibration times in closed quantum
  many-body systems},\ in\ \href {https://doi.org/10.1007/978-3-319-99046-0_18}
  {\emph {\bibinfo {booktitle} {Thermodynamics in the Quantum Regime}}}\
  (\bibinfo  {publisher} {Springer International Publishing},\ \bibinfo {year}
  {2018})\ p.\ \bibinfo {pages} {435–455}\BibitemShut {NoStop}%
\end{thebibliography}%


\appendix

\section{The Esseen inequality}\label{app:esseenInequality}
	For completeness, in this section we present a proof of Esseen's inequality (see Eq.~\eqref{eq:esseenInequality}). The presentation closely follows~\cite{esseen1945fourier}. A similar proof can also be found in \cite[Theorem 2.5.2]{Kolassa1997}.
	
	\begin{theorem}[Esseen's inequality]
		Let $F(y)$ be a real, non-decreasing function, and let $G(y)$ be a real function of bounded variation. Suppose $F(-\infty)=G(-\infty)=0$ and $F(\infty)=G(\infty)$. Moreover, denote by $\hat{f}(\omega)$ and $\hat{g}(\omega)$ their corresponding Fourier transform, and assume that there exist three positive finite constants $A$, $\Omega$ and $\varepsilon$ such that:
		\begin{enumerate}
			\item for all $y$, it holds that $|G'(y)|\leq A$;
			\item the following integral is bounded:
			\begin{align}
				\varepsilon := \int_{0}^{\Omega} \de \omega\; \frac{|\hat{f}(\omega)-\hat{g}(\omega)|}{|\omega|} \,.
			\end{align}
		\end{enumerate}
		Then, for every $k>1$ there exists a finite constant $c(k)$ (which only depends on $k$) such that:
		\begin{align}
			\Delta = \sup_y \,|F(y)- G(y)| \leq  \frac{c(k)A}{\Omega}+\norbra{\frac{k}{k-1}}\frac{\varepsilon}{\pi} \,.
		\end{align}
	\end{theorem}
	We can apply this Theorem to prove Eq.~\eqref{eq:esseenInequality} since $F_{\widehat{H}}(y)$ is non-decreasing (being a cumulative function), $G(y)$ defined in Eq.~\eqref{eq:GaussianCumulative} is of bounded variation, and satisfies $|G'(y)|\leq (2\pi)^{-1/2}$. Then, setting $k=2$ and $C:= ((2\pi)^{-1/2}c(2))$ gives Eq.~\eqref{eq:esseenInequality}.
	
	\begin{proof}
		We begin by rewriting the difference $\hat{f}(\omega)-\hat{g}(\omega)$ as:
		\begin{align}
			\hat{f}(\omega)-\hat{g}(\omega) = \int_{-\infty}^\infty \de (F(y)-G(y))\; e^{i\omega y}  = -i\omega \int_{-\infty}^\infty \de y \;(F(y)-G(y)) e^{i\omega y}\,,
		\end{align}
		where in the last step we integrated by parts and used the fact that $F(-\infty)=G(-\infty)=0$ and $F(\infty)=G(\infty)$. Then, by Fourier inversion, we obtain that:
		\begin{align}
			F(y)-G(y)  =\frac{1}{2\pi} \int_{-\infty}^\infty \de\omega \;e^{-i\omega y}\, \frac{\hat{f}(\omega)-\hat{g}(\omega)}{-i\omega}\,.\label{eq:A4}
		\end{align}
		Then, if one could assume that the integrand in the right-hand side was absolutely integrable, the result would follow by setting $\Omega = \infty$. For this reason, the rest of the proof can be considered as an estimation of the error made by restricting Eq.~\eqref{eq:A4} to a finite interval.
		
		To this end, without loss of generality, consider $|F(y)-G(y)|$ to reach its maximum at zero (if this is not the case, one can traslate the origin to the point, which only affects the Fourier transform by a phase factor). Moreover, for the moment we assume that $A =\Omega= 1$. 
		Then, let us introduce a filter function $H(y)$ with Fourier transform $\hat{h}(\omega)$ satisfying the following conditions:
		\begin{enumerate}
			\item $H(y)$ and $\hat{h}(\omega)$ are real and non-negative;
			\item $\hat{h}(0)=1$ (that is, the integral of $H(y)$ is normalized);\label{it:2Ess}
			\item there exists a finite constant $b$ such that:\label{it:3Ess}
			\begin{align}
				\int_{-\infty}^\infty\de y\; |y| H(y) = b\,;
			\end{align}
			\item $\hat{h}(\omega)$ has support on $[-1,1]$ and it is bounded by $1$, i.e., $\hat{h}(\omega)\leq1$. \label{it:4Ess}
		\end{enumerate}
		We present here a choice of such a function, but the argument is independent of the specific choice. Define $\hat{k}(\omega)$ as:
		\begin{align}
			\hat{k}(\omega) := \begin{cases}
				1-|\omega|& {\text{if $\omega\in[-1,1]$}}\\
				0& {\text{otherwise}}
			\end{cases}\,,
		\end{align}
		and let $K(y)$ be its Fourier transform:
		\begin{align}
			K(y)=\frac{1}{2\pi}\int_{-\infty}^{\infty}\de \omega\;e^{-i\omega y} \hat{k}(\omega) = \frac{1}{2\pi} \norbra{\frac{\sin(y/2)}{(y/2)}}^2\,.
		\end{align}
		Then, we define:
		\begin{align}
			H(y) := \frac{2}{3\pi}\, K(y/2)^2\;;\qquad\qquad\qquad \hat{h}(\omega) =\frac{ \int_{-\infty}^{\infty}\de \tilde{\omega} \;\hat{k}(2\omega- \tilde{\omega})\hat{k}(\tilde{\omega})}{\int_{-\infty}^{\infty}\de \tilde{\omega} \;\hat{k}(\tilde{\omega})^2}\,,\label{eq:A8}
		\end{align}
		Due to the connection with $\hat{k}(\omega)$, one can verify the properties of $\hat{h}(\omega)$ by inspection, while the property~\ref{it:3Ess} of $H(y)$ follows from the fact that $(yK(y)^2)$ decays as $y^{-3}$ at infinity.
		
		Consider the convolution between $(F(y)-G(y) )$ and $H(y)$. This satisfies the property that, for arbitrary $x$:
		\begin{align}
			\int_{-\infty}^{\infty}\de y\; H(x-y)&(F(y)-G(y) )= \frac{1}{2\pi}\int_{-\infty}^{\infty}\de y \int_{-\infty}^\infty \de\omega  \;H(x-y)\,\norbra{e^{-i\omega y}\, \frac{\hat{f}(\omega)-\hat{g}(\omega)}{-i\omega}} =
			\\
			&\qquad\quad=  \frac{1}{2\pi}\int_{-\infty}^{\infty}\de \omega \;e^{-i\omega x}\,\hat{h}(\omega)\norbra{ \frac{\hat{f}(\omega)-\hat{g}(\omega)}{-i\omega}} = \frac{1}{2\pi}\int_{-1}^{1}\de \omega \;e^{-i\omega x}\,\hat{h}(\omega)\norbra{ \frac{\hat{f}(\omega)-\hat{g}(\omega)}{-i\omega}}\,,
		\end{align}
		where in the first step we used the identity in Eq.~\eqref{eq:A4}, then the fact that $\hat{h}(\omega)$ is the Fourier transform of $H(y)$, and finally the fact that $\hat{h}(\omega)$ has support only on $[-1,1]$ (see condition~\ref{it:4Ess}). The relation just obtained allows us to prove the upper-bound:
		\begin{align}
			\left|\int_{-\infty}^{\infty}\de y\; H(x-y)(F(y)-G(y) )\right| \leq  \frac{1}{2\pi}\int_{-1}^{1}\de \omega \;|\hat{h}(\omega)| \frac{\left|\hat{f}(\omega)-\hat{g}(\omega)\right|}{|\omega|} \leq \frac{1}{\pi}\int_{0}^{1}\de \omega \; \frac{\left|\hat{f}(\omega)-\hat{g}(\omega)\right|}{|\omega|} = \frac{\varepsilon}{\pi}\,,\label{eq:A11}
		\end{align}
		where in the second inequality we used the fact that $\hat{h}(\omega)\leq 1$ (condition~\ref{it:4Ess}), together with the fact that $\hat{f}(\omega)$ and $\hat{g}(\omega)$ are Fourier transform of real functions, which implies that $(\hat{f}(-\omega)-\hat{g}(-\omega))=\overline{(\hat{f}(\omega)-\hat{g}(\omega))}$.
		
		Suppose now that $F(0)\geq G(0)$ (the other case can be treated symmetrically). Then, by definition $F(0)-G(0)= \Delta$. Moreover, since $F(y)$ is non-decreasing and $G(y)$ has derivative bounded by $A=1$, in the interval $y\in[0,\Delta]$, it also holds that:
		\begin{align}
			F(y)-G(y)\geq \Delta - y\geq 0\,, \qquad\qquad{\text{for $y\in[0,\Delta]$}}\,.\label{eq:A12}
		\end{align}
		Then, we can also lower bound Eq.~\eqref{eq:A11} as:
		\begin{align}
			&\left|\norbra{\int_{-\infty}^{0}+\int^{\Delta}_{0}+\int_{\Delta}^{\infty}}\de y\; H(x-y)(F(y)-G(y) )\right|\geq
			\\
			&\qquad\qquad\geq \int^{\Delta}_{0}\de y\; H(x-y)(F(y)-G(y) ) - \norbra{\int_{-\infty}^{0}+\int_{\Delta}^{\infty}}\de y\; H(x-y)|F(y)-G(y)|\geq
			\\
			&\qquad\qquad\geq  \int^{\Delta}_{0}\de y\;H(x-y)(\Delta-y) - \norbra{\int_{-\infty}^{0}+\int_{\Delta}^{\infty}}\de y\; H(x-y)\Delta\,,\label{eq:A15}
		\end{align} 
		where in the last line we used the relation in Eq.~\eqref{eq:A12} together with the property $|F(y)-G(y)|\leq \Delta$. It should also be noticed that thanks to condition~\ref{it:2Ess}, we can rewrite Eq.~\eqref{eq:A15} as:
		\begin{align}
			\text{Eq.~\eqref{eq:A15}} &= \norbra{\int^{\Delta}_{0}\de y\;H(x-y)(2\Delta-y) }- \Delta = \norbra{\int^{\Delta-x}_{-x}\de y\;H(-y)(2\Delta -x-y) }- \Delta \geq
			\\
			&\geq \norbra{\int^{\Delta-x}_{-x}\de y\;H(-y)(2\Delta -x) }- \Delta - b\,,
		\end{align}
		where in the last step we used the assumption~\ref{it:3Ess}. At this point, introduce $m\in[0,1]$, such that $x= m\Delta$. Putting everything together, it holds that:
		\begin{align}
			\norbra{\int^{\Delta-x}_{-x}\de y\;H(-y)(2\Delta -x) } -\Delta = \Delta \norbra{\norbra{(2-m)\int^{\Delta (1-m)}_{-\Delta m}\de y\;H(-y)}-1}\leq \frac{\varepsilon}{\pi} +b\,.\label{eq:A18}
		\end{align}
		Let us focus in the expression of the integral in the second inequality. It should be noticed that thanks to the normalization of $H(y)$, for any $k>1$ we can find a $m(k)$ small enough and a $\alpha(k)$ large enough such that:
		\begin{align}
			\norbra{(2-m(k))\int^{\alpha(k) (1-m(k))}_{-\alpha(k) m(k)}\de y\;H(-y)}-1 = 1-\frac{1}{k} = \frac{k-1}{k}\,.
		\end{align}
		Then, for a fixed $k$, there are two possibilities: if $\Delta > \alpha(k)$, we have:
		\begin{align}
			\norbra{(2-m(k))\int^{\Delta (1-m(k))}_{-\Delta m(k)}\de y\;H(-y)}-1 \geq \norbra{(2-m(k))\int^{\alpha(k) (1-m(k))}_{-\alpha(k) m(k)}\de y\;H(-y)}-1 = \frac{1}{k}\,,
		\end{align}
		where we implicitly used the positivity of $H(y)$. Plugging this inequality in Eq.~\eqref{eq:A18}, and rearranging, we obtain:
		\begin{align}
			\Delta \norbra{\frac{k-1}{k}}\leq  \frac{\varepsilon}{\pi} +b \qquad\implies\qquad\Delta\leq \norbra{\frac{k}{k-1}}\norbra{\frac{\varepsilon}{\pi} +b}\,.
		\end{align}
		On the other hand, if $\Delta \leq \alpha(k)$, we can use a trivial upper-bound. In this way, we have:
		\begin{align}
			\Delta \leq \max\left\{ \norbra{\frac{k}{k-1}}\norbra{\frac{\varepsilon}{\pi} +b}, \,\alpha(k) \right\}\leq \norbra{\frac{k}{k-1}}\frac{\varepsilon}{\pi} +\norbra{b\norbra{\frac{k}{k-1}}+\alpha(k)}\,.\label{eq:A22}
		\end{align}
		This proves the claim for $A=\Omega = 1$ with $c(k):= \norbra{bk/(k-1)+\alpha(k)}$.
        
		In order to prove the general case, let us introduce the two rescaled functions:
		\begin{align}
			F_1(y):= \frac{\Omega}{A}F\norbra{\frac{y}{\Omega}}\,;\qquad\qquad\qquad	G_1(y):= \frac{\Omega}{A}G\norbra{\frac{y}{\Omega}}
		\end{align}
		with Fourier transform $\hat{f}_1(\omega)$ and $\hat{g}_1(\omega)$ which satisfy:
		\begin{align}
			\int_{0}^{1} \de \omega\; \frac{|\hat{f}_1(\omega)-\hat{g}_1(\omega)|}{|\omega|} = \frac{\Omega\,\varepsilon}{A}\,,
		\end{align}
		and $|G_1'(y)|\leq 1$. Then, we can apply Eq.~\eqref{eq:A22}, which gives us:
		\begin{align}
			|F_1(y)-G_1(y)| \leq c(k)+\frac{k\,\Omega\,\varepsilon}{A\pi}\qquad\implies\qquad|F(y)-G(y)| \leq \frac{c(k)A}{\Omega}+\frac{k\varepsilon}{\pi}\,.
		\end{align}
		This proves the claim for arbitrary $A$ and $\Omega$.
	\end{proof}


\section{Bounds on the individual error terms} \label{sec:bounds}

\subsection{Bounding $\eta_{1,j}(\omega,K)$}\label{app:eta1}
	Let us begin by rewriting the term inside of the parenthesis as:
	\begin{align}
		\xi_j^1(\omega) &= \norbra{e^{i\omega \widehat{H}_j^\ell(1)}R^{M}_{1,j}(\omega)-\idO} = \int_{0}^{\omega}\de\omega_1 \;e^{i\omega_1 \widehat{H}_j^\ell(1)}\norbra{i\,\widehat{H}_j^\ell(1) R^{M}_{1,j}(\omega_1) +(\partial_{\tilde{\omega}}R^{M}_{1,j}(\tilde{\omega}))\big|_{\tilde{\omega}=\omega_1}} = 
		\\
		&=i\omega \widehat{H}_j^\ell(1)+   \int_{0}^{\omega}\de\omega_1 \;\int_{0}^{\omega_1}\de\omega_2 \;\partial_{\tilde{\omega}}\norbra{e^{i\hat{\omega} \widehat{H}_j^\ell(1)}\norbra{i\,\widehat{H}_j^\ell(1) R^{M}_{1,j}(\hat{\omega}) +(\partial_{\tilde{\omega}}R^{M}_{1,j}(\tilde{\omega}))\big|_{\tilde{\omega}=\hat{\omega}}}}\big|_{\hat{\omega}=\omega_2}\,,
	\end{align}
	where we implicitly used the fact that $R^{M}_{1,j}(0)=\idO$ and $(\partial_{\tilde{\omega}}R^{M}_{1,j}(\tilde{\omega}))\big|_{\tilde{\omega}=0} =0$ (see Lemma~\ref{lemma:clusterExp}). Then, we
	can express $\eta_{1,j}(\omega)$ as:
	\begin{align}
		&\eta_{1,j}(\omega,K) =  \int_{0}^{\omega}\de\omega_1 \;\int_{0}^{\omega_1}\de\omega_2 \;\average{h_j\partial_{\tilde{\omega}}\norbra{e^{i\hat{\omega} \widehat{H}_j^\ell(1)}\norbra{i\,\widehat{H}_j^\ell(1) R^{M}_{1,j}(\hat{\omega}) +(\partial_{\tilde{\omega}}R^{M}_{1,j}(\tilde{\omega}))\big|_{\tilde{\omega}=\hat{\omega}}}}\big|_{\hat{\omega}=\omega_2}}{\rho} = 
		\\
		&=\int_{0}^{\omega}\de\omega_1 \int_{0}^{\omega_1}\de\omega_2 \average{h_j\norbra{e^{i\omega_2 \widehat{H}_j^\ell(1)}\norbra{- (\widehat{H}_j^\ell(1))^2 R^{M}_{1,j}(\hat{\omega}) +2i\,\widehat{H}_j^\ell(1) (\partial_{\tilde{\omega}}R^{M}_{1,j}(\tilde{\omega}))\big|_{\tilde{\omega}=\omega_2}+(\partial^{(2)}_{\tilde{\omega}}R^{M}_{1,j}(\tilde{\omega}))\big|_{\tilde{\omega}=\omega_2}}}}{\rho}\,.
	\end{align}
	This expression allows us to upper bound the absolute value of $\eta_{1,j}(\omega)$  as:
	\begin{align}
		|\eta_{1,j}(\omega)| \leq E\int_{0}^{\omega}\de\omega_1 \;\int_{0}^{\omega_1}\de\omega_2 \;\norbra{\|\widehat{H}_j^\ell(1)\|^2 \|R^{M}_{1,j}(\hat{\omega})\| +2\|\widehat{H}_j^\ell(1)\| \|(\partial_{\tilde{\omega}}R^{M}_{1,j}(\tilde{\omega}))\big|_{\tilde{\omega}=\omega_2}\|+\|(\partial^{(2)}_{\tilde{\omega}}R^{M}_{1,j}(\tilde{\omega}))\big|_{\tilde{\omega}=\omega_2}\|}\,.
	\end{align}
	Thanks to Lemma~\ref{lemma:clusterExp}, we can upper bound all the terms depending on $R^{M}_{1,j}({\omega})$. Moreover, using Eq.~\eqref{eq:dimDef}, we can estimate the norm of $\widehat{H}_j^\ell(1)$ as:
	\begin{align}
		\|\widehat{H}_j^\ell(1)\|\leq \frac{c_DE }{\sigma_H}\;\sum_{r=0}^{2R\ell}\,r^{D-1}\leq \frac{c_DE (2R\ell)^D }{\sigma_H}\,,
	\end{align}
	Putting everything together we obtain:
	\begin{align}
		|\eta_{1,j}(\omega)| &\leq  \frac{c_D^2(2R)^{2D}E^3\,\ell^D}{\sigma_H^2}\omega^2\norbra{\ell^{D}+2\, \norbra{\frac{\Gamma}{c_D (2R)^D}}\ell^{\frac{D}{2}}+2\norbra{\frac{\Gamma}{c_D (2R)^D}}^2} \leq
		\\
		&\leq\frac{c_D^2(2R)^{2D}E^3\,\ell^{2D}}{\sigma_H^2}\omega^2\norbra{1+\sqrt{2}\norbra{\frac{\Gamma}{\ell^{\frac{D}{2}}c_D (2R)^D}} }^2 \leq \frac{B_1\,E\,\ell^{2D}}{\sigma_H^2}(E\omega)^2\,,\label{eq:b1Def}
	\end{align}
	where we implicitly defined $B_1 := \norbra{c_D(2R)^{D}+\sqrt{2}\,\Gamma}^2$.
	
	\subsection{Bounding $\eta_{2,j}(\omega,K)$}\label{app:eta2}
	The supports of $\widehat{h}_i$ and $\widehat{h}_j$ are at least at a distance $d(i,j)-2R$. Moreover, using the fact that $\langle\widehat{h}_j\rangle_{\rho}=0$ and the decay of correlations (see Eq.~\eqref{eq:decayOfCorrelations}), it follows that:
	\begin{align}
		|\eta_{2,j}(\omega,K)| &=\omega\,\left|\average{\widehat{h}_jz_j^\ell(1)}{\rho}\right|\leq\frac{\omega}{\sigma_H}\,\sum_{\substack{i\in\mathcal{X}
				\\
				d(i,j)> 2 R \ell }}\;  \left|\average{\widehat{h}_j\widehat{h}_i}{\rho}\right|\leq
		\\
		&\leq \frac{\omega}{\sigma_H}\,\sum_{\substack{i\in\mathcal{X}
				\\
				d(i,j)> 2 R \ell }}\;  R\,\alpha(d(i,j)-2R)E^2 \leq	\frac{c_D\, E^2R}{\sigma_H}\,\norbra{\sum_{r=2 R \ell+1}^\infty\; \alpha(r-2R)\, r^{D-1}}\,\omega\,,\label{eq:93}
	\end{align}
	where we used the fact that $|\supp(\widehat{h}_i)|\leq R$, and implicitly assumed that $2R\ell-2R>0$ (that is $\ell>1$).
	
	\subsection{Bounding $\eta_{3,j}(\omega,K)$}\label{app:eta3}
	In order to bound $\eta_{3,j}(\omega,K)$, let us first focus on averages of the type $\average{\widehat{h}_j\, \xi_j^1(\omega)\dots\xi_j^m(\omega)}{\rho}$. First, repeating the step of taking a derivative and integrating again, we obtain:
	\begin{align}
		\xi_j^k(\omega) = \norbra{e^{i\omega \widehat{H}_j^\ell(k)}R^{M}_{k,j}(\omega)-\idO} = \int_{0}^{\omega}\de\omega_1 \;e^{i\omega_1 \widehat{H}_j^\ell(k)}\norbra{i\,\widehat{H}_j^\ell(k) R^{M}_{k,j}(\omega_1) +(\partial_{\tilde{\omega}}R^{M}_{k,j}(\tilde{\omega}))\big|_{\tilde{\omega}=\omega_1}}\,.\label{eq:70}
	\end{align}
	Thanks to this expression, in the range $\in [0,\,\Omega_1]$ (and for $k\leq K$), we obtain:
	\begin{align}
		\|\xi_j^k(\omega) \| &\leq \int_{0}^{\omega}\de\omega_1 \;\norbra{\|\widehat{H}_j^\ell(k) R^{M}_{k,j}(\omega_1)\| +\|(\partial_{\tilde{\omega}}R^{M}_{k,j}(\tilde{\omega}))\big|_{\tilde{\omega}=\omega_1}\|} \leq 
		\\
		&\leq2\norbra{\frac{c_D  E}{\sigma_H}\sum_{r=2R\ell(k-1)}^{2R\ell k}\,r^{D-1}}\omega +  \norbra{\frac{2\,E\,\Gamma \ell^{\frac{D}{2}} k^{\frac{D-1}{2}}}{\sigma_H} }\omega\leq  \norbra{c_D(2R)^D + \frac{\Gamma}{\ell^{\frac{D}{2}}k^{\frac{D-1}{2}}}}\norbra{\frac{2E\,\ell^D\,k^{D-1}}{\sigma_H}}\omega\leq
		\\
		&\leq\norbra{\frac{B_2\,\ell^D\,k^{D-1}}{\sigma_H}}(E\omega)\label{eq:73}
	\end{align}
	where in the second inequality we applied Lemma~\ref{lemma:clusterExp}, then used Eq.~\eqref{eq:dimDef} to upper bound the norm of $\widehat{H}_j^\ell(k)$, we introduced the constant $B_2 := 2\norbra{c_D(2R)^D + \Gamma}$. Applying this relation to each of the $\xi_j^k(\omega)$, we obtain:
	\begin{align}
		\left|\average{\widehat{h}_j\, \xi_j^1(\omega)\dots\xi_j^m(\omega)}{\rho}\right| &\leq E (m!)^{D-1}\norbra{\frac{B_2\,\ell^D}{\sigma_H}}^m(E\omega)^m\,,\label{eq:B26}
	\end{align}
	Similarly, it can also be shown that:
	\begin{align}
		&\left|\average{\idO+\gamma_j^{k}(\omega)}{\rho}\right| = \left|\average{e^{-i \omega(\widehat{H}-z^{\ell}_j(k))}S_{k,j}^M(\omega)}{\rho}\right|\leq \|S_{k,j}^M(\omega)\|\leq 2\,,
	\end{align}
	and that:
	\begin{align}
		&\left|\average{\gamma_j^{2}(\omega)}{\rho}\right|\leq \int_{0}^{\omega}\de\omega_1 \;\norbra{\|\widehat{H}_j^\ell(1)S_{2,j}^M(\omega)\|+\|(\partial_{\tilde{\omega}}S^{M}_{2,j}(\tilde{\omega}))\big|_{\tilde{\omega}=\omega_1}\|}\leq 
		\\
		&\;\;\leq \frac{2c_D\norbra{\sum_{r=0}^{2R\ell}\,r^{D-1}}E}{\sigma_H}\,\omega + \norbra{\frac{2\,E\,\Gamma \ell^{\frac{D}{2}} 2^{\frac{D-1}{2}}}{\sigma_H} }\omega\leq  \norbra{2c_D (2R)^{D}+\frac{2^{\frac{D+1}{2}}\Gamma}{\ell^{\frac{D}{2}}}}\norbra{\frac{\ell^D }{\sigma_H}}(E\omega)\leq\norbra{\frac{B_3\ell^D }{\sigma_H}}(E\omega)\,,
	\end{align}
	where we set $B_3:= 2\norbra{c_D (2R)^{D}+2^{\frac{D-1}{2}}\Gamma}\geq B_2$.
	Then, wrapping everything up we obtain:
	\begin{align}
		|\eta_{3,j}(\omega,K)|&\leq \left|\average{\widehat{h}_j\xi_j^1(\omega)}{\rho}\average{\gamma_j^{2}(\omega)}{\rho}\right|+\sum_{m=2}^{K-1}\;\left|\average{\widehat{h}_j\, \xi_j^1(\omega)\dots\xi_j^m(\omega)}{\rho}\average{\idO+\gamma_j^{m+1}(\omega)}{\rho}\right|\leq
		\\
		&\leq E \norbra{\norbra{\frac{B_3^2\,\ell^{2D} }{\sigma_H^2}}(E\omega)^2 + 2\sum_{m=2}^{K-1}\; (m!)^{D-1}\norbra{\frac{B_2\,\ell^D}{\sigma_H}}^m(E\omega)^m} \leq
		\\
		&\leq E \norbra{\frac{B_3^2\,\ell^{2D} }{\sigma_H^2}}(E\omega)^2\norbra{1 + 2\sum_{m=0}^{K-3}\; ((m+2)!)^{D-1}\norbra{\frac{B_2\,\ell^D}{\sigma_H}}^m(E\omega)^{m}} \leq
		\\
		&\leq E \norbra{\frac{B_3^2\,\ell^{2D} }{\sigma_H^2}}(E\omega)^2\norbra{1 + 2\sum_{m=0}^{K-3}\; (m+2)^{2(D-1)}\norbra{\frac{B_2\,\ell^D K^{D-1}}{\sigma_H}}^m(E\omega)^{m}} \,.
	\end{align}
	where we repeatedly made use of the fact that $B_3\geq B_2$, and in the last step we applied the inequality $(m+2)! = (m+2)(m+1)(m!) \leq (m+2)^2 K^m$ (since $m\leq K-3$).
	
	\subsection{Bound on $\eta(\omega, K) $}\label{app:boundEta}
	
	In this section, we put together the results from Sec.~\ref{app:eta1},~\ref{app:eta2} and~\ref{app:eta3}. First, it should be noticed that multiplying the results of Sec.~\ref{app:eta1} and~\ref{app:eta3}  by $N/\sigma_H$ gives an estimate over the whole space. On the other hand, the term in Sec.~\ref{app:eta2} needs a bit more care. In this case by summing over $j$ we obtain: 
	\begin{align}
		|\eta_{2}(\omega,K)|&\leq\frac{1}{\sigma_H}\sum_{j\in\mathcal{X}}\;|\eta_{2,j}(\omega,K)|\leq\frac{c_DN\, E^2R}{\sigma_H^2}\,\norbra{\sum_{r=2 R \ell+1}^\infty\;  \alpha(r-2R)\, r^{D-1}}\,\omega =
		\\
		&=\frac{N\, E^2R(2R)^{D-1}}{\sigma_H^2}\,\norbra{c_D\sum_{r=2R(\ell-1)+1}^\infty\;  \alpha(r)\, \frac{(r+2R)^{D-1}}{(2R)^{D-1}}}\,\omega \leq\frac{ C_\alpha(2R(\ell-1))N\, E(2R)^{D}}{\sigma_H^2}\,(E\omega) \,,
	\end{align}
	where we used the definition of the function $C_\alpha(\ell)$ from Eq.~\eqref{eq:decayOfCorrConst}. Then, putting everything together, we can upper bound $\eta(\omega, K) $ as:
	\begin{align}
		|\eta(\omega, K)| &\leq\frac{N\,E}{\sigma_H^2}
		\Bigg(C_\alpha(2R(\ell-1))\, (2R)^{D}\,(E\omega)+\nonumber
		\\
		&\qquad\qquad\qquad\;\;\;\;\;+\frac{\ell^{2D}}{\sigma_H}(E\omega)^2\norbra{B_1+B_3^2\norbra{1 + 2\sum_{m=0}^{K-3}\; (m+2)^{2(D-1)}\norbra{\frac{B_2\,\ell^D K^{D-1}}{\sigma_H}}^m(E\omega)^{m}}}\Bigg)\,.
	\end{align}
	At this point, let us assume that $\norbra{\frac{B_2\,\ell^D K^{D-1}E\omega}{\sigma_H}}\leq \frac{1}{2}$, which holds for $\omega \leq \Omega_2 $, where $\Omega_2:=\frac{\sigma_H}{2B_2\,E\ell^D K^{D-1}}$. Then, we can upper bound the sum as:
	\begin{align}
		\sum_{m=0}^{K-3}\; (m+2)^{2(D-1)}\norbra{\frac{B_2\,\ell^D K^{D-1}}{\sigma_H}}^m(E\omega)^{m}\leq \sum_{m=0}^{K-3}\; \frac{(m+2)^{2(D-1)}}{2^m}\leq\sum_{m=0}^{\infty}\; \frac{(m+2)^{2(D-1)}}{2^m}:= s_{2(D-1)}\,,
	\end{align}
	where we introduced the finite constant $s_D := \sum_{m=0}^{\infty}\; \frac{(m+2)^{D}}{2^m}$. Then, we finally have:
	\begin{align}
		|\eta(\omega, K)| &\leq (E^2 (2R)^{D})\norbra{\frac{N\,C_\alpha(2R(\ell-1))}{\sigma_H^2}}
		\omega +(E^3\norbra{B_1+B_3^2(1+2s_{2(D-1)})})\norbra{\frac{N\,\ell^{2D}}{\sigma_H^3}}\omega^2\leq
		\\
		&\leq\norbra{\norbra{\frac{(2R)^{D}}{c_0}}C_\alpha(2R(\ell-1))}
		\omega +\norbra{\frac{B_4 }{c_0^{3/2}}\norbra{\frac{\ell^{2D}}{ \sqrt{N}}}} \,\omega^2	\,,
	\end{align}
	where we introduced the constant $B_4 := \norbra{B_1+B_3^2(1+2s_{2(D-1)})}$, and in the second line we applied the assumption on the variance from Eq.~\eqref{eq:varianceScaling}. This proves Eq.~\eqref{eq:etaConstants}.

	\subsection{Bounding $\nu_{1,j}(\omega,K)$}\label{app:nu1}
	We can bound $\nu_{1,j}(\omega,K)$ through a direct application of Eq.~\eqref{eq:B26}. Indeed, we have:
	\begin{align}
		|\nu_{1,j}(\omega,K)|=\left|\average{\widehat{h}_j\, \xi_j^0(\omega)\dots\xi_j^K(\omega)e^{i\omega z_j^\ell(K)}}{\rho}\right| \leq   E (K!)^{D-1}\norbra{\frac{B_2\,\ell^D}{\sigma_H}}^K(E\omega)^K
	\end{align}
	where we implicitly used the fact that $\|e^{i\omega z_j^\ell(K)}\| = 1$.
	
	\subsection{Bounding $\nu_{2,j}^{nc}(\omega,K)$ and $\nu^{nc}_{4,j}(\omega,K)$}\label{app:nu2}
	The corrections $\nu_{2,j}^{nc}(\omega,K)$ and $\nu^{nc}_{4,j}(\omega,K)$ arise from the non-commutativity of $\widehat{H}$. Indeed, both $\Xi_j^m(\omega)$ and $\Gamma_j^{m}(\omega)$ are zero for commuting models. Let us rewrite them here for concreteness:
	\begin{align}
		\Xi_j^m(\omega) :&= e^{i\omega (z_j^\ell(m-1)-z_j^\ell(m))}\norbra{e^{-i\omega (z_j^\ell(m-1)-z_j^\ell(m))}e^{i\omega z_j^\ell(m-1)}e^{-i\omega z_j^\ell(m) }-R^M_{m,j}(\omega)}e^{i\omega z_j^\ell(m) } = 
		\\
		&=e^{i\omega (z_j^\ell(m-1)-z_j^\ell(m))}\norbra{R^\infty_{m,j}(\omega)-R^M_{m,j}(\omega)}e^{i\omega z_j^\ell(m) }\,,
	\end{align}
	and:
	\begin{align}
		\Gamma_j^{m}(\omega) :&= e^{-i \omega(\widehat{H}-z^{\ell}_j(m))} \norbra{e^{i \omega(\widehat{H}-z^{\ell}_j(m))} e^{i \omega z^{\ell}_j(m)}e^{-i \omega\widehat{H}}-S^M_{m,j}(\omega)}e^{i \omega\widehat{H}}=
		\\
		&= e^{-i \omega(\widehat{H}-z^{\ell}_j(m))} \norbra{S^\infty_{m,j}(\omega)-S^M_{m,j}(\omega)}e^{i \omega\widehat{H}}\,,
	\end{align}
	where we used the definition in Eq.~\eqref{eq:RSinfty}. Then, since for commuting models $R^\infty_{m,j}(\omega)=R^M_{m,j}(\omega)=S^M_{m,j}(\omega)=S^\infty_{m,j}(\omega)=\idO$, this proves that these corrections are zero in this case.
	
	On the other hand, in the general case we can apply Lemma~\ref{lemma:clusterExp}, to obtain:
	\begin{align}
		&|\nu_{2,j}^{nc}(\omega,K)| \leq \sum_{m=1}^{K-1}\;\left|\average{\widehat{h}_j\, \xi_j^0(\omega)\dots\xi_j^{m-1}(\omega)\Xi_j^{m}(\omega)}{\rho}\right|\leq
		\\
		&\leq E \norbra{\frac{2\Gamma^{2} \ell^{D}(E\omega)^{2}}{\sigma_H^{2}}}\;\sum_{m=0}^{K-2}\norbra{ (m!)^{D-1}\norbra{\frac{B_2\,\ell^D}{\sigma_H}}^m(E\omega)^m}(m+1)^{D-1}\norbra{\frac{\Gamma \ell^{\frac{D}{2}} (m+1)^{\frac{D-1}{2}}}{\sigma_H} (E\omega)}^{M-1}\leq
		\\
		&\leq E \norbra{\frac{2\Gamma^{2} \ell^{D}(E\omega)^{2}}{\sigma_H^{2}}}\frac{1}{(2\ell^{\frac{D}{2}})^{M-1}}\norbra{\sum_{m=0}^{K-2}\frac{(m+1)^{D-1}}{2^{m}}}\leq  \norbra{\frac{E\,\Gamma^{2} s_{D-1}}{2^{M}\ell^{\frac{D}{2}(M-3)}\sigma_H^{2}}}(E\omega)^{2}\,.\label{eq:94}
	\end{align}
	where we used the bound in Eq.~\eqref{eq:B26}, together with the estimate in Lemma~\ref{lemma:clusterExp}, and in the last line we assumed $\omega \leq \Omega_2 $, where $\Omega_2:=\frac{\sigma_H}{2B_2\,E\ell^D K^{D-1}}$ to show that:
	\begin{align}
		\norbra{\frac{\Gamma \ell^{\frac{D}{2}} (m+1)^{\frac{D-1}{2}}}{ \sigma_H} (E\omega)}\leq\norbra{\frac{\Gamma \ell^{D} K^{\frac{D-1}{2}}}{\ell^{\frac{D}{2}} \sigma_H} (E\omega)}\leq \frac{\Gamma}{2B_2\ell^{\frac{D}{2}}}\leq\frac{1}{2\ell^{\frac{D}{2}}}\,,
	\end{align}
	where we also used the fact that $\Gamma<B_2$. Similarly, we can also bound:
	\begin{align}
		|\nu^{nc}_{4,j}(\omega&,K)| \leq \sum_{m=1}^{K-1}\;\left|\average{\widehat{h}_j\, \xi_j^0(\omega)\dots\xi_j^m(\omega)}{\rho}\average{\Gamma_j^{m+1}(\omega)}{\rho}\right|\leq
		\\
		&\leq E \norbra{\frac{2\Gamma^{2} \ell^{D}(E\omega)^{2}}{\sigma_H^{2}}}\;\sum_{m=1}^{K-1}\norbra{ (m!)^{D-1}\norbra{\frac{B_2\,\ell^D}{\sigma_H}}^m(E\omega)^m}(m+1)^{D-1}\norbra{\frac{\Gamma \ell^{\frac{D}{2}} (m+1)^{\frac{D-1}{2}}}{\sigma_H} (E\omega)}^{M-1}\leq
		\\
		&\leq\norbra{\frac{E\,\Gamma^{2} s_{D-1}}{2^{M}\ell^{\frac{D}{2}(M-3)}\sigma_H^{2}}}(E\omega)^{2}\,,
	\end{align}
	where the final result is obtained by carrying out the same steps as in Eq.~\eqref{eq:94}
	\subsection{Bounding $\nu_{3,j}(\omega,K)$}\label{app:nu3}
	First, we recall that:
	\begin{align}
		&{\rm supp}(z_j^\ell(m))\subseteq \{i| d(j,i)\in [2R (\ell m-1),N]\}\,;
		\\
		&{\rm supp}(\xi_j^m(\omega)) \subseteq{\rm supp}(R_{m,j}^M(\omega) ) \subseteq\{i| d(j,i)\in [2R (\ell (m-1)-M),2R (\ell\, m+M)]\}\,,
	\end{align}
	so that the distance between the support of $(\widehat{h}_j\, \xi_j^0(\omega)\dots\xi_j^m(\omega))$ and $e^{i\omega z_j^\ell(m+1)}$ is at least $2R(\ell - M-1)$, which we assume to be greater than $1$. Then, we can directly use Eq.~\eqref{eq:decayOfCorrelations}, together with the bound in Eq.~\eqref{eq:B26}, to obtain:
	\begin{align}
		&\left|\average{\widehat{h}_j\, \xi_j^0(\omega)\dots\xi_j^m(\omega)\norbra{e^{i\omega z_j^\ell(m+1)}-\average{e^{i\omega z_j^\ell(m+1)}}{\rho}}}{\rho}\right|\leq
		\\
		&\;\;\;\;\;\;\;\leq E (m!)^{D-1}\norbra{\frac{B_2\,\ell^D}{\sigma_H}}^m(E\omega)^m\min\{|\supp((\widehat{h}_j\, \xi_j^0(\omega)\dots\xi_j^m(\omega)))|,|\supp(e^{i\omega z_j^\ell(m+1)})|\}\alpha(2R(\ell - M-1))\,.
	\end{align}
	Plugging this in the definition of $\nu_{3,j}(\omega, K)$ we get:
	\begin{align}
		|\nu_{3,j}(\omega, K)|\leq  E\sum_{m=0}^{K-1}\;\norbra{  (m!)^{D-1}\norbra{\frac{B_2\,\ell^D}{\sigma_H}}^m(E\omega)^m(2R \ell\, (m+1))}\alpha(2R(\ell - M-1))\,,\label{eq:123}
	\end{align}
	where we used the fact that $|\supp((\widehat{h}_j\, \xi_j^0(\omega)\dots\xi_j^m(\omega)))|\leq|\supp(e^{i\omega z_j^\ell(m+1)})|$ and $|\supp((\widehat{h}_j\, \xi_j^0(\omega)\dots\xi_j^m(\omega)))|\leq 2R (\ell\, m+M) \leq 2R (\ell\, m+\ell)$.
	
	It should be noticed that, assuming the decay rate in Eq.~\eqref{eq:decayOfCorrelations2}, we can bring down the estimate in Eq.~\eqref{eq:123} to:
	\begin{align}
		|\nu_{3,j}(\omega, K)|\leq  E\sum_{m=0}^{K-1}\;\norbra{  (m!)^{D-1}\norbra{\frac{B_2\,\ell^D}{\sigma_H}}^m(E\omega)^m}\alpha(2R(\ell - M-1))\,.
	\end{align}
	
	\subsection{Bounding $\nu_{5,j}(\omega,K)$}\label{app:nu5}
	In this case, let us directly bound $\nu_{5}(\omega,K) := \frac{i}{\sigma_H}\,\sum_{j\in\mathcal{X}}\nu_{5,j}(\omega,K)$, that is:
	\begin{align}
		&|\nu_{5}(\omega,K) |=\frac{1}{\sigma_H}\sum_{m=1}^{K-1}\;\left|\average{\norbra{\sum_{j\in\mathcal{X}}\average{\widehat{h}_j\, \xi_j^0(\omega)\dots\xi_j^m(\omega)}{\rho}\norbra{\gamma_j^{m+1}(\omega) - \average{\gamma_j^{m+1}(\omega) }{\rho}}} e^{i \omega\widehat{H}}}{\rho}\right|\leq
		\\
		&\quad\leq\frac{1}{\sigma_H}\sum_{m=1}^{K-1}\;\left|\average{\norbra{\sum_{j\in\mathcal{X}}\average{\widehat{h}_j\, \xi_j^0(\omega)\dots\xi_j^m(\omega)}{\rho}\norbra{\gamma_j^{m+1}(\omega) - \average{\gamma_j^{m+1}(\omega) }{\rho}}}^{\!\!2\,} }{\rho}\right|^{1/2}\leq
		\\
		&\quad\leq \sum_{m=1}^{K-1}\norbra{ \frac{E}{\sigma_H} (m!)^{D-1}\norbra{\frac{B_2\,\ell^DE\omega}{\sigma_H}}^m }\!\!\norbra{\sum_{i,j\in\mathcal{X}}\;\average{(\gamma_i^{m+1}(\omega) - \average{\gamma_i^{m+1}(\omega) }{\rho}) (\gamma_j^{m+1}(\omega) - \average{\gamma_j^{m+1}(\omega) }{\rho})}{\rho}}^{1/2},\label{eq:appB50}
	\end{align}
	where in the first step we applied Cauchy-Schwarz inequality, and then used the estimate in Eq.~\eqref{eq:B26}. 
	
	Let us focus now on the term inside the last parenthesis. We can rewrite it as:
	\begin{align}
		\gamma_j^{m}(\omega) - &\average{\gamma_j^{m}(\omega) }{\rho} = \norbra{\norbra{e^{-i \omega(\widehat{H}-z^{\ell}_j(m))} S^M_{m,j}(\omega)- \idO} -\average{e^{-i \omega(\widehat{H}-z^{\ell}_j(m))} S^M_{m,j}(\omega)- \idO}{\rho}}=
		\\
		&=\norbra{e^{-i \omega(\widehat{H}-z^{\ell}_j(m))} S^M_{m,j}(\omega) - \average{e^{-i \omega(\widehat{H}-z^{\ell}_j(m))} S^M_{m,j}(\omega)}{\rho}} 
	\end{align}
	It should be noticed that $\widehat{H}-z^{\ell}_j(m) = \sum_{r=1}^{m}\,\widehat{H}_j^\ell(r)$ has support on sites $i$ satisfying $d(i,j)\leq 2R(\ell m+1)$, so the same holds for $e^{-i \omega(\widehat{H}-z^{\ell}_j(m))}$ as well. On the other hand, $S^M_{m,j}(\omega)$ has support on sites for which $d(i,j)\leq 2R(\ell m+M)$. This means that the sites in the support of $\gamma_j^{m}(\omega)$ satisfy $d(i,j)\leq 2R(\ell m+M)$. This also implies that the distance $d({\rm supp}(\gamma_i^{m}(\omega)),{\rm supp}(\gamma_j^{m}(\omega)))\geq d(i,j) - (4R(\ell m+M) +2R)$. 
	
	To keep the notation compact, let us introduce the constant $r_{m} := (4R(\ell m+M) +2R$. It should be noticed that $\supp(\gamma_i^{m}(\omega))\leq r_m$. Then, assuming $d(i,j)\geq r_{m+1}+1$, we can apply Eq.~\eqref{eq:decayOfCorrelations} to obtain:
	\begin{align}
		&\left|\average{(\gamma_i^{m+1}(\omega) - \average{\gamma_i^{m+1}(\omega) }{\rho}) (\gamma_j^{m+1}(\omega) - \average{\gamma_j^{m+1}(\omega) }{\rho})}{\rho}\right|\leq \nonumber
		\\
		&\qquad\qquad\qquad\qquad\qquad\leq \norbra{\frac{B_2\,\ell^D\,(m+1)^{D-1}E\omega}{\sigma_H}}^2r_{m+1}\,\alpha(d(i,j)-r_{m+1})\,.
	\end{align}
	where we used $\|\gamma_j^{m}(\omega) \|\leq \norbra{\frac{B_2\,\ell^D\,(m+1)^{D-1}E\omega}{\sigma_H}}$, which can be proved with methods akin to Eq.~\eqref{eq:70} and Eq.~\eqref{eq:73}. 
	
	Then, we can estimate the sum in Eq.~\eqref{eq:appB50} as:
	\begin{align}
		&\sum_{i,j\in\mathcal{X}}\;\left|\average{(\gamma_i^{m+1}(\omega) - \average{\gamma_i^{m+1}(\omega) }{\rho}) (\gamma_j^{m+1}(\omega) - \average{\gamma_j^{m+1}(\omega) }{\rho})}{\rho}\right|=
		\\
		&\qquad=\sum_{i\in\mathcal{X}}\norbra{\sum_{\substack{j\in\mathcal{X}\\
					d(i,j)<r_{m+1}+1}}+\sum_{\substack{j\in\mathcal{X}\\
					d(i,j)\geq r_{m+1}+1}}}\;\left|\average{(\gamma_i^{m+1}(\omega) - \average{\gamma_i^{m+1}(\omega) }{\rho}) (\gamma_j^{m+1}(\omega) - \average{\gamma_j^{m+1}(\omega) }{\rho})}{\rho}\right|\leq
		\\
		&\qquad\leq   \norbra{\frac{B_2\,\ell^D\,(m+1)^{D-1}E\omega}{\sigma_H}}^2\,\norbra{4(r_{m+1}+1)^{D} + c_D\sum_{r=r_{m+1}+1} \, r_{m+1}\alpha(r-r_{m+1})\, r^{D-1}}N\,,
	\end{align}
	where in the last step we implicitly used the fact that $\|\gamma_j^{m}(\omega) -\average{\gamma_j^{m}(\omega)}{\rho}\|\leq 2\|\gamma_j^{m}(\omega)\|$, and the estimate in Eq.~\eqref{eq:dimDef}.
	For conciseness, let us denote the term inside of the parenthesis by $c_{r_{m+1}}[\alpha]$. This can be estimated as:
	\begin{align}
		c_{r_{m+1}}[\alpha]  &= 4(r_{m+1}+1)^{D}\norbra{1 + \frac{c_D}{4}\sum_{r=1} \, \alpha(r)\, \frac{(r+r_{m+1})^{D-1}}{(r_{m+1}+1)^{D-1}}} \leq
		\\
		&\leq 4(r_{m+1}+1)^{D}\norbra{1 + c_D\sum_{r=1} \, \alpha(r)\, {(r+1)^{D-1}}} = 4(1+C_\alpha(1))\,(r_{m+1}+1)^{D}\,,
	\end{align}
	where we used the definition of the function $C_\alpha(\ell)$ from Eq.~\eqref{eq:decayOfCorrConst}. Moreover, under the assumption $M\leq \ell$, we can upper-bound $r_m\leq 6R\ell (m+1)$.
	Wrapping everything up, we finally obtain:
	\begin{align}
		|\nu_{5}(\omega,K) | \leq \sum_{m=1}^{K-1}\norbra{ \frac{E}{\sigma_H} (m!)^{D-1}\norbra{\frac{2B_2\,\ell^DE\omega}{\sigma_H}}^m }\!\!\norbra{\frac{6B_2 R^{\frac{D}{2}}\,\ell^{\frac{3D}{2}}\,(m+2)^{2(D-1)}E\omega}{\sigma_H}}\sqrt{(1+C_\alpha(1)) N}\,.\label{eq:113}
	\end{align}
	
	\subsection{Bound on $\nu(\omega, K)$}\label{app:boundNu}
	In this section we put together the results from Sec.~\ref{app:nu1},~\ref{app:nu2},~\ref{app:nu3} and~\ref{app:nu5}. Before doing so, though, they need some additional manipulations. Throughout the section we are assuming that $\omega \leq \Omega_2$.
	Let us begin with an estimate of $\nu_{1}(\omega,K) := \frac{i}{\sigma_H}\sum_{j} \nu_{1,j}(\omega,K)$ :
	\begin{align}
		|\nu_{1}(\omega,K)| &\leq \frac{1}{\sigma_H}\sum_{j\in\mathcal{X}}|\nu_{1,j}(\omega,K)| \leq \frac{1}{\sigma_H}\sum_{j\in\mathcal{X}}  E\, (K!)^{D-1}\norbra{\frac{B_2\,\ell^D}{\sigma_H}}^K(E\omega)^K\leq
		\\
		&\leq (E^2B_2)\norbra{\frac{\,\ell^DK^{D-1}N}{2^{K-1}\sigma_H^2}}\,\omega\,.
	\end{align}		
	Similarly, we also have that:
	\begin{align}
		|\nu_{3}(\omega,K)|&\leq \frac{1}{\sigma_H}\sum_{j\in\mathcal{X}}|\nu_{3,j}(\omega,K)| \leq \frac{2R \ell EN\alpha(2R(\ell - M-1))}{\sigma_H}\sum_{m=0}^{K-1}\;\norbra{  (m+1)(m!)^{D-1}\norbra{\frac{B_2\,\ell^D}{\sigma_H}}^m(E\omega)^m }\leq
		\\
		& \leq \frac{2R \ell EN\alpha(2R(\ell - M-1))}{\sigma_H}\sum_{m=0}^{\infty}\,\frac{(m+1)}{2^m} = (4E)\norbra{\frac{2R \ell N\alpha(2R(\ell - M-1))}{\sigma_H}}\;\label{eq:140}
	\end{align}
	Finally, we can manipulate Eq.~\eqref{eq:113} to give:
	\begin{align}
		|\nu_{5}(\omega,K) | &\leq  (12B_2^2R^{\frac{D}{2}}\sqrt{(1+C_\alpha(1))}\,)\frac{\ell^{\frac{5D}{2}}E^3\sqrt{N}\omega^2}{\sigma_H^3}\sum_{m=1}^{K-1}\norbra{ \frac{(m+2)^{3(D-1)}}{2^{m}}}\,\leq
		\\
		&\leq (B_5E^3)\frac{\ell^{\frac{5D}{2}}\sqrt{N}}{\sigma_H^3}\omega^2\,,
	\end{align}
	where we introduced the constant $B_5:=(12B_2^2R^{\frac{D}{2}}\sqrt{(1+C_\alpha(1))}\,s_{3(D-1)})$.

	Finally, we can deduce directly from Sec.~\ref{app:nu2}, that:
	\begin{align}
		|\nu_{2}^{nc}(\omega,K)| \leq E^3B_6\norbra{\frac{N}{2^{M}\ell^{\frac{D}{2}(M-3)}\sigma_H^{3}}}\omega^2 \,;
		\qquad\;\;
		|\nu_{4}^{nc}(\omega,K)| \leq E^3B_6\norbra{\frac{N}{2^{M}\ell^{\frac{D}{2}(M-3)}\sigma_H^{3}}}\omega^2 \,.
	\end{align}
	where we introduced the constant $B_6$  defined as:
	\begin{align}
		B_6:=
		\begin{cases}
			0 & {\text{for commuting models}}\\
			(\Gamma^{2} s_{D-1}) & {\text{otherwise}}
		\end{cases}\,.
	\end{align}
	
	Putting everything together, and applying the estimate on $\sigma_H$ from Eq.~\eqref{eq:varianceScaling}, gives us:
	\begin{align}
		|\nu(\omega, K)| \leq& \frac{4}{\sqrt{c_0}}\norbra{\sqrt{N}\,2R\ell \alpha(2R(\ell - M-1))}+\frac{B_2}{c_0}\norbra{\frac{\,\ell^DK^{D-1}}{2^{K-1}}}\omega + \nonumber
		\\
		&\qquad\qquad\qquad\qquad\qquad\qquad\qquad+ \frac{ 1}{c_0^{\frac{3}{2}}}\norbra{B_5\norbra{\frac{\ell^{2D+\frac{D}{2}}}{N}}+B_6\norbra{\frac{1}{2^{M}\ell^{\frac{D}{2}(M-3)}\sqrt{N}}}}\omega^2\,.\label{eq:145}
	\end{align}
	This proves Eq.~\eqref{eq:nuConstants}. It should be noticed that for a decay of the type in Eq.~\eqref{eq:decayOfCorrelations2}, we can repeat the same steps in Eq.~\eqref{eq:140} to get:
	\begin{align}
		|\nu_{3}(\omega,K)|\leq (2E)\norbra{\frac{N\alpha(2R(\ell - M-1))}{\sigma_H}}\,.
	\end{align}
	This implies that Eq.~\eqref{eq:145} changes only in the first term, where now we have $\tilde{c}_3 := c_3/(4R\ell)$.


	\section{Properties of $R_{m,j}^M(\omega) $ and $S_{m,j}^M(\omega) $} \label{app:clusterExp}
	
	In this section we give a general treatment of the properties of the function:
	\begin{align}
		\zeta(\omega) =e^{i\omega (X+Y)}e^{-i\omega X}e^{-i\omega Y}\,,\label{eq:d1}
	\end{align}
	before moving on to specializing to $R_{m,j}^M(\omega) $ and $S_{m,j}^M(\omega) $. In particular, let us introduce the $M$-truncated Taylor expansion of $\zeta(\omega)$:
	\begin{align}
		\zeta^M(\omega) := \sum_{n=0}^M \; \frac{\omega^n}{n!} \zeta^{(n)}(0) \,,\label{eq:zetaM}
	\end{align}
	where $\zeta^{(n)}(\omega)$ denotes the $n$-th derivative of $\zeta(\omega)$. We prove that:
	\begin{lemma}\label{lemma:clusterExp2}
		Let $X = \sum_{i\in\mathcal{X}} X_i$ and $Y = \sum_{i\in\mathcal{X}} Y_i$ be two operators, such that each of $X_i,\, Y_i$ is $R$-local, and  $\|X_i\|,\,\|Y_i\|\leq B$. Define $\zeta(\omega)$ and $\zeta^M(\omega)$ as in Eq.~\eqref{eq:d1} and Eq.~\eqref{eq:zetaM}. Let $\mathcal{S}$ be the support of $[X,Y]$, and let $C_{X\cap Y} := |\mathcal{S}|$. Then, the following properties hold:
		\begin{enumerate}
			\item $\zeta(0)=\idO$ and $\zeta^{(1)}(0)=0$;\label{it:zeros}
			\item the support of $\zeta^{M}(\omega)$ satisfies:\label{it:supp}
			\begin{align}
				{\rm supp}(\zeta^{M}(\omega))\subseteq\{i|\,d(i, \mathcal{S})\leq 2R(M-2)\}\,;
			\end{align}
			\item the norm of the $n$-th derivative of $\zeta(\omega)$ is upper bounded by:\label{it:norm}
			\begin{align}
				\|\zeta^{(n)}(0)\|\leq ((2\lambda)\sqrt{\gamma})^{n} (n)!\,.\label{eq:l4e31}
			\end{align}
			where $\lambda:= (2c_D (2R)^D B)$ and $\gamma := \max\{1, \frac{  C_{X\cap Y}}{(4c_D (2R)^D )}\}$. This implies that:
			\begin{align}
				\|\zeta^M(\omega)\|\leq  \frac{1-((2\lambda)  \sqrt{\gamma}|\omega|)^{M+1}}{1-((2\lambda)  \sqrt{\gamma}|\omega|)}\,;\label{eq:l4e32}
			\end{align}
		\end{enumerate}
		Now, assume $\omega\in[0, \,\frac{1}{2(2\lambda)\sqrt{\gamma}})$. Then, it also holds that:
		\begin{align}
			\|\zeta^M(\omega)\|\leq 2\norbra{1 - \frac{1}{2^{M+1}}}\,;
			\;\;\;\;\;\;\;
			\|\zeta(\omega)-\zeta^M(\omega)\|\leq 2((2\lambda)  \sqrt{\gamma}\,\omega)^{M+1}\,; 
			\;\;\;\;\;\;\;
			\|\partial^{(k)}_\omega \zeta^M(\omega)\|\leq 2\,(k!) \,((2\lambda)  \sqrt{\gamma}\,)^k\,.\label{it:convNorm}
		\end{align}
	\end{lemma}
	The next sections are dedicated to the proof of this Lemma. First, in Sec.~\ref{subsec:expTaylor} we give a recursive expression of $\zeta^{(n)}(\omega)$; then, in Sec.~\ref{subsec:support} we characterize its support, and in Sec.~\ref{subsec:opnorm} its operator norm; finally, in Sec.~\ref{subsec:lemma3} we specialize the results obtained to $R_{m,j}^M(\omega) $ and $S_{m,j}^M(\omega)$ so to prove Lemma~\ref{lemma:clusterExp}.
	
	\subsection{Recursive expression of the derivatives of $\zeta(\omega)$}\label{subsec:expTaylor}
	Our starting point is the derivative of Eq.~\eqref{eq:d1}, which takes the form:
	\begin{align}
		\zeta^{(1)}(\omega) &= -i \, e^{i\omega (X+Y)}\norbra{Y e^{-i\omega X}e^{-i\omega Y} - e^{-i\omega X} Y e^{-i\omega Y}} = -i\zeta(\omega) \,e^{i\omega Y}\norbra{e^{i\omega X}Ye^{-i\omega X} -Y}e^{-i\omega Y}\,
		\\
		&=\zeta(\omega) \,e^{i\omega Y}\norbra{\int_{0}^{\omega}\de\omega_1\;e^{i\omega_1 X}[X,Y]e^{-i\omega_1 X}}e^{-i\omega Y} = \zeta(\omega) e^{i\omega Y} \theta(\omega) e^{-i\omega Y}  = \zeta(\omega) \,\Theta(\omega)\,,\label{eq:d3}
	\end{align}
	where we implicitly defined $\theta(\omega)$ to be the integral inside of the parenthesis, and $\Theta(\omega):= e^{i\omega Y} \theta(\omega) e^{-i\omega Y}$. The latter satisfies:
	\begin{align}
		\Theta^{(1)}(\omega) = i [Y, e^{i\omega Y}\theta(\omega) e^{-i\omega Y}] +  e^{i\omega Y} \theta^{(1)}(\omega) e^{-i\omega Y}\,.
	\end{align}
	This implies the recursive formula for the $k$-th derivative:
	\begin{align}
		\Theta^{(k)}(\omega) = \sum_{m=0}^{k}\binom{k}{m} \, (i)^{k-m} [Y, (e^{i\omega Y}\theta^{(m)}(\omega) e^{-i\omega Y})]_{k-m}\,,\label{eq:d5}
	\end{align}
	where we introduced the notation $[A,B]_n$ for the nested commutator, inductively defined as:
	\begin{align}
		\begin{cases}
			[A,B]_{0} = B\\
			[A,B]_{n} =[A,[A,B]_{n-1}]
		\end{cases}\,.
	\end{align}
	From its explicit expression in Eq.~\eqref{eq:d3} it can be noticed that $\theta(0)=0$, which implies condition~\ref{it:zeros}, i.e., $\zeta^{(1)}(0) =0$, Moreover, for $m\geq1$ one has:
	\begin{align}
		\theta^{(m)}(\omega) = (i)^{m-1}e^{i\omega X} [X,Y]_{m}e^{-i\omega X}\,.
	\end{align}
	Then, by deriving Eq.~\eqref{eq:d3}, we get an expression for the $(n+1)$-th derivative of $\zeta(\omega)$:
	\begin{align}
		\zeta^{(n+1)}(\omega) &= \sum_{k=0}^{n}\binom{n}{k} \, \zeta^{(n-k)}(\omega) \Theta^{(k)}(\omega) = 
		\\
		&=\zeta^{(n)}\Theta(\omega)+\sum_{k=1}^{n}\binom{n}{k} (i)^{k} \zeta^{(n-k)}(\omega)   [Y, (e^{i\omega Y}\theta(\omega) e^{-i\omega Y})]_{k}+\nonumber
		\\
		&\qquad\qquad\qquad\qquad+\sum_{k=1}^{n}\sum_{m=1}^k\binom{n}{k} \binom{k}{m} \, (i)^{k-1}\zeta^{(n-k)} (\omega)  [Y, (e^{i\omega Y}e^{i\omega X} [X,Y]_{m}e^{-i\omega X} e^{-i\omega Y})]_{k-m}\,,
	\end{align}
	This relation allows us to recursively compute $\zeta^{(n+1)}(\omega)$ in terms of derivatives of at most order $n$. In particular, the $(n+1)$-th derivative evaluated in $0$ is given by:
	\begin{align}
		\zeta^{(n+1)}(0) &=  \sum_{k=1}^{n}\sum_{m=1}^k\binom{n}{n-k, k-m, m} \, (i)^{k-1}\zeta^{(n-k)} (0) \; [Y, [X,Y]_{m}]_{k-m}\,.\label{eq:d11}
	\end{align}
	where we used the fact that $\Theta(0)=\theta(0)=0$. This expression is the main result of this section.

    
	\subsection{Support of $\zeta^{M}(\omega)$}\label{subsec:support}
	In this section, we use Eq.~\eqref{eq:d11} to characterize the support of $\zeta^{M}(\omega)$ as:
	\begin{align}
		{\rm supp}(\zeta^{M}(\omega))\subseteq\{i|\,d(i, \mathcal{S})\leq 2R(M-2)\}\,;\label{eq:suppProof}
	\end{align}
	where $\mathcal{S}$ is the support of $[X,\,Y]$. This is done by induction. First, it should be noticed that for $M=2$, we have:
	\begin{align}
		\zeta^{M=2}(\omega) = \idO + \frac{\omega^2}{2} \zeta^{(2)}(0) = \idO + \frac{\omega^2}{2}\,[X,Y]\,,
	\end{align}
	which proves the induction basis. Then, suppose that Eq.~\eqref{eq:suppProof} holds for some $M$. Then, for $M+1$, since $\zeta^{M+1}(\omega)= \zeta^{M}(\omega)+ \frac{\omega^{M+1}}{(M+1)!} \zeta^{(M+1)}(0)$  we have:
	\begin{align}
		{\rm supp}(\zeta^{M+1}(\omega))\subseteq{\rm supp}(\zeta^{M}(\omega))\cup {\rm supp}(\zeta^{(M+1)}(\omega))\,.
	\end{align}
	Then, the claim would follow if we could prove that ${\rm supp}(\zeta^{(M+1)}(\omega))\subseteq\{i|\,d(i, \mathcal{S})\leq 2(M-1)R\}$. To this end, it is sufficient to show that:
	\begin{align}
		{\rm supp}\norbra{[Y, [X,Y]_{m}]_{n-m} } \subseteq \{i |\, d(i, \mathcal{S})\leq 2R(n-1)\}\,.\label{eq:nestCommSupp}
	\end{align}
	Then, using Eq.~\eqref{eq:d11} gives the claim in~\ref{it:supp}, since:
	\begin{align}
		\zeta^{(M+1)}(0) &=  \sum_{k=1}^{M}\sum_{m=1}^k\binom{n}{n-k, k-m, m} \, (i)^{k-1}\zeta^{(n-k)} (0) \; [Y, [X,Y]_{m}]_{k-m}
	\end{align}
	contains nested commutator of order at most $M$.
	
	It remains to prove Eq.~\eqref{eq:nestCommSupp}, which we do by induction. In particular, let $Z=\sum_{i\in \mathcal{X}} Z_i$ and $H=\sum_{i\in \mathcal{X}} H_i$ be a sum of $R$-local observables. Then, it holds that:
	\begin{align}
		{\rm supp}\norbra{[Z,H]_{n}}  \subseteq \{i |\, d(i, {\rm supp}\norbra{H})\leq 2R\,n\}\,.\label{eq:zhInd}
	\end{align}
	Indeed, for $n=0$ this is trivial. Then, suppose it holds for some integer $n$, and we want to prove it for $n+1$. It should be noticed that we can rewrite $[Z,H]_{n+1}$ as:
	\begin{align}
		[Z,H]_{n+1}  = [Z,[Z,H]_{n}] = \sum_{\substack{i\in\mathcal{X}\\d(i,\mathcal{S}_{n}^H)\leq 2R}} [Z_i, [Z,H]_{n}]\,,
	\end{align}
	where we introduced the notation $\mathcal{S}_{n}^H$ for the support of $[Z,H]_{n}$, and we used the fact that the only non-commuting terms can arise when $Z_i$ is at most at distance $2R$ from $\mathcal{S}_{n}^H$ (each of the terms in $\mathcal{S}_{n}^H$ acts on a ball of at most radius $R$ around its site, and the same holds for $Z_i$). Then, we obtain:
	\begin{align}
		{\rm supp}\norbra{[Z,H]_{n+1}}  \subseteq \{i |\, d(i, \mathcal{S}_{n}^H)\leq 2\,R\}\subseteq \{i |\, d(i, {\rm supp}\norbra{H})\leq 2R\,(n+1)\}\,,
	\end{align}
	where in the last step we applied the triangle inequality. This proves Eq.~\eqref{eq:zhInd}. With this result in hand, it directly follows that:
	\begin{align}
		{\rm supp}\norbra{[X,Y]_{m}} =  {\rm supp}\norbra{[X,[X,Y]]_{m-1}}\subseteq \{i |\, d(i, \mathcal{S})\leq 2R(m-1)\}\,,
	\end{align}
	and that:
	\begin{align}
		{\rm supp}\norbra{[Y, [X,Y]_{m}]_{n-m} } \subseteq \{i |\, d(i,{\rm supp}\norbra{[X,Y]_{m}})\leq 2R(n-m)\}\subseteq \{i |\, d(i, \mathcal{S})\leq 2R(n-1)\}\,,
	\end{align}
	where, once again, we used the triangle inequality to prove the last implication. 
	
	\subsection{Operator norm}\label{subsec:opnorm}
	The main ingredient of this section is the following estimate:
	\begin{align}
		\|[Y,[X,Y]_{m}]_{n-m}\| \leq 2BC_{X\cap Y} \, (n! \,\lambda^n)\,.\label{eq:uppBNormComm}
	\end{align}
	where $\lambda := (2c_D (2R)^D B)$. This allows us to recursively estimate the norm of $\zeta^{(n)}(0)$. Indeed, starting from Eq.~\eqref{eq:d11}, we have:
	\begin{align}
		\|\zeta^{(n+1)}(0)\| &\leq \sum_{k=1}^{n}\sum_{m=1}^k\binom{n}{n-k, k-m, m} \,\|\zeta^{(n-k)} (0)\| \;\| [Y, [X,Y]_{m}]_{k-m}\| \leq
		\\
		&\leq2BC_{X\cap Y}\,\sum_{k=1}^{n}\frac{n!}{(n-k)!} \|\zeta^{(n-k)} (0)\| \,\lambda^k\norbra{\sum_{m=1}^k\binom{k}{m} }\leq
		\\
		&\leq ( (n+1)!(2\lambda)^{n+1})\,\norbra{\frac{ BC_{X\cap Y}}{(2\lambda) (n+1)}\sum_{k=1}^{n}\frac{\|\zeta^{(n-k)} (0)\| }{((n-k)! (2\lambda)^{n-k})} }\,,\label{eq:d26}
	\end{align}
	which can be rearranged as:
	\begin{align}
		\frac{\|\zeta^{(n+1)}(0)\|}{( (n+1)!(2\lambda)^{n+1})}\leq \frac{ BC_{X\cap Y}}{(2\lambda) (n+1)}\norbra{\sum_{k=1}^{n}\frac{\|\zeta^{(n-k)} (0)\| }{((n-k)! (2\lambda)^{n-k})}}\,.\label{eq:intStep}
	\end{align}
	Let us introduce the constant $\gamma := \max\{1, \frac{  C_{X\cap Y}}{(4c_D (2R)^D )}\}$. The expression just obtained allows us to derive the upper bound:
	\begin{align}
		\|\zeta^{(n)}(0)\|\leq ((2\lambda) \sqrt{\gamma})^{n} (n!)\,.
	\end{align}
	First, notice that $\|\zeta^{(0)}(0)\|=1$ and $\|\zeta^{(1)}(0)\|=0$. Then, by induction on Eq.~\eqref{eq:intStep} we have:
	\begin{align}
		\frac{\|\zeta^{(n+1)}(0)\| }{( (n+1)!(2\lambda)^{n+1})}  \leq \frac{\gamma}{n+1} \sum_{k=1}^{n}\frac{\|\zeta^{(n-k)} (0)\| }{( (n-k)!(2\lambda)^{n-k})} \leq \frac{\gamma}{n+1} \norbra{1+\sum_{k=1}^{n-2} \,\gamma^{\frac{n-k}{2}}} \leq \frac{\gamma^{1+\frac{n-1}{2}}  (n-1)}{n+1}\leq \gamma^{\frac{n+1}{2}}\,,
	\end{align}
	where we implicitly used the fact that $1\leq \gamma\implies \gamma^{\frac{n-k}{2}}\leq \gamma^{\frac{n-1}{2}}$. This proves Eq.~\eqref{eq:l4e31}. Then, as a direct consequence we also have:
	\begin{align}
		\|\zeta^M(\omega)\|\leq \sum_{n=0}^M \; \frac{|\omega|^n}{n!} \|\zeta^{(n)}(0) \| \leq  \sum_{n=0}^M \; ((2\lambda)  \sqrt{\gamma}|\omega|)^n  \leq \frac{1-((2\lambda)  \sqrt{\gamma}|\omega|)^{M+1}}{1-((2\lambda)  \sqrt{\gamma}|\omega|)}\,,\label{eq:162in}
	\end{align}
	which proves Eq.~\eqref{eq:l4e32}.
	
	At this point, let us assume that $\omega\in[0, \,\frac{1}{2(2\lambda)\, \sqrt{\gamma}})$. In this regime, $((2\lambda)  \sqrt{\gamma}\omega)\leq \frac{1}{2}$. Then, it follows from Eq.~\eqref{eq:162in} that:
	\begin{align}
		\|\zeta^M(\omega)\| \leq 2\norbra{1-\frac{1}{2^{M+1}}}\,,
	\end{align}
	and that:
	\begin{align}
		\|\zeta(\omega)-\zeta^M(\omega)\|\leq\sum_{n=M+1}^\infty \; \frac{|\omega|^n}{n!} \|\zeta^{(n)}(0) \| \leq  \frac{((2\lambda)  \sqrt{\gamma}\omega)^{M+1}}{1-((2\lambda)  \sqrt{\gamma}\,\omega)}\leq 2((2\lambda)  \sqrt{\gamma}\,\omega)^{M+1}\,.
	\end{align}
	This proves the first two items of Eq.~\eqref{it:convNorm}. Moreover, consider now the $k$-th derivative of $\zeta^M(\omega)$. Then, in the same range, it also holds that:
	\begin{align}
		\|\partial_{\omega}^{(k)} \,(\zeta^M(\omega) )\|\leq \sum_{n=k}^M \frac{|\omega|^{n-k}}{(n-k)!} \|\zeta^{(n)}(0) \|\leq k! ((2\lambda)  \sqrt{\gamma}\,)^k\,\sum_{n=k}^M \binom{n}{k}  ((2\lambda)  \sqrt{\gamma}\omega)^{n-k}\leq 2\,(k!) \,((2\lambda)  \sqrt{\gamma}\,)^k\,
	\end{align}
	where in the last step we used the identity $\sum_{n=k}^\infty \binom{n}{k} (1/2)^{n-k} = 2^{k+1}$. This completes the proof of Lemma~\ref{lemma:clusterExp2}.
	
	Then, it remains to prove Eq.~\eqref{eq:uppBNormComm}. 
	As in the previous section, let $Z^1=\sum_{i\in \mathcal{X}} Z_i$, $Z^2=\sum_{i\in \mathcal{X}} Z_i$ and $H=\sum_{i\in \mathcal{X}} H_i$ to be three arbitrary operators that are sum of $R$-local observables, which satisfy $\|Z^{1}_i\|,\,\|Z^{2}_i\|\leq B$ and $\|H_i\|\leq \tilde{B}$. Then, we can give a first estimate:
	\begin{align}
		\|[Z^1,[Z^2,H]_{\tilde{m}}]_{\tilde{n}}\| \leq \sum_{i_0,\dots, i_n} \|[Z^1_{i_{\tilde{m}+\tilde{n}}}, [Z^1_{i_{\tilde{m}+\tilde{n}-1}},\dots[Z^1_{i_{\tilde{m}+1}},[Z^2_{i_{\tilde{m}}},\dots,[Z^2_{i_1},H_{i_0}]]]\| \leq  2^{n+1} \tilde{B} \,B^{\tilde{m}+\tilde{n}}\#_{\tilde{m}+\tilde{n}}\,.
	\end{align}
	where we denote by $\#_n$ is the number of non-zero commutators of local terms at order $n$. This number can be estimated by using the locality of $Z^1$, $Z^2$ and $H$. Let $\tilde{Z}_{i_k}$ be either $Z^1_{i_k}$ or $Z^2_{i_k}$, depending on whether $k> \tilde{m}$ or not. Then, it should be noticed that the terms in the sum are non-zero if and only if for all $k\in\{1,\dots,\tilde{m}+\tilde{n}\}$, the support of $\tilde{Z}_{i_k}$ intersects the support of $[\tilde{Z}_{i_{k-1}},\dots [\tilde{Z}_{i_1},H_{i_0}]]$. This allows to estimate the number of terms in the sum by reducing it to a combinatorial problem. Indeed, suppose that we have a sequence of indices $\{i_{k},\dots, i_{0} \}$ corresponding to a non-zero commutator. This corresponds to at most $(k+1)$ different points over the lattice. When adding an additional term, say $i_{k+1}$, in order for the total commutator not to be zero, $i_{k+1}$ must be inside one of the $2R$-ball around any of the (at most) $k+1$ different points already present in the lattice. Then, this means that:
	\begin{align}
		\#_{k+1} \leq (k+1) c_D (2R)^D \, \#_{k}\,\qquad\implies\qquad\#_{k+1}\leq ((k+1)!) (c_D (2R)^D)^k \, \#_{0}\,.
	\end{align}
	Finally, it should be noticed that $\#_{0}$ simply corresponds to the number of sites in the support of $H$. This prove that:
	\begin{align}
		\|[Z,H]_{n}\| \leq  2\tilde{B} (({\tilde{m}+\tilde{n}})!) (2c_D (2R)^DB)^{\tilde{m}+\tilde{n}}  |{\rm supp}\norbra{H}|\,.\label{eq:173}
	\end{align}
	Applying this expression to Eq.~\eqref{eq:uppBNormComm}, we obtain:
	\begin{align}
		\|[Y,[X,Y]_{m}]_{n-m}\| &=\|[Y,[X,([X,Y])]_{m-1}]_{n-m}\|\leq \|[X,Y]\| ((n-1)!)\lambda^{n-1} |{\rm supp}\norbra{[X,Y]}|\leq
		\\
		&\leq 2BC_{X\cap Y} ((n-1)!)\lambda^{n}\leq2BC_{X\cap Y} (n!)\lambda^{n}\,.
	\end{align}
	This concludes the proof.
	
	\subsection{Proof of Lemma~\ref{lemma:clusterExp}}\label{subsec:lemma3}
	We are now ready to show how Lemma~\ref{lemma:clusterExp} is a direct consequence of Lemma~\ref{lemma:clusterExp2}. Indeed, both  $R_{m,j}^{\infty}(\omega) $ and $S_{m,j}^{\infty}(\omega) $ are of the form in Eq.~\eqref{eq:d1}, since:
	\begin{align}
		R^{\infty}_{m,j}(\omega):=e^{-i\omega (z_j^\ell(m-1)-z_j^\ell(m))}e^{i\omega z_j^\ell(m-1) }e^{-i\omega z_j^\ell(m) }\,;\qquad S^{\infty}_{m,j}(\omega):=e^{i \omega(\widehat{H}-z^{\ell}_j(m))} e^{i \omega z^{\ell}_j(m)}e^{-i \omega\widehat{H}}\,.
	\end{align}
	Then, let us define $X_R:=-z_j^\ell(m-1)$ and $Y_R = z_j^\ell(m)$, together with $X_S:=-z_j^\ell(m)$ and $Y_S = \widehat{H}$. In order to apply Lemma~\ref{lemma:clusterExp2}, we first need to match the various constants. First, notice that in both cases we can set $B = \frac{E}{\sigma_H}$. Then, define the two sets $\mathcal{S}_R := {\rm supp}([X_R, Y_R])$ and $\mathcal{S}_S := {\rm supp}([X_S, Y_S])$.
	In the first case we have:
	\begin{align}
		\mathcal{S}_R&=  {\rm supp}([\widehat{H}_j^\ell(m)+z_j^\ell(m),  z_j^\ell(m)]) ={\rm supp}([\widehat{H}_j^\ell(m),  z_j^\ell(m)]) \subseteq
		\\
		&\qquad\qquad\qquad\subseteq\{i|\,\exists k,\, d(j,k)\in [2R\ell (m-1),2R\ell m], \,d(i, k)\leq 2R\}\,,\label{eq:178} 
	\end{align} 
	where we used the fact that $z_j^\ell(m-1)=\widehat{H}_j^\ell(m)+z_j^\ell(m)$. In other words, the set in Eq.~\eqref{eq:178} comprises all the points that are at a distance at most $2R$ from the support of $\widehat{H}_j^\ell(m)$. This expression can be further simplified using the triangle inequality:
	\begin{align}
		\mathcal{S}_R&\subseteq\{i| d(j,i)\in [2R(\ell (m-1)-1),2R(\ell m+1)]\}\,,
	\end{align}
	Then, following Sec.~\ref{subsec:support}, we obtain:
	\begin{align}
		{\rm supp}(R_{m,j}^M(\omega) ) \subseteq\{i| d(j,i)\in [2R (\ell (m-1)-M),2R (\ell\, m+M)]\}\,.
	\end{align}
	A similar discussion also gives:
	\begin{align}
		\mathcal{S}_S&=  {\rm supp}\norbra{\sqrbra{z_j^\ell(m), \sum_{k=1}^{m}\,\widehat{H}_j^\ell(k)+z_j^\ell(m)}} ={\rm supp}\norbra{\sqrbra{z_j^\ell(m), \widehat{H}_j^\ell(m)}}\subseteq
		\\
		&\qquad\qquad\qquad\subseteq\{i| d(j,i)\in [2R(\ell (m-1)-1),2R(\ell m+1)]\}\,,
	\end{align}
	where we implicitly used the fact that $d({\rm supp}(\widehat{H}_j^\ell(k)),{\rm supp}(z_j^\ell(m)))> 2R\ell(m-k) $ to commute the elements in the sum with $k\neq m$. This shows that $\mathcal{S}_S = \mathcal{S}_R$. Then, this also means that $C_{X\cap Y}= C_{X_R\cap Y_R}= C_{X_S\cap Y_S}$. Using  directly  Eq.~\eqref{eq:173}, we get $C_{X\cap Y} = |\widehat{H}_j^\ell(m)| \leq c_D(2R\ell)^{D} m^{D-1}$.
	
	We are now ready to use Lemma~\ref{lemma:clusterExp2}. First, it should be noticed that:
	\begin{align}
		\|R^{(m)}_{m,j}(0)\|,\,\|S^{(m)}_{m,j}(0)\|\leq \norbra{\frac{\Gamma E\,\ell^{\frac{D}{2}} m^{\frac{D-1}{2}}}{\sigma_H}}^{n} (n!)
	\end{align}
	where we introduced the constant $\Gamma= \max\{(4c_D (2R)^D ),  (2c_D^2(2R)^{D})\}$. This expression directly implies that:
	\begin{align}
		\|R_{m,j}^M(\omega)\|,\,\|S_{m,j}^M(\omega)\| \leq \frac{1-\norbra{\frac{\Gamma \ell^{\frac{D}{2}} m^{\frac{D-1}{2}}}{\sigma_H} (E|\omega|)}^{M+1}}{1-\norbra{\frac{\Gamma \ell^{\frac{D}{2}} m^{\frac{D-1}{2}}}{\sigma_H} (E|\omega|)}}\,.
	\end{align}
	
	Moreover, let $\omega\leq \Omega_1 = \norbra{\frac{\sigma_H}{2E\,\Gamma \ell^{\frac{D}{2}} K^{\frac{D-1}{2}}}}$. Carrying out the same computations as in Sec.~\ref{subsec:opnorm}, in this regime we have:
	\begin{align}
		\|R_{m,j}^M(\omega)\|\leq 2\norbra{1 - \frac{1}{2^{M+1}}}\,;\;\;&\;\;\|S_{m,j}^M(\omega)\|\leq 2\norbra{1 - \frac{1}{2^{M+1}}}\,;
		\\
		\|R_{m,j}^\infty(\omega)-R_{m,j}^M(\omega)\|\leq 2\norbra{\frac{\Gamma \ell^{\frac{D}{2}} m^{\frac{D-1}{2}}}{\sigma_H} (E|\omega|)}^{M+1}\,;\;\;&\;\;\|S_{m,j}^\infty(\omega)-R_{m,j}^M(\omega)\|\leq 2\norbra{\frac{\Gamma \ell^{\frac{D}{2}} m^{\frac{D-1}{2}}}{\sigma_H} (E|\omega|)}^{M+1}\,;
		\\
		\|\partial^{(k)}_\omega R_{m,j}^M(\omega)\|\leq 2\,(k!) \,\norbra{\frac{E\,\Gamma \ell^{\frac{D}{2}} m^{\frac{D-1}{2}}}{\sigma_H} }^k\,;\;\;&\;\;\|\partial^{(k)}_\omega S_{m,j}^M(\omega)\|\leq 2\,(k!) \,\norbra{\frac{E\,\Gamma \ell^{\frac{D}{2}} m^{\frac{D-1}{2}}}{\sigma_H} }^k\,.
	\end{align}
	This concludes the proof of Lemma~\ref{lemma:clusterExp}.

\end{document}